\documentclass[10pt, a4paper]{article}

\usepackage{enumitem}
\usepackage{xcolor}
\usepackage{upgreek}

\usepackage{amsmath,amssymb,amsthm,mathtools}
\usepackage{tikz,graphicx}
\usepackage{tkz-euclide}
\usepackage{pict2e}
\usepackage{todonotes}
\theoremstyle{plain}

\newtheorem*{theorem*}{Theorem}
\newtheorem{theorem}{Theorem}[section] 
\newtheorem{lemma}[theorem]{Lemma}
\newtheorem{proposition}[theorem]{Proposition}

\newtheorem{remark}[theorem]{Remark}
\theoremstyle{definition}

\newtheorem{definition}[theorem]{Definition}

\usepackage{tikz}
\usepackage{pict2e}
\usetikzlibrary{arrows,calc,chains, positioning, shapes.geometric,shapes.symbols,decorations.markings,arrows.meta}
\usetikzlibrary{patterns,angles,quotes}
\usetikzlibrary{decorations.markings}
\tikzset{middlearrow/.style={
			decoration={markings,
				mark= at position 0.6 with {\arrow{#1}} ,
			},
			postaction={decorate}
		}
	}
\tikzset{->-/.style={decoration={
				markings,
				mark=at position #1 with {\arrow{latex}}},postaction={decorate}}}
	
	\tikzset{-<-/.style={decoration={
				markings,
				mark=at position #1 with {\arrowreversed{latex}}},postaction={decorate}}}

\newcommand{\ds}{\displaystyle}

\numberwithin{equation}{section}

\def\bigO{{\cal O}}

\tikzset{
	master/.style={
		execute at end picture={
			\coordinate (lower right) at (current bounding box.south east);
			\coordinate (upper left) at (current bounding box.north west);
		}
	},
	slave/.style={
		execute at end picture={
			\pgfresetboundingbox
			\path (upper left) rectangle (lower right);
		}
	}
}

\tikzset{middlearrow/.style={
		decoration={markings,
			mark= at position 0.6 with {\arrow{#1}} ,
		},
		postaction={decorate}
	}
}

\newcommand{\re}{\text{\upshape Re\,}}
\newcommand{\im}{\text{\upshape Im\,}}

\def\Xint#1{\mathchoice
{\XXint\displaystyle\textstyle{#1}}%
{\XXint\textstyle\scriptstyle{#1}}%
{\XXint\scriptstyle\scriptscriptstyle{#1}}%
{\XXint\scriptscriptstyle\scriptscriptstyle{#1}}%
\!\int}
\def\XXint#1#2#3{{\setbox0=\hbox{$#1{#2#3}{\int}$ }
\vcenter{\hbox{$#2#3$ }}\kern-.59\wd0}}

\def\dashint{\Xint-}

\tikzset{
    master/.style={
        execute at end picture={
            \coordinate (lower right) at (current bounding box.south east);
            \coordinate (upper left) at (current bounding box.north west);
        }
    },
    slave/.style={
        execute at end picture={
            \pgfresetboundingbox
            \path (upper left) rectangle (lower right);
        }
    }
}

\usepackage[colorlinks=true]{hyperref}
\hypersetup{urlcolor=blue, citecolor=blue, linkcolor=black}

\begin{document}

\title{Asymptotics of Muttalib-Borodin determinants with Fisher-Hartwig singularities}
\author{Christophe Charlier}

\maketitle

\begin{abstract}
Muttalib-Borodin determinants are generalizations of Hankel determinants and depend on a parameter $\theta>0$. In this paper, we obtain large $n$ asymptotics for $n \times n$ Muttalib-Borodin determinants whose weight possesses an arbitrary number of Fisher-Hartwig singularities. As a corollary, we obtain asymptotics for the expectation and variance of the real and imaginary parts of the logarithm of the underlying characteristic polynomial, several central limit theorems, and some global bulk rigidity upper bounds. Our results are valid for all $\theta > 0$.
\end{abstract}




\section{Introduction and statement of results}\label{section: introduction}

The main result of this paper is an asymptotic formula as $n \to + \infty$ for
\begin{align}\label{n fold integral}
D_{n}(w) & := \frac{1}{n!}\int_{a}^{b}\cdots \int_{a}^{b} \prod_{1 \leq j < k \leq n} (x_{k}-x_{j})(x_{k}^{\theta}-x_{j}^{\theta}) \prod_{j=1}^{n} w(x_{j})dx_{j} \nonumber \\
& = \det \bigg( \int_{a}^{b}x^{k+j \theta}w(x)dx \bigg)_{j,k=0}^{n-1},
\end{align}
with $0<a<b$, $\theta > 0$, and the weight $w$ is of the form
\begin{align}\label{weight}
w(x) = e^{W(x)}\omega(x),
\end{align}
where the function $W:[a,b]\to \mathbb{R}$ is analytic in a neighborhood of $[a,b]$,
\begin{align}
& \omega(x) = (x-a)^{\alpha_{0}} (b-x)^{\alpha_{m+1}} \prod_{j=1}^{m}\omega_{\alpha_{j}}(x)\omega_{\beta_{j}}(x), \qquad m \in \mathbb{N}=\{0,1,\ldots\}, \label{def of omega} \\
& \omega_{\alpha_{j}}(x) = |x-t_{j}|^{\alpha_{j}}, \qquad \omega_{\beta_{j}}(x) = \begin{cases}
e^{i\pi \beta_{j}}, & \mbox{if } x<t_{j}, \\
e^{-i\pi \beta_{j}}, & \mbox{if } x>t_{j},
\end{cases} \label{def of omega alpha and beta}
\end{align}
and $0<a<t_{1}<\ldots<t_{m}<b<+\infty$,
\begin{align}\label{conditions on the parameters}
\re \alpha_{0},\ldots,\re \alpha_{m+1}>-1, \quad \re \beta_{1},\ldots,\re \beta_{m} \in (-\tfrac{1}{4},\tfrac{1}{4}).
\end{align}
The parameters $\alpha_{j}$ and $\beta_{j}$ describe the root-type and jump-type singularities of $w$, respectively. In total, the weight $w$ has $m$ Fisher-Hartwig (FH) singularities in the interior of its support, and two root-type FH singularities at the edges $a$ and $b$. The condition $\re \alpha_{j}>-1$ ensures that $D_{n}(w)$ is well-defined. Since $\omega_{\beta_{j}+n_{0}}=(-1)^{n_{0}}\omega_{\beta_{j}}$ for any $n_{0}\in \mathbb{Z}$ and $\beta_{j}\in \mathbb{C}$, one can reduce the general case $\beta_{j}\in \mathbb{C}$ to $\re \beta_{j} \in (-\frac{1}{2},\frac{1}{2}]$ without loss of generality. The restriction $\re \beta_{j} \in (-\frac{1}{4},\frac{1}{4})$ in \eqref{conditions on the parameters} is due to some technicalities in our analysis (see \eqref{reason why we have need beta<1/4}). 

\medskip The determinant $D_{n}(w)$ arises naturally in the study of certain Muttalib-Borodin (MB) ensembles, and for this reason we call $D_{n}(w)$ a \textit{Muttalib-Borodin determinant}. Given a non-negative weight $\mathsf{w}$ with sufficient decay at $+\infty$, the associated MB ensemble of parameter $\theta>0$ is the joint probability density function 
\begin{align}\label{MB density}
\frac{1}{n!D_{n}(\mathsf{w})}\prod_{1\leq j<k\leq n}(x_{k}-x_{j})(x_{k}^{\theta}-x_{j}^{\theta})\prod_{j=1}^{n}\mathsf{w}(x_{j}), \qquad x_{1},\ldots,x_{n}\in [0,+\infty),
\end{align}
where $D_{n}(\mathsf{w})$ is the normalization constant. For $\alpha_{0},\alpha_{m+1}>-1$, the determinant $D_{n}(w)$ is for example of interest in the study of the random polynomial $\mathsf{p}_{n}(t)=\prod_{j=1}^{n}(t-x_{j})$, where $x_{1},\ldots,x_{n}$ are distributed according to the MB ensemble associated to the weight
\begin{align}\label{weight w0}
\mathsf{w}(x) = (x-a)^{\alpha_{0}}(b-x)^{\alpha_{m+1}}e^{W(x)}\chi_{(a,b)}(x), \quad \chi_{(a,b)}(x)= \begin{cases}
1, & x \in (a,b), \\
0, & \mbox{otherwise}.
\end{cases}
\end{align}
Indeed, as can be seen from \eqref{n fold integral}--\eqref{def of omega alpha and beta} and \eqref{MB density}, we have
\begin{align}\label{expectation of characteristic polynomial}
\mathbb{E} \bigg( \prod_{k=1}^{m} |\mathsf{p}_{n}(t_{k})|^{\alpha_{k}}e^{2i\beta_{k}\arg \mathsf{p}_{n}(t_{k})} \bigg) = \frac{D_{n}(w)}{D_{n}(\mathsf{w})}\prod_{k=1}^{m}e^{-i\pi n \beta_{k}},
\end{align}
where
\begin{align*}
\arg \mathsf{p}_{n}(t) = \sum_{j=1}^{n} \arg(t-x_{j}), \qquad \mbox{ with } \qquad \arg (t-x_{j}) = \begin{cases}
0, & \mbox{if } x_{j}<t, \\
-\pi, & \mbox{if } x_{j}>t.
\end{cases}
\end{align*}
Equivalently, \eqref{expectation of characteristic polynomial} can be rewritten as
\begin{align}\label{moment generating function in introduction}
\mathbb{E} \bigg( \prod_{k=1}^{m} |\mathsf{p}_{n}(t_{k})|^{\alpha_{k}} e^{2\pi i\beta_{k}N_{n}(t_{k})} \bigg) = \frac{D_{n}(w)}{D_{n}(\mathsf{w})}\prod_{k=1}^{m}e^{i\pi n \beta_{k}},
\end{align}
where $N_{n}(t) \in \{0,1,\ldots,n\}$ is the counting function of \eqref{MB density} and is given by
\begin{align*}
N_{n}(t) = \#\{x_{j}: x_{j}\leq t\}, \qquad t \in \mathbb{R}.
\end{align*}
In particular, formula \eqref{moment generating function in introduction} with $\alpha_{1}=\ldots=\alpha_{m}=0$ shows that the moment generating function of the MB ensemble \eqref{MB density} can be expressed as a ratio of two MB determinants. 

\medskip The densities \eqref{MB density} were introduced by Muttalib \cite{Muttalib} in the context of disordered conductors in the metallic regime. These models are also named after Borodin \cite{Borodin}, who studied, for the classical Laguerre and Jacobi weights, the limiting local microscopic behavior of the random points $x_{1},\ldots,x_{n}$ as $n \to +\infty$. The notable feature of MB ensembles is that neighboring points $x_{j},x_{k}$ repel each other as $\sim (x_{k}-x_{j})(x_{k}^{\theta}-x_{j}^{\theta})$, which differs, for $\theta \neq 1$, from the simpler and more standard situation $\sim (x_{k}-x_{j})^{2}$. In fact, MB ensembles fall within a special class of determinantal point processes known as biorthogonal ensembles, and a main difficulty in their asymptotic analysis for $\theta \neq 1$ is the lack of a simple Christoffel-Darboux formula for the underlying biorthogonal polynomials.\footnote{See \cite[Theorem 1.1]{ClaeysRomano} for a formula valid only for $\theta \in \mathbb{Q}$. For $\theta \notin \mathbb{Q}$, there is simply no Christoffel-Darboux formula available in the literature.} MB ensembles have attracted considerable attention over the years, partly due to their relation to eigenvalue distributions of random matrix models \cite{Cheliotis, ForWang, KuijSti}. MB ensembles also arise in the study of random plane partitions \cite{BeteaOccelli} and  the Dyson Brownian motion under a moving boundary \cite{GlDMS2019,GMS2021}. 


\medskip For $\theta=1$, MB determinants are Hankel determinants and the large $n$ asymptotics of $D_{n}(w) = D_{n}(e^{W}\omega \, \chi_{(a,b)})$ have been obtained by Deift, Its and Krasovsky \cite{DIK,DeiftItsKrasovsky}. In fact, asymptotics of Hankel determinants with FH singularities have been studied by many authors and are now understood even in the more complicated situation where the weight varies wildly with $n$; more precisely, for $\theta=1$ the large $n$ asymptotics of $D_{n}(e^{-nV}e^{W}\omega)$ are known up to and including the constant term, for any potential $V$ such that the points $x_{1},\ldots,x_{n}$ accumulate on a single interval as $n \to +\infty$ (the so-called ``one-cut regime"), see \cite{KMcLVAV, Krasovsky, Garoni, ItsKrasovsky, BerWebbWong, Charlier, CharlierGharakhloo}. Asymptotics of Hankel determinants with FH singularities have also been studied in various transition regimes of the parameters: see \cite{BCI2016, WXZ2018} for FH singularities approaching the edges, \cite{ClaeysFahs} for two merging root-type singularities, and \cite{ChDeano} for a large jump-type singularity. We also mention that the problem of finding asymptotics of large Toeplitz determinants with several FH singularities presents many similarities with the Hankel case and has also been widely studied, see e.g. \cite{FisherHartwig, Widom2, Basor, Basor2, BS1986, Ehrhardt, DIK, DeiftItsKrasovsky} for important early works. 

\medskip Very few results exist on MB determinants for general values of $\theta$.  It was noticed in \cite{Cheliotis, ForIps} that MB determinants associated to the classical Jacobi and Laguerre weights are Selberg integrals which can be evaluated explicitly, and the asymptotics of MB determinants without FH singularities have been studied in \cite{BGK2015}. To the best of our knowledge, for $\theta \neq 1$ no results are available in the literature on the large $n$ asymptotics of MB determinants whose weight has FH singularities in the interior of its support. The purpose of this paper is to take a first step toward the solution of this problem. 



\medskip We now introduce the necessary material to present our results. As is usually the case in the asymptotic analysis of $n$-fold integrals, see e.g. \cite[Section 6.1]{Deift}, an important role in the asymptotics of $D_{n}(w)$ is played by an equilibrium measure. As can be seen from \eqref{n fold integral}, the main contribution in the large $n$ asymptotics of $D_{n}(w)$ comes from the $n$-tuples $(x_{1},\ldots,x_{n})$ which minimize
\begin{align*}
\sum_{1\leq j < k \leq n} \log |x_{k}-x_{j}|^{-1} + \sum_{1\leq j < k \leq n} \log |x_{k}^{\theta}-x_{j}^{\theta}|^{-1}.
\end{align*}
Hence, we are led to consider the problem of finding the probability measure $\mu_{\theta}$ minimizing
\begin{align}\label{equilibrium problem}
\mu \mapsto \int_{a}^{b}\int_{a}^{b} \log\frac{1}{|x-y|}d\mu(x)d\mu(y) + \int_{a}^{b}\int_{a}^{b} \log\frac{1}{|x^{\theta}-y^{\theta}|}d\mu(x)d\mu(y) 
\end{align}
among all Borel probability measures $\mu$ on $[a,b]$. This measure $\mu_{\theta}$ is called the equilibrium measure; in our case it is absolutely continuous with respect to the Lebesgue measure, supported on the whole interval $[a,b]$, and if $\mu$ is a probability measure satisfying the following Euler-Lagrange equality 
\begin{align}\label{EL equality}
& \int_{a}^{b} \log |x-y|d\mu(y) + \int_{a}^{b} \log |x^{\theta}-y^{\theta}|d\mu(y) = - \ell, & & \mbox{for } x \in [a,b],
\end{align}
where $\ell \in \mathbb{R}$ is a constant, then $\mu = \mu_{\theta}$ \cite{SaTo,ClaeysRomano}. Similar equilibrium problems related to MB ensembles have been studied in detail by Claeys and Romano in \cite{ClaeysRomano} (see also \cite[Theorem 1]{ClaeysWang}), but in our case the equilibrium measure has two hard edges and this is not covered by \cite{ClaeysRomano}. Nevertheless, as in \cite{ClaeysRomano}, the following function $J$ plays an important role in the construction of $\mu_{\theta}$:
\begin{align}\label{def of the map J}
J(s) = J(s;c_{0},c_{1}) = (c_{1}s+c_{0}) \bigg( \frac{s+1}{s} \bigg)^{\frac{1}{\theta}}, \qquad c_{0}>c_{1}>0,
\end{align}
where the branch cut lies on $[-1,0]$ and is such that $J(s) = c_{1}s(1+\bigO(s^{-1}))$ as $s \to \infty$. It is easy to check that $J'(s)=0$ if and only if $s \in \{s_{a},s_{b}\}$, where 
\begin{align}\label{def of sa sb}
s_{a} = \frac{1-\theta}{2\theta}-\frac{1}{2\theta}\sqrt{4 \theta \frac{c_{0}}{c_{1}} + (1-\theta)^{2}}, \quad s_{b} = \frac{1-\theta}{2\theta}+\frac{1}{2\theta}\sqrt{4 \theta\frac{c_{0}}{c_{1}} + (1-\theta)^{2}}.
\end{align}
Since $c_{0}>c_{1}>0$, these points always satisfy $s_{a}<-1$ and $\frac{1}{\theta}<s_{b}$. It is also easy to verify (see Lemma \ref{lemma:c0c1} for the proof) that for any $0<a<b<+\infty$, there exists a unique tuple $(c_{0},c_{1})$ which satisfies
\begin{align}\label{system that defines c0 and c1}
J(s_{a})=a, \qquad J(s_{b})=b, \qquad c_{0}>c_{1}>0.
\end{align}
The following proposition was proved in \cite{ClaeysRomano} and summarizes some important properties of $J$.
\begin{proposition}[Claeys--Romano \cite{ClaeysRomano}]\label{prop:ClaeysRomano}
Let $\theta \geq 1$ and $c_{0}>c_{1}>0$ be such that \eqref{system that defines c0 and c1} holds. There are two complex conjugate curves $\gamma_{1}$ and $\gamma_{2}$ starting at $s_{a}$ and ending at $s_{b}$ in the upper and lower half plane respectively which are mapped to the interval $[a,b]$ through $J$. Let $\gamma$ be the counterclockwise oriented closed curve consisting of the union of $\gamma_{1}$ and $\gamma_{2}$, enclosing a region $D$. The maps 
\begin{align}\label{two bijection}
J:\mathbb{C}\setminus \overline{D} \to \mathbb{C}\setminus[a,b], \qquad J:D\setminus[-1,0]\to \mathbb{H}_{\theta}\setminus [a,b]
\end{align}
are bijections, where $\mathbb{H}_{\theta} := \{z \in \mathbb{C}\setminus\{0\}: -\frac{\pi}{\theta}<\arg z < \frac{\pi}{\theta}\}$. See also Figure \ref{fig: the mapping J}.
\end{proposition}
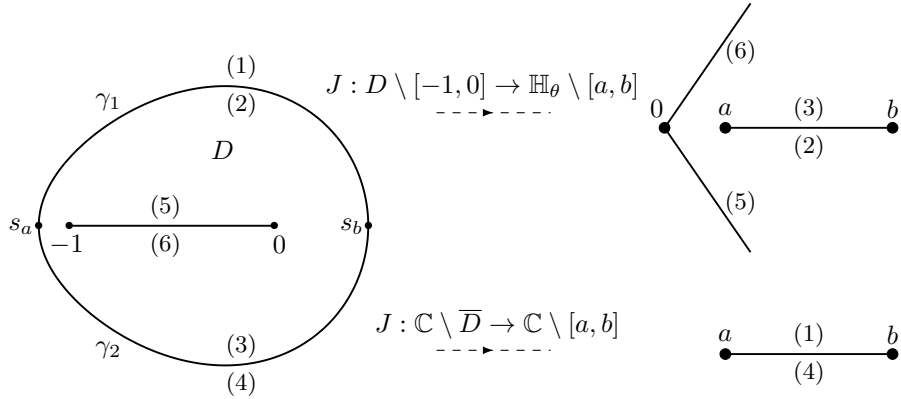
\begin{figure}
\begin{center}
\begin{tikzpicture}[master]
\node at (0,0) {\includegraphics[scale=0.3]{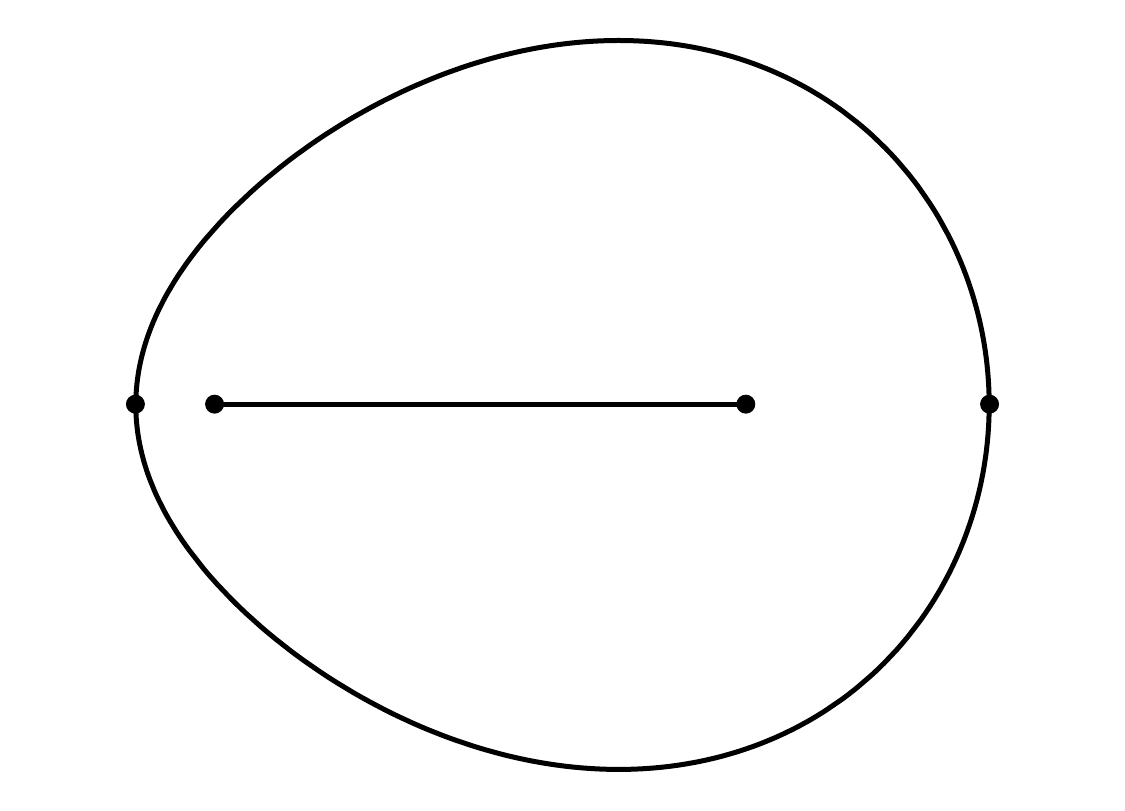}};
\node at (-2.4,0) {$s_{a}$};
\node at (1.95,0) {$s_{b}$};
\node at (-1.8,-0.25) {$-1$};
\node at (1,-0.25) {$0$};
\node at (0.25,1) {$D$};
\node at (0.5,2.1) {\small $(1)$};
\node at (0.5,1.6) {\small $(2)$};
\node at (0.5,-1.6) {\small $(3)$};
\node at (0.5,-2.1) {\small $(4)$};
\node at (-0.5,0.25) {\small $(5)$};
\node at (-0.5,-0.25) {\small $(6)$};
\node at (-1.25,1.65) {$\gamma_{1}$};
\node at (-1.25,-1.65) {$\gamma_{2}$};
\end{tikzpicture}
\begin{tikzpicture}[slave]
\draw[line width=0.25 mm] (0,1.3)--($(0,1.3)+(180/3.24:2)$);
\draw[line width=0.25 mm] (0,1.3)--($(0,1.3)+(-180/3.24:2)$);
\draw[line width=0.25 mm] (0.8,1.3)--(3,1.3);
\draw[fill] (0,1.3) circle (0.7mm);
\draw[fill] (0.8,1.3) circle (0.7mm);
\draw[fill] (3,1.3) circle (0.7mm);
\node at (-0.1,1.55) {$0$};
\node at (0.8,1.55) {$a$};
\node at (3,1.55) {$b$};
\node at (1.9,1.55) {\small $(3)$};
\node at (1.9,1.05) {\small $(2)$};
\node at (1,2.3) {\small $(6)$};
\node at (1,0.3) {\small $(5)$};

\draw[dashed, ->-=0.5] (-3,1.5)--(-1.5,1.5);
\node at (-2.4,1.85) {$J:D\setminus[-1,0] \to \mathbb{H}_{\theta}\setminus [a,b]$};
\draw[dashed, ->-=0.5] (-3,-1.65)--(-1.5,-1.65);
\node at (-2.2,-1.3) {$J:\mathbb{C}\setminus \overline{D} \to \mathbb{C}\setminus [a,b]$};

\draw[line width=0.25 mm] (0.8,-1.7)--(3,-1.7);
\draw[fill] (0.8,-1.7) circle (0.7mm);
\draw[fill] (3,-1.7) circle (0.7mm);
\node at (0.8,-1.45) {$a$};
\node at (3,-1.45) {$b$};
\node at (1.9,-1.45) {\small $(1)$};
\node at (1.9,-1.95) {\small $(4)$};
\end{tikzpicture}
\end{center}
\caption{\label{fig: the mapping J}The mapping $J$.}
\end{figure}
The case $\theta < 1$ was not considered in \cite{ClaeysRomano} but only requires minor modifications. The extension of Proposition \ref{prop:ClaeysRomano} to all values of $\theta  >0$ is given in Proposition \ref{prop:ClaeysRomano generalization} below. In particular, we show that Proposition \ref{prop:ClaeysRomano} is still valid for $\theta<1$, except that $J:D\setminus[-1,0]\to \mathbb{H}_{\theta}\setminus [a,b]$ is no longer a bijection. For any $\theta > 0$, let 
\begin{align*}
I_{1}: \mathbb{C}\setminus[a,b] \to \mathbb{C}\setminus \overline{D}
\end{align*}
denote the inverse of $J:\mathbb{C}\setminus \overline{D} \to \mathbb{C}\setminus[a,b]$, and let $I_{1,\pm}(x) := \lim_{\epsilon\to 0_{+}} I_{1}(x \pm i\epsilon)$, $x \in (a,b)$. As shown in Figure \ref{fig: the mapping J}, we have
\begin{align*}
I_{1,+}(x)\in \gamma_{1}, \qquad I_{1,-}(x) \in \gamma_{2}, \qquad x \in (a,b).
\end{align*}

\begin{proposition}\label{prop:density}
Let $\theta > 0$, $b>a>0$, and let $(c_{0},c_{1})$ be the unique solution to
\begin{align}\label{uniqueness of c0 and c1}
J(s_{a})=a, \qquad J(s_{b})=b, \qquad c_{0}>c_{1}>0.
\end{align} The unique equilibrium measure $\mu_{\theta}$ satisfying \eqref{EL equality} is given by $d\mu_{\theta}(x) = \rho(x)dx$, where
\begin{align}\label{explicit expression for rho intro}
\rho(x) = -\frac{1}{\pi} \im \bigg( \frac{I_{1,+}'(x)}{I_{1,+}(x)} \bigg) = - \frac{1}{\pi} \frac{d}{dx}\arg I_{1,+}(x), \qquad x \in (a,b),
\end{align}
with $\arg I_{1,+}(x) \in (0,\pi)$ for all $x \in (a,b)$. 
\end{proposition}
\begin{remark}
It can be readily verified using \eqref{def of the map J} and \eqref{explicit expression for rho intro} that $\rho$ blows up like an inverse square root near $a$ and $b$. Indeed, since
\begin{align}
& J(s) = b+\frac{J''(s_{b})}{2}(s-s_{b})^{2}+\bigO\big((s-s_{b})^{3}\big), & & \mbox{as } s \to s_{b}, \label{expansion of J near the endpoints b} \\
& J(s) = a+\frac{J''(s_{a})}{2}(s-s_{a})^{2}+\bigO\big((s-s_{a})^{3}\big), & & \mbox{as } s \to s_{a}, \label{expansion of J near the endpoints a}
\end{align}
with $J''(s_{b})>0$, $J''(s_{a})<0$, we obtain
\begin{align}
& \rho(x) = \frac{1}{\sqrt{2}\pi s_{b} \sqrt{J''(s_{b})}}\frac{1}{\sqrt{b-x}} + \bigO(1), & & \mbox{as } x \to b, \; x < b, \label{asymptotics of rho near b} \\
& \rho(x) = \frac{1}{\sqrt{2}\pi |s_{a}| \sqrt{|J''(s_{a})|}}\frac{1}{\sqrt{x-a}} + \bigO(1), & & \mbox{as } x \to a, \; x > a. \label{asymptotics of rho near a}
\end{align}
\end{remark}
The following theorem is our main result.

\begin{theorem}\label{theorem J}
Let $\theta > 0$, $m \in \mathbb{N}$ and $a,t_{1},\ldots,t_{m},b \in \mathbb{R}$, $\alpha_{0},\ldots,\alpha_{m+1},\beta_{1},\ldots$, $\beta_{m} \in \mathbb{C}$ be such that $0<a<t_{1}<\ldots<t_{m}<b$,
\begin{equation*}
\re \alpha_{0},\ldots,\re \alpha_{m+1}>-1, \qquad \re \beta_{1},\ldots,\re \beta_{m} \in (-\tfrac{1}{4},\tfrac{1}{4}),
\end{equation*}
and let $W: [a,b]\to\mathbb{R}$ be analytic. Let $(c_{0},c_{1})$ be the unique solution to \eqref{uniqueness of c0 and c1}, and let $\rho$ be as in \eqref{explicit expression for rho intro}.  As $n \to +\infty$, we have
\begin{equation}\label{asymp thm Jn}
D_{n}(w) = \exp\left(C_{1} n^{2} + C_{2} n + C_{3} \log n + C_{4} + \bigO \Big( \frac{1}{n^{1-4\beta_{\max}}} \Big)\right),
\end{equation}
with $\beta_{\max} = \max \{ |\re \beta_{1}|,\ldots,|\re \beta_{m}| \}$,
\begin{align}
& C_{1} = - \frac{\ell}{2} = \frac{1}{2}\log c_{1} + \frac{\theta}{2} \log c_{0}, \label{C1 thm} \\
& C_{2} = \frac{1-\theta}{2}\log c_{0} - \frac{1}{2}\log \theta + \log(2\pi) + \int_{a}^{b} W(x) \rho(x) dx \nonumber \\
& \hspace{0.8cm} + \sum_{j=0}^{m+1}\alpha_{j}\int_{a}^{b}\log|t_{j}-x|\rho(x)dx  + \sum_{j=1}^{m} \pi i \beta_{j}\bigg( 1-2 \int_{t_{j}}^{b} \rho(x)dx \bigg), \label{C2 thm} \\
& C_{3} = - \frac{1}{4} + \frac{\alpha_{0}^{2}+\alpha_{m+1}^{2}}{2} + \sum_{j=1}^{m} \bigg( \frac{\alpha_{j}^{2}}{4} - \beta_{j}^{2} \bigg), \label{C3 thm}
\end{align}
$t_{0}:=a$, $t_{m+1}:=b$, $C_{4}$ is independent of $n$, $(c_{0},c_{1})$ is the unique solution to \eqref{uniqueness of c0 and c1}, the density $\rho$ is given by \eqref{explicit expression for rho intro}, and $\ell$ is the associated Euler-Lagrange constant defined in \eqref{EL equality}.  The constant $C_{2}$ can also be rewritten using the relations
\begin{align}\label{identity for rho and real and im parts}
\int_{a}^{t}\rho(x)dx = \frac{\pi-\arg I_{1,+}(t)}{\pi}, \quad \int_{a}^{b}\log|t-x|\rho(x)dx = \log (c_{1}|I_{1,+}(t)|), \quad t \in (a,b).
\end{align}
Furthermore, the error term in \eqref{asymp thm Jn} is uniform for all $\alpha_{k}$ in compact subsets of $\{ z \in \mathbb{C}: \re z >-1 \}$, for all $\beta_{k}$ in compact subsets of  $\{ z \in \mathbb{C}: \re z \in \big( \frac{-1}{4},\frac{1}{4} \big) \}$, for $\theta$ in compact subsets of $(0,+\infty)$ and uniform in $t_{1},\ldots,t_{m}$, as long as there exists $\delta > 0$ independent of $n$ such that
\begin{equation}\label{assumption on tj delta}
\min_{1\leq j\neq k \leq m}\{ |t_{j}-t_{k}|,|t_{j}-b|,|t_{j}-a|\} \geq \delta.
\end{equation}
\end{theorem}
\begin{remark}
For $\theta=1$, $\gamma$ is a circle and
\begin{align*}
& \ell = -2 \log \frac{b-a}{4}, \qquad \rho(x) = \frac{1}{\pi \sqrt{(x-a)(b-x)}}.
\end{align*}
Substituting these expressions in \eqref{C1 thm}--\eqref{C3 thm}, we obtain
\begin{align*}
& C_{1}\big|_{\theta=1} = \log\frac{b-a}{4}, \qquad C_{3}\big|_{\theta=1} = - \frac{1}{4} + \frac{\alpha_{0}^{2}+\alpha_{m+1}^{2}}{2} + \sum_{j=1}^{m} \bigg( \frac{\alpha_{j}^{2}}{4} - \beta_{j}^{2} \bigg) \\
& C_{2}\big|_{\theta=1} = \log(2\pi)    + \hspace{-0.1cm} \int_{a}^{b} \hspace{-0.1cm} W(x) \rho(x) dx + \log\frac{b-a}{4} \hspace{-0.05cm} \sum_{j=0}^{m+1}\alpha_{j} \hspace{-0.05cm} + \hspace{-0.05cm} \sum_{j=1}^{m} \pi i \beta_{j}\bigg( 1-2 \int_{t_{j}}^{b} \rho(x)dx \bigg). 
\end{align*}
These values for $C_{1}|_{\theta=1}$, $C_{2}|_{\theta=1}$, $C_{3}|_{\theta=1}$ are consistent with \cite{DIK}. The constant $C_{4}|_{\theta=1}$ was also obtained in \cite{DIK} (see also \cite[Theorem 1.3 with $V=0$]{CharlierGharakhloo}) and contains Barnes' $G$-function. It would be interesting to obtain an explicit expression for $C_{4}$ valid for all values of $\theta>0$, but this problem seems difficult, see also Remark \ref{remark: constant} below.
\end{remark}

Many statistical properties of MB ensembles have been widely studied over the years: see \cite{ESS2011, ClaeysRomano, ForLiu, Kuijlaars, BLTW2017} for equilibrium problems, \cite{Borodin, KuijMolag, Molag, ZhangUniversality, Zhang, WangZhang} for results on the limiting correlation kernel, \cite{LambertBOE} (see also \cite{BreuerDuits}) for central limit theorems for smooth test functions in the Laguerre and Jacobi MB ensembles when $\frac{1}{\theta} \in \mathbb{N}$, and \cite{CGS2019,CLM2021} for large gap asymptotics. As can be seen from \eqref{expectation of characteristic polynomial}--\eqref{moment generating function in introduction}, the determinant $D_{n}(w)$ is the joint moment generating function of the random variables
\begin{align*}
\re \log \mathsf{p}_{n}(t_{1}),\ldots,\re \log \mathsf{p}_{n}(t_{m}), \im \log \mathsf{p}_{n}(t_{1}), \ldots,\im \log \mathsf{p}_{n}(t_{m}),
\end{align*}
and therefore Theorem \ref{theorem J} contains significant information about \eqref{MB density}. In particular, we can deduce from it new asymptotic formulas for the expectation and variance of $\im \log \mathsf{p}_{n}(t)$ (or equivalently $N_{n}(t)$) and $\re \log |\mathsf{p}_{n}(t)|$, several central limit theorems for test functions with poor regularity (such as discontinuities), and some global bulk rigidity upper bounds. 
\begin{theorem}\label{thm:rigidity}
Let $\theta > 0$, $m \in \mathbb{N}$ and $t_{1},\ldots,t_{m}$ be such that $a<t_{1}<\ldots<t_{m}<b$. Let $x_{1}, x_{2}, \ldots , x_{n}$ be distributed according to the MB ensemble \eqref{MB density} where $\mathsf{w}$ is given by \eqref{weight w0}, and define $\mathsf{p}_{n}(t)$, $N_{n}(t)$ by
\begin{align*}
\mathsf{p}_{n}(t)=\prod_{j=1}^{n}(t-x_{j}), \qquad N_{n}(t) = \#\{x_{j}:x_{j}\leq t\} \in \{0,1,2,\ldots,n\}, \qquad t \in \mathbb{R}.
\end{align*}
Let $\xi_{1}\leq \xi_{2} \leq \ldots \leq \xi_{n}$ denote the ordered points, \vspace{0.1cm}
\begin{align*}
\xi_{1}=\min\{x_{1},\ldots,x_{n}\}, \quad \xi_{j} = \inf_{t\in [a,b]}\{t:N_{n}(t)=j\}, \quad j=1,\ldots,n,
\end{align*} \\[-0.2cm]
and let $\kappa_{k}$ be the classical location of the $k$-th smallest point $\xi_{k}$, 
\begin{align}\label{def of kappa k}
\int_{a}^{\kappa_{k}}\rho(x)dx = \frac{k}{n}, \qquad k=1,\ldots,n.
\end{align}
\begin{itemize}
\item[(a)] Let $t \in (a,b)$ be fixed. As $n \to \infty$, we have
\begin{align}
& \mathbb{E}(N_{n}(t)) = \int_{a}^{t}\rho(x)dx \; n+\bigO(1) = \frac{\pi-\arg I_{1,+}(t)}{\pi}n + \bigO(1), \label{asymp expectation} \\
& \mathbb{E}(\log |\mathsf{p}_{n}(t)|) = \int_{a}^{b}\log|t-x|\rho(x)dx \; n+\bigO(1), \label{asymp expectation ln |pn|} \\
& \mathrm{Var}(N_{n}(t)) = \frac{1}{2\pi^{2}}\log n + \bigO(1), \quad \mathrm{Var}(\log |\mathsf{p}_{n}(t)|) = \frac{1}{2}\log n + \bigO(1). \label{asymp variance}
\end{align} 
\item[(b)] Consider the random variables $\mathcal{M}_{n}(t_{j})$, $\mathcal{N}_{n}(t_{j})$ defined for $j=1,\ldots,m$ by
\begin{align}
& \mathcal{M}_{n}(t_{j}) = \sqrt{2} \frac{\log |\mathsf{p}_{n}(t_{j})|-n\int_{a}^{b}\log|t_{j}-x|\rho(x)dx}{\sqrt{\log n}}, \label{def of Mn and Nn 1} \\
& \mathcal{N}_{n}(t_{j}) = \sqrt{2}\pi \frac{N_{n}(t_{j})-n\int_{a}^{t_{j}}\rho(x)dx}{\sqrt{\log n}}. \label{def of Mn and Nn 2}
\end{align}
As $n \to +\infty$, we have the convergence in distribution
\begin{align}\label{convergence in distribution 1}
\big( \mathcal{M}_{n}(t_{1}),\ldots,\mathcal{M}_{n}(t_{m}),\mathcal{N}_{n}(t_{1}),\ldots,\mathcal{N}_{n}(t_{m})\big) \quad \overset{d}{\longrightarrow} \quad \mathsf{N}(\vec{0},I_{2m}),
\end{align}
where $I_{2m}$ is the $2m \times 2m$ identity matrix, and $\mathsf{N}(\vec{0},I_{2m})$ is a multivariate normal random variable of mean $\vec{0}=(0,\ldots,0)$ and covariance matrix $I_{2m}$.
\item[(c)] Let $k_{j}=[n \int_{a}^{t_{j}}\rho(x)dx]$, $j=1,\ldots,m$, where $[x]:= \lfloor x + \frac{1}{2}\rfloor$ is the closest integer to $x$. Consider the random variables $Z_{n}(t_{j})$ defined by
\begin{align}\label{def of Zn}
Z_{n}(t_{j}) = \sqrt{2}\pi \frac{n\rho(\kappa_{k_{j}})}{\sqrt{\log n}}(\xi_{k_{j}}-\kappa_{k_{j}}), \qquad j=1,\ldots,m.
\end{align}
As $n \to +\infty$, we have
\begin{align}\label{convergence in distribution 2}
\big( Z_{n}(t_{1}),Z_{n}(t_{2}),\ldots,Z_{n}(t_{m})\big) \quad \overset{d}{\longrightarrow} \quad \mathsf{N}(\vec{0},I_{m}).
\end{align}
\item[(d)] For all small enough $\delta >0$ and $\epsilon >0$, there exist $c>0$ and $n_{0}>0$ such that
\begin{align}
& \mathbb P\left(\sup_{a+\delta \leq x \leq b-\delta}\bigg|N_{n}(x)- n\int_{a}^{x}\rho(x)dx  \bigg|\leq \frac{\sqrt{1+\epsilon}}{\pi}\log n \right) \geq 1-cn^{-\epsilon}, \label{probabilistic upper bound 1} \\
& \mathbb{P}\bigg( \max_{\delta n \leq k \leq (1-\delta)n}  \rho(\kappa_{k})|\xi_{k}-\kappa_{k}| \leq \frac{\sqrt{1+\epsilon}}{\pi} \frac{\log n}{n} \bigg) \geq 1-cn^{-\epsilon}, \label{probabilistic upper bound 2}
\end{align}
for all $n \geq n_{0}$.
\end{itemize}
\end{theorem}
\begin{proof}
See Section \ref{Section: rigidity}.
\end{proof}
\begin{remark}
For $\theta=1$, the terms of order $1$ in \eqref{asymp expectation}--\eqref{convergence in distribution 1} are also known and can be obtained using the results of \cite{DIK}. The generalization of these formulas for general external potential (in the one-cut regime), but again for $\theta=1$, can be obtained using \cite{Charlier, CharlierGharakhloo}. We point out that analogous asymptotic formulas for the expectation and variance of the counting function of several universal point processes are also available in the literature, see e.g. \cite{SoshnikovSineAiryBessel, ChSine, ChCl3, ChBessel, DXZ2020 thinning, ChPearcey} for the sine, Airy, Bessel and Pearcey point processes.

The results \eqref{convergence in distribution 1} and \eqref{convergence in distribution 2} are central limit theorems (CLTs) for test functions with discontinuities and log singularities. For $\theta=1$ but general potential, similar CLTs can also be derived from the results of \cite{Charlier, CharlierGharakhloo}. Also, in the recent work \cite{BMP2021}, the authors obtained a comparable CLT for $\beta$-ensembles with a general potential (in the case where the equilibrium measure has two soft edges).

The probabilistic upper bounds \eqref{probabilistic upper bound 1}--\eqref{probabilistic upper bound 2} show that the maximum fluctuations of $N_{n}$, and of the random points $\xi_{1},\ldots,\xi_{n}$, are of order $\frac{\log n}{n}$ with overwhelming probability. In comparison, \eqref{def of Zn} shows that the individual fluctuations are of order $\frac{\smash{\sqrt{\log n}}}{n}$. Both \eqref{probabilistic upper bound 1} and \eqref{probabilistic upper bound 2} are statements concerning the bulk of the MB ensemble \eqref{MB density}--\eqref{weight w0} and can be compared with other global bulk rigidity estimates such as \cite{ErdosYauYin, ArguinBeliusBourgade, ChhaibiMadauleNajnudel, HolcombPaquette, PaquetteZeitouni, LambertCircular, CFLW, CGMY2020, BMP2021}. We expect the upper bounds \eqref{probabilistic upper bound 1}--\eqref{probabilistic upper bound 2} to be sharp (including the constants $\frac{1}{\pi}$), but Theorem \ref{theorem J} alone is not sufficient to prove the complementary lower bound. 

Also, Theorem \ref{theorem J} does not allow to obtain global rigidity estimates near the hard edges $a$ and $b$, and we refer to \cite{ChCl4} for results in this direction. 
\end{remark}

Let us now explain our strategy to prove Theorem \ref{theorem J}. As already mentioned, MB ensembles are biorthogonal ensembles \cite{Borodin}. Consider the families of polynomials $\{p_{j}\}_{j\geq 0}$ and $\{q_{j}\}_{j\geq 0}$ such that $p_{j}(x) = \kappa_{j}x^{j}+...$ and $q_{j}(x) = \kappa_{j}x^{j}+...$ are degree $j$ polynomials defined by the biorthogonal system
\begin{align}
& \int_{a}^{b} p_{k}(x) x^{j \theta}w(x)dx = \kappa_{k}^{-1}\delta_{k,j}, \qquad k=0,1,... \qquad j=0,1,2,...,k, \label{p ortho in bioortho} \\
& \int_{a}^{b} x^{k} q_{j}(x^{\theta})w(x)dx = \kappa_{k}^{-1}\delta_{k,j}, \qquad j=0,1,... \qquad k=0,1,2,...,j. \label{q ortho in bioortho}
\end{align}
These polynomials are always unique (up to multiplicative factors of $-1$), and by \cite[Proposition 2.1 (ii)]{ClaeysRomano} they satisfy 
\begin{align}\label{def of kappa k ^2}
\kappa_{k}^{2} = \frac{D_{k}(w)}{D_{k+1}(w)}, \qquad k =0,1,\ldots, \quad \mbox{where} \quad D_{0}(w):=1.
\end{align}
Let $M \in \mathbb{N}$ be fixed. Assuming that $p_{M},\ldots,p_{n-1}$ exist, we obtain the formula
\begin{align}\label{Dn in terms of the product}
D_{n}(w) = D_{M}(w) \prod_{k=M}^{n-1} \kappa_{k}^{-2}.
\end{align}
When the weight $w$ is positive, which is the case if
\begin{align*}
\alpha_{0},\ldots,\alpha_{m+1}\in \mathbb{R} \qquad \mbox{ and } \qquad \beta_{1},\ldots,\beta_{m}\in i \mathbb{R},
\end{align*}
the existence of $p_{j}$ and $q_{j}$ are guaranteed for all $j$, see \cite[Section 2]{ClaeysRomano}. This is not the case for general values of the parameters $\alpha_{j}$ and $\beta_{j}$, but it will follow from our analysis that all polynomials $p_{M},\ldots,p_{n-1}$ exist, provided that $M$ is chosen large enough. Our proof proceeds by first establishing precise asymptotics for $\kappa_{k}$ as $k \to +\infty$, which are then substituted in \eqref{Dn in terms of the product} to produce the asymptotic formulas \eqref{asymp thm Jn}--\eqref{C3 thm}. Note that, since the formula \eqref{Dn in terms of the product} also involves the value of $D_{M}(w)$ for some large but fixed $M$, our method does not give any hope to obtain the multiplicative constant $C_{4}$ of Theorem \ref{theorem J} (for more on that, see Remark \ref{remark: constant} below).

\medskip To obtain the large $n$ asymptotics of $\kappa_{n}$, we use the Riemann-Hilbert (RH) approach of \cite{ClaeysRomano}, and a generalization of the Deift--Zhou \cite{DeiftZhou} steepest descent method developed in \cite{ClaeysWang} by Claeys and Wang. More precisely, in \cite{ClaeysRomano} the authors have formulated a RH problem (for $\theta \geq 1$), whose solution is denoted $Y$, which uniquely characterizes $\kappa_{n}^{-1}p_{n}$ as well the following $\theta$-deformation of its Cauchy transform
\begin{align}\label{def of Cpj}
\frac{1}{\kappa_{n}} Cp_{n}(z) := \frac{1}{2\pi i \kappa_{n}}\int_{a}^{b} \frac{p_{n}(x)}{x^{\theta}-z^{\theta}}w(x)dx, \qquad z \in \mathbb{H}_{\theta}\setminus [a,b].
\end{align}
The RH problem for $Y$ from \cite{ClaeysRomano} is non-standard in the sense that it is of size $1\times 2$ and the different entries of the solution live on different domains. In the asymptotic analysis of this RH problem, several steps of the classical Deift--Zhou steepest descent method do not work or need to be substantially modified. In \cite{ClaeysWang}, Claeys and Wang developed a generalization of the Deift--Zhou steepest descent method to handle this type of RH problems, but so far their method has not been used to obtain asymptotic results for the biorthogonal polynomials \eqref{p ortho in bioortho}--\eqref{q ortho in bioortho}. The main technical contribution of the present paper is precisely the successful implementation of the method of \cite{ClaeysWang} on the RH problem for $Y$ from \cite{ClaeysRomano}.\footnote{Simultaneously and independently to this work, Wang and Zhang in \cite{WangZhang} also performed an asymptotic analysis of $Y$. Their situation is different from ours: they consider the case $a=0$, $\theta$ integer, and no FH singularities.} As in \cite{ClaeysWang}, in the small norm analysis the mapping $J$ plays an important role and allows to transform the $1\times 2$ RH problem to a scalar RH problem with non-local boundary conditions (a so-called \textit{shifted RH problem}). The methods of \cite{ClaeysRomano} rely on the fact that for $\theta \geq 1$, the principal root $z \mapsto z^{\theta}$ is a bijection from $\mathbb{H}_{\theta}$ to $\mathbb{C}\setminus (-\infty,0]$. The treatment of the case $\theta <1$ involves a natural Riemann surface and only requires minor modifications of \cite{ClaeysRomano}.

\medskip We mention that another RH approach to the study of MB ensembles has been developed by Kuijlaars and Molag in \cite{KuijMolag, Molag}. Their approach has the advantage to be more structured (for example, the solution of their RH problem has unit determinant), but it only allows values of $\theta$ such that $\frac{1}{\theta}\in \{1,2,3,\ldots\}$.

\begin{remark}\label{remark: constant}
An explicit expression for $C_{4}$ in (1.22) would allow to obtain more precise asymptotics for the mean and variance of the counting function in (1.29)--(1.31), as well as for the moment generating function (1.9), and is therefore of interest. The method used in \cite{DIK} to evaluate $C_{4}|_{\theta=1}$ relies on a Christoffel-Darboux formula and on the fact that $\mathrm{D}:=D_{n}(w)|_{\theta=1,\alpha_{1}=\ldots=\alpha_{m}=\beta_{1}=\ldots=\beta_{m}=0, W\equiv 0}$ reduces to a Selberg integral. The Christoffel-Darboux formula is essential to obtain convenient identities for $\partial_{\alpha_{j}}\log D_{n}(w)$, $\partial_{\beta_{j}}\log D_{n}(w)$, and the fact that $\mathrm{D}$ is explicit is used to determine the constant of integration. For MB ensembles, the only Christoffel-Darboux formulas that are available are valid for $\theta \in \mathbb{Q}$, see \cite[Theorem 1.1]{ClaeysRomano}. Since the asymptotic formula \eqref{asymp thm Jn} is already proved for all values of $\theta$, there is still hope that the evaluation of $C_{4}$ for all $\theta \in \mathbb{Q}$ will allow to determine $C_{4}$ for all values of $\theta$ by a continuity argument. However, even for $\theta \in \mathbb{Q}$, the evaluation of $C_{4}$ seems to be a difficult problem. Indeed, for $\theta \neq 1$, the only Selberg integral which we are aware of and that could be used is $D_{n}(w)|_{a=0,\alpha_{1}=\ldots=\alpha_{m}=\beta_{1}=\ldots=\beta_{m}=0,W\equiv 0}$, see \cite[eq (27)]{ForIps}. In particular, with this method one would need uniform asymptotics for $Y$ as $n \to + \infty$ and simultaneously $a \to 0$. For $a=0$, one expects from \cite{ClaeysRomano} that the density of the equilibrium measure blows up like $\sim x^{\smash{-\frac{1}{1+\theta}}}$ as $x \to 0$ which, in view of \eqref{asymptotics of rho near a}, indicates that a critical transition takes place as $n \to + \infty$, $a\to 0$. 
\end{remark}

\begin{remark}
We emphasize that only the case $a>0$ is considered in this work. The case $a=0$ is more complicated, because it requires a delicate local analysis around $0$ which so far has only been solved for particular values of $\theta$: see \cite{KuijMolag} for $\theta=\frac{1}{2}$ and \cite{Molag} when $1/\theta$ is an integer. We also mention the work \cite{WangZhang}, which was done simultaneously and independently to this work, where this local analysis was solved for integer values of $\theta$. Solving this local analysis for general values of $\theta>0$ remains an outstanding problem, and is the reason as to why we restrict ourselves to $a>0$. 
\end{remark}

\paragraph{Outline.} Proposition \ref{prop:density} is proved in Section \ref{section: equilibrium problem}.  In Section \ref{section: steepest descent analysis}, we formulate the RH problem for $Y$ from \cite{ClaeysRomano} which uniquely characterizes $p_{n}$ and $Cp_{n}$. In Sections \ref{section: steepest descent analysis}--\ref{section:small norm shifted RHP}, we perform an asymptotic analysis of the RH problem for $Y$ following the method of \cite{ClaeysWang}. In Section \ref{section: steepest descent analysis}, we use two functions, denoted $g$ and $\widetilde{g}$, to normalize the RH problem and open lenses. In Section \ref{section: local parametrices and S to P transformation}, we build local parametrices (\textit{without} the use of the global parametrix) and use them to define a new RH problem $P$. Section \ref{section: global parametrix} is devoted to the construction of the global parametrix $P^{(\infty)}$ and here the function $J$ plays a crucial role. In Section \ref{section:small norm shifted RHP}, we use again $J$ and obtain small norm estimates for the solution of a scalar shifted RH problem. In Section \ref{section: final computation}, we use the analysis of Sections \ref{section: equilibrium problem}--\ref{section:small norm shifted RHP} to obtain the large $n$ asymptotics for $\kappa_{n}$. We then substitute these asymptotics in \eqref{Dn in terms of the product} and prove Theorem \ref{theorem J}. The proof of Theorem \ref{thm:rigidity} is done in Section \ref{Section: rigidity}.

\section{Equilibrium problem}\label{section: equilibrium problem}
In this section we prove Proposition \ref{prop:density} using (an extension of) the method of \cite[Section 4]{ClaeysRomano}. An important difference with \cite{ClaeysRomano} is that in our case the equilibrium measure has two hard edges. 
\begin{lemma}\label{lemma:c0c1}
Let $\theta>0$, $b>a>0$, and recall that $s_{a}=s_{a}(\frac{c_{0}}{c_{1}})$ and $s_{b}=s_{b}(\frac{c_{0}}{c_{1}})$ are given by \eqref{def of sa sb}. There exists a unique tuple $(c_{0},c_{1})$ satisfying
\begin{align}\label{lol lemma}
J(s_{a}(\tfrac{c_{0}}{c_{1}});c_{0},c_{1})=a, \qquad J(s_{b}(\tfrac{c_{0}}{c_{1}});c_{0},c_{1})=b, \qquad c_{0}>c_{1}>0.
\end{align}
\end{lemma}
\begin{proof}
Let $x:= \frac{c_{0}}{c_{1}}>1$, and note that 
\begin{align}\label{J two fcts of x}
J(s_{a}(\tfrac{c_{0}}{c_{1}});c_{0},c_{1}) = c_{1}J(s_{a}(x);x,1), \qquad J(s_{b}(\tfrac{c_{0}}{c_{1}});c_{0},c_{1}) = c_{1}J(s_{b}(x);x,1).
\end{align}
For $x>1$, define $f(x) = \frac{J(s_{b}(x);x,1)}{J(s_{a}(x);x,1)}$. A simple computation shows that $f(x) \to +\infty$ as $x \to 1_{+}$, that $f(x) \to 1_{+}$ as $x \to +\infty$, and that $f'(x)<0$ for all $x>1$. This implies that for any $b>a>0$, there exists a unique $x_{\star}>1$ such that $f(x_{\star}) = \frac{b}{a}$. By \eqref{lol lemma}--\eqref{J two fcts of x}, the claim follows with 
\begin{align*}
c_{1} = \frac{b}{J(s_{b}(x_{\star});x_{\star},1)}>0, \qquad c_{0} = x_{\star}\frac{b}{J(s_{b}(x_{\star});x_{\star},1)}.
\end{align*}
\end{proof}
Proposition \ref{prop:density} is first proved for $\theta \geq 1$ in Section \ref{section: density proof theta geq 1}, and then we indicate the changes to make to treat the general case $\theta >0$ in Section \ref{section: proof of measure for all theta}. We mention that the general case $\theta > 0$ is not more complicated than the case $\theta \geq 1$, but it requires to introduce more notation and material.

\subsection{Proof of Proposition \ref{prop:density} for $\theta \geq 1$}\label{section: density proof theta geq 1}
Let 
\begin{align}\label{def of I1 and I2}
I_{1}: \mathbb{C}\setminus[a,b] \to \mathbb{C}\setminus \overline{D} \quad \mbox{ and } \quad I_{2}: \mathbb{H}_{\theta}\setminus[a,b] \to D \setminus [-1,0]
\end{align}
denote the inverses of the two functions in \eqref{two bijection}. We will also use the notation
\begin{align*}
I_{j,\pm}(x) = \lim_{\epsilon\to 0_{+}} I_{j}(x \pm i\epsilon), \qquad j=1,2, \quad x \in (a,b).
\end{align*}
As shown in Figure \ref{fig: the mapping J}, we have
\begin{align*}
I_{1,+}(x)=I_{2,-}(x), \qquad I_{1,-}(x) = I_{2,+}(x), \qquad x \in (a,b).
\end{align*}
Now, we make the ansatz that there exists a probability measure $\mu_{\theta}$, supported on $[a,b]$ with a continuous density $\rho$, which satisfies the Euler-Lagrange equality \eqref{EL equality}. Following \cite{ClaeysRomano}, we consider the following functions
\begin{align}
& g(z) = \int_{a}^{b} \log(z-y)d\mu_{\theta}(y), & & z \in \mathbb{C}\setminus (-\infty,b], \label{def of g} \\
& \widetilde{g}(z) = \int_{a}^{b} \log(z^{\theta}-y^{\theta})d\mu_{\theta}(y), & & z \in \mathbb{H}_{\theta}\setminus [0,b], \label{def of g tilde}
\end{align}
where the principal branches are taken for the logarithms and for $z \mapsto z^{\theta}$. For $x > 0$, we also define
\begin{align*}
g_{\pm}(x) = \lim_{\epsilon \to 0_{+}} g(x\pm i \epsilon), \quad \widetilde{g}_{\pm}(x) = \lim_{\epsilon \to 0_{+}} \widetilde{g}(x\pm i \epsilon), \quad \widetilde{g}(e^{\pm \frac{\pi i}{\theta}}x) = \lim_{z \to e^{\pm \frac{\pi i}{\theta}}x, \; z \in \mathbb{H}_{\theta}} \widetilde{g}(z).
\end{align*}
Using \eqref{EL equality} and $\int_{a}^{b}d\mu_{\theta}=1$, we infer that $g$ and $\widetilde{g}$ satisfy the following conditions.
\subsubsection*{RH problem for $(g,\widetilde{g})$}
\begin{itemize}
\item[(a)] $(g,\widetilde{g})$ is analytic in $(\mathbb{C}\setminus(-\infty,b],\mathbb{H}_{\theta}\setminus[0,b])$.
\item[(b)] $g_{\pm}(x) + \widetilde{g}_{\mp}(x) = -\ell \hspace{1.1cm}$ for $x \in (a,b)$, 

$\widetilde{g}(e^{\frac{\pi i}{\theta}}x) = \widetilde{g}(e^{-\frac{\pi i}{\theta}}x) + 2\pi i \hspace{0.2cm}$ for $x>0$, 

$\widetilde{g}_{+}(x) = \widetilde{g}_{-}(x) + 2\pi i \hspace{1cm}$ for $x \in (0,a)$, 

$g_{+}(x) = g_{-}(x) + 2\pi i \hspace{1cm}$ for $x < a$.
\item[(c)] $g(z) = \log (z) + \bigO(z^{-1}) \hspace{0.1cm}$ as $z \to \infty$,

$\widetilde{g}(z) = \theta \log z + \bigO(z^{-\theta}) \hspace{0.1cm}$ as $z \to \infty$ in $\mathbb{H}_{\theta}$.
\end{itemize}
Consider the derivatives
\begin{align}\label{def of G Gp}
G(z) = g'(z), \qquad \widetilde{G}(z) = \widetilde{g}'(z).
\end{align}
The properties of $(g,\widetilde{g})$ then imply that $(G,\widetilde{G})$ satisfy the following RH problem.
\subsubsection*{RH problem for $(G,\widetilde{G})$}
\begin{itemize}
\item[(a)] $(G,\widetilde{G})$ is analytic in $(\mathbb{C}\setminus[a,b],\mathbb{H}_{\theta}\setminus[a,b])$.
\item[(b)] $G_{\pm}(x) + \widetilde{G}_{\mp}(x) = 0 \hspace{1cm}$ for  $x \in (a,b)$,

$\widetilde{G}(e^{-\frac{\pi i}{\theta}}x) = e^{\frac{2\pi i}{\theta}} \widetilde{G}(e^{\frac{\pi i}{\theta}}x) \hspace{0.2cm}$ for $x > 0$.
\item[(c)] $G(z) = \frac{1}{z}+\bigO(z^{-2}) \hspace{0.5cm}$ as $z \to \infty$,

$\widetilde{G}(z) = \frac{\theta}{z}+\bigO(z^{-1-\theta}) \hspace{0.15cm}$ as $z \to \infty$ in $\mathbb{H}_{\theta}$.
\end{itemize}
To find a solution to this RH problem, we follow \cite{ClaeysWang, ClaeysRomano} and define 
\begin{align}\label{G to M transformation}
M(s) = \begin{cases}
G(J(s)), & \mbox{for } s \mbox{ outside } \gamma, \\
\widetilde{G}(J(s)), & \mbox{for } s \mbox{ inside } \gamma,
\end{cases}
\end{align}
where $J$ is given by \eqref{def of the map J} with $c_{0}>c_{1}>0$ such that $J(s_{a})=a$ and $J(s_{b})=b$.  By combining the RH conditions of $(G,\widetilde{G})$ with the properties of $J$ summarized in Proposition \ref{prop:ClaeysRomano}, we see that $M$ satisfies the following RH problem.
\subsubsection*{RH problem for $M$}
\begin{itemize}
\item[(a)] $M$ is analytic in $\mathbb{C}\setminus (\gamma \cup [-1,0])$.
\item[(b)] Let $[-1,0]$ be oriented from left to right, and recall that $\gamma$ is oriented in the counterclockwise direction. For $s \in (\gamma\cup (-1,0))\setminus \{s_{a},s_{b}\}$, we denote $M_{+}(s)$ and $M_{-}(s)$ for the left and right boundary values, respectively. The jumps for $M$ are given by
\begin{align*}
& M_{+}(s) + M_{-}(s) = 0, & & \mbox{for } s \in \gamma\setminus \{s_{a},s_{b}\}, \\
& M_{+}(s) = e^{\frac{2\pi i}{\theta}}M_{-}(s), & & \mbox{for } s \in (-1,0).
\end{align*}
\item[(c)] $M(s) = \frac{1}{J(s)}(1 + \bigO(s^{-1}))$ as $s \to \infty$,

$M(s) = \frac{\theta}{J(s)}(1 + \bigO(s)) \hspace{0.37cm}$ as $s \to 0$,

$M(s) = \bigO(1) \hspace{1.85cm}$ as $s \to -1$.
\end{itemize}
We now apply the transformation $N(s) = J(s)M(s)$ and obtain the following RH problem.
\subsubsection*{RH problem for $N$}
\begin{itemize}
\item[(a)] $N$ is analytic in $\mathbb{C}\setminus \gamma$.
\item[(b)] $N_{+}(s) + N_{-}(s) = 0$ for $s \in \gamma \setminus\{s_{a},s_{b}\}$.
\item[(c)] $N(s) = 1+\bigO(s^{-1})$ as $s \to \infty$.

$N(0) = \theta$ and $N(-1)=0$.
\end{itemize}
The solution of this RH problem is not unique without prescribing the behavior of $N$ near $s_{a}$ and $s_{b}$. Recalling that $a>0$, one expects the density $\rho$  to blow up like an inverse square root near $a$ and $b$ (as is usually the case near standard hard edges). To be consistent with this heuristic, using \eqref{def of G Gp}, \eqref{G to M transformation} and $N(s) = J(s)M(s)$ we verify that $N$ must blow up like $(s-s_{j})^{-1}$, as $s \to s_{j}$, $j=a,b$. With this in mind, we consider the following solution to the RH problem for $N$:
\begin{align}\label{def of N}
N(s) = \begin{cases}
\ds 1+\frac{d_{a}}{s-s_{a}}+\frac{d_{b}}{s-s_{b}}, & \mbox{outside } \gamma, \\
\ds -1 - \frac{d_{a}}{s-s_{a}} - \frac{d_{b}}{s-s_{b}}, & \mbox{inside } \gamma,
\end{cases}
\end{align}
where $d_{a}$ and $d_{b}$ are chosen such that $N(0) = \theta$ and $N(-1) = 0$, i.e. such that
\begin{align*}
& \frac{d_{a}}{s_{a}} + \frac{d_{b}}{s_{b}} = 1+\theta \qquad \mbox{ and } \qquad \frac{d_{a}}{1+s_{a}} + \frac{d_{b}}{1+s_{b}} = 1.
\end{align*}
This system can be solved explicitly,
\begin{align}\label{explicit expressions for da and db}
d_{a} = \frac{s_{a}(1+s_{a})(s_{b}\theta-1)}{s_{b}-s_{a}}, \qquad d_{b} = \frac{s_{b}(1+s_{b})(1-s_{a}\theta)}{s_{b}-s_{a}},
\end{align}
and since $s_{a}<-1$ and $\frac{1}{\theta}<s_{b}$, we have $d_{a}>0$, $d_{b}>0$. Writing 
\begin{align*}
d\mu_{\theta}(x) = \rho(x)dx, \qquad x \in (a,b),
\end{align*}
we obtain
\begin{align}
& \rho(x)  = - \frac{1}{2\pi i}(G_{+}(x)-G_{-}(x)) = - \frac{1}{2\pi i x}\big(N_{-}(I_{1,+}(x))-N_{-}(I_{1,-}(x))\big) \label{explicit expression for rho} \\
& = - \sum_{j=a,b}\frac{d_{j}}{2\pi i x}\bigg( \frac{1}{I_{1,+}(x)-s_{j}} - \frac{1}{I_{1,-}(x)-s_{j}} \bigg)  = - \sum_{j=a,b} \frac{d_{j}}{\pi x} \im \bigg( \frac{1}{I_{1,+}(x)-s_{j}} \bigg). \nonumber
\end{align}
By construction, $\int_{a}^{b}\rho(x)dx=1$, but it remains to check that $\rho$ is indeed a density. This can be readily verified from \eqref{explicit expression for rho}: since $d_{a}>0$, $d_{b}>0$ and $\im I_{1,+}(x)>0$ for all $x \in (a,b)$, we have $\rho(x)>0$ for all $x \in (a,b)$. Thus, we have shown that the unique measure $\mu_{\theta}$ satisfying \eqref{EL equality} is given by $d\mu_{\theta}(x)=\rho(x)dx$ with $\rho$ as in \eqref{explicit expression for rho}.

\medskip To conclude the proof of Proposition \ref{prop:density} for $\theta \geq 1$, it remains to prove that $\rho$ can be rewritten in the simpler form \eqref{explicit expression for rho intro}. For this, we first use the relation $J(I_{k}(z))=z$ for $z\in \mathbb{C}\setminus [a,b]$, $k=1,2$, to obtain
\begin{align}\label{im big 1}
\frac{I_{k}'(z)}{I_{k}(z)}  = \frac{1}{I_{k}(z)J'(I_{k}(z))}  = \frac{\theta(1+I_{k}(z))(c_{1}I_{k}(z)+c_{0})}{z(-c_{0}+c_{1}I_{k}(z)(\theta-1+\theta I_{k}(z)))}.
\end{align}
On the other hand, using the explicit expressions for $d_{a}$ in $d_{b}$ given by \eqref{explicit expressions for da and db}, we arrive at
\begin{align}\label{im big 2}
\sum_{j=a,b} \frac{d_{j}}{z} \frac{1}{I_{k}(z)-s_{j}}  = \frac{-1}{z} \frac{c_{0}+c_{1} I_{k}(z)+c_{0}\theta (1+I_{k}(z))}{c_{0}-c_{1}I_{k}(z)(\theta-1+\theta I_{k}(z))}, 
\end{align}
where $z\in \mathbb{C}\setminus [a,b], \; k=1,2$.
Using \eqref{im big 1} and \eqref{im big 2}, it is direct to verify that
\begin{align}\label{interesting identity}
\frac{1}{z}+\sum_{j=a,b} \frac{d_{j}}{z} \frac{1}{I_{k}(z)-s_{j}}  = \frac{I_{k}'(z)}{I_{k}(z)}, \qquad z\in \mathbb{C}\setminus [a,b], \; k=1,2,
\end{align}
which implies in particular \eqref{explicit expression for rho intro}: 
\begin{align*}
\rho(x) = - \sum_{j=a,b} \frac{d_{j}}{\pi x} \im \bigg( \frac{1}{I_{1,+}(x)-s_{j}} \bigg) = -\frac{1}{\pi} \im \bigg( \frac{I_{1,+}'(x)}{I_{1,+}(x)} \bigg), \qquad x \in (a,b).
\end{align*}
Formulas \eqref{explicit expression for rho intro} and \eqref{interesting identity} will allow us to simplify several complicated expressions appearing in later sections, and can already be used to find an explicit expression for $\ell$. 
\begin{lemma}\label{lemma:asymp of I1 and I2 at infty}
As $z \to \infty$, we have
\begin{align}\label{asymp of I1 and I2 at infty}
& I_{1}(z) = c_{1}^{-1} z + \bigO(1), \qquad I_{2}(z) = z^{-\theta} \big( c_{0}^{\theta} + \bigO(z^{-\theta}) \big).
\end{align}
\end{lemma}
\begin{proof}
It suffices to combine the expansions 
\begin{align*}
J(s) = c_{1}s+\bigO(1) \; \mbox{ as } \; s \to \infty, \qquad \qquad J(s) = c_{0}s^{-\frac{1}{\theta}}(1+\bigO(s)) \; \mbox{ as } \; s \to 0,
\end{align*}
with the identities $J(I_{k}(z))=z$, $k=1,2$.
\end{proof}
\begin{proposition}\label{prop:simplified expression for ell}
$\ell = -\log c_{1} - \theta \log c_{0}.$
\end{proposition}
\begin{proof}
Using \eqref{def of G Gp}, \eqref{G to M transformation}, $N(s)=J(s)M(s)$, \eqref{def of N} and \eqref{interesting identity}, we obtain
\begin{align*}
& g'(z) = M(I_{1}(z)) = \frac{1}{z}\bigg(1+\frac{d_{a}}{I_{1}(z)-s_{a}}+\frac{d_{b}}{I_{1}(z)-s_{b}} \bigg) = \frac{I_{1}'(z)}{I_{1}(z)}, & & \hspace{-0.25cm} z \in \mathbb{C}\setminus (-\infty,b], \\
& \widetilde{g}'(z) = M(I_{2}(z)) = \frac{-1}{z}\bigg(1+\frac{d_{a}}{I_{2}(z)-s_{a}}+\frac{d_{b}}{I_{2}(z)-s_{b}} \bigg) = -\frac{I_{2}'(z)}{I_{2}(z)}, & & z \in \mathbb{H}_{\theta}\setminus [0,b].
\end{align*}
Hence, by \eqref{asymp of I1 and I2 at infty} and the condition (c) of the RH problem for $(g,\widetilde{g})$, we find
\begin{align*}
& g(z) = \log(z) + \int_{b}^{z}\bigg( \frac{I_{1}'(x)}{I_{1}(x)} - \frac{1}{x} \bigg)dx - \int_{b}^{\infty}\bigg( \frac{I_{1}'(x)}{I_{1}(x)} - \frac{1}{x} \bigg)dx, & & z \in \mathbb{C}\setminus (-\infty,b],  \\
& \widetilde{g}(z) = \theta \log(z) - \int_{b}^{z}\bigg( \frac{I_{2}'(x)}{I_{2}(x)} + \frac{\theta}{x} \bigg)dx + \int_{b}^{\infty}\bigg( \frac{I_{2}'(x)}{I_{2}(x)} + \frac{\theta}{x} \bigg)dx, & & z \in \mathbb{H}_{\theta}\setminus [0,b]. 
\end{align*}
The integrals over $(b,\infty)$ can be evaluated explicitly using \eqref{asymp of I1 and I2 at infty}:
\begin{align}
& - \int_{b}^{\infty}\bigg( \frac{I_{1}'(x)}{I_{1}(x)} - \frac{1}{x} \bigg)dx = \lim_{r \to + \infty} \log \frac{rI_{1}(b)}{bI_{1}(r)} = \log \frac{c_{1}I_{1}(b)}{b} = \log \frac{c_{1}s_{b}}{b}, \label{int expl g} \\
& \int_{b}^{\infty}\bigg( \frac{I_{2}'(x)}{I_{2}(x)} + \frac{\theta}{x} \bigg)dx = \lim_{r\to + \infty} \log \frac{r^{\theta}I_{2}(r)}{b^{\theta}I_{2}(b)} = \log \frac{c_{0}^{\theta}}{b^{\theta}s_{b}}. \label{int expl g tilde}
\end{align}
Substituting \eqref{int expl g}--\eqref{int expl g tilde} in the above expressions for $g$ and $\widetilde{g}$, and using the Euler-Lagrange equality $\ell = -(g(b)+\widetilde{g}(b))$, we find the claim.
\end{proof}

\subsection{Proof of Proposition \ref{prop:density} for all $\theta > 0$}\label{section: proof of measure for all theta}
We first prove a generalization of Proposition \ref{prop:ClaeysRomano}. 
\begin{proposition}\emph{(extension of \cite[Lemma 4.3]{ClaeysRomano} to all $\theta>0$)}.\label{prop:ClaeysRomano generalization}
Let $\theta > 0$, and let $c_{0}>c_{1}>0$ be such that \eqref{system that defines c0 and c1} holds. There are two complex conjugate curves $\gamma_{1}$ and $\gamma_{2}$ starting at $s_{a}$ and ending at $s_{b}$ in the upper and lower half plane respectively which are mapped to the interval $[a,b]$ through $J$. Let $\gamma$ be the counterclockwise oriented closed curve consisting of the union of $\gamma_{1}$ and $\gamma_{2}$, enclosing a region $D$. The maps 
\begin{align}\label{two bijection generalization}
J:\mathbb{C}\setminus \overline{D} \to \mathbb{C}\setminus[a,b], \qquad J^{\theta} : D \setminus [-1,0] \to \mathbb{C}\setminus\big( (-\infty,0]\cup [a^{\theta},b^{\theta}] \big)
\end{align}
are bijections, where $J^{\theta}(s):= \frac{s+1}{s}(c_{1}s+c_{0})^{\theta}$ and the principal branch is taken for $(c_{1}s+c_{0})^{\theta}$.
\end{proposition}
\begin{remark}
We emphasize that for $\theta < 1$, the definition of $J^{\theta}(s)$ does \textbf{not} coincide with $J(s)^{\theta}$ where the principal branch is taken for $(\cdot)^{\theta}$. On the contrary, for all $\theta > 0$ and $s \in D\setminus[-1,0]$, the definition \eqref{def of the map J} of $J(s)$ coincides with $J(s) = J^{\theta}(s)^{\frac{1}{\theta}}$ where the principal branch is chosen for $(\cdot)^{\frac{1}{\theta}}$.
\end{remark}
\begin{proof}
Write $s=re^{i\phi}$ with $-\pi < \phi \leq \pi$. It is readily checked that $J(s)>0$ if and only if
\begin{align}\label{arg annoying}
\arg \bigg( \frac{c_{0}}{c_{1}}+re^{i\phi} \bigg) + \frac{1}{\theta} \arg(1+re^{i\phi}) - \frac{\phi}{\theta} = 2k \pi, \qquad k \in \mathbb{Z},
\end{align}
where the branch for $\arg$ is chosen such that $\arg(z) \in (-\pi,\pi]$ for all $z \in \mathbb{C}\setminus \{0\}$. For $\phi \in (0,\pi)$, the left-hand side is increasing in $r$ (since $\frac{c_{0}}{c_{1}}>0$), tends to $-\frac{\phi}{\theta}$ as $r \to 0$, and to $\phi$ as $r \to \infty$. The set of points $(\phi,k)$ for which there exists a (necessary unique) $r$ satisfying \eqref{arg annoying} is therefore given by $\{(\phi,k): \phi > 2\pi |k| \theta, \; -k \in \mathbb{N}\}$. For each $k\in \{0,-1,\ldots,-\lceil \frac{1}{2\theta}\rceil+1\}$, denote $\Gamma_{k}$ for the set of points $re^{i\theta}$ with $\phi \in (0,\pi)$ satisfying \eqref{arg annoying}. It is not hard to verify that $\Gamma_{0}$ joins $s_{a}$ with $s_{b}$, while the other curves $\Gamma_{1},\ldots,\Gamma_{-\lceil \frac{1}{2\theta}\rceil+1}$ join $-1$ with $0$, see also Figure \ref{fig: the mapping J several curves} (left). The curve $\gamma_{1} := \Gamma_{0}$ is mapped bijectively by $J$ to $(a,b)$, and since $J(s) = \smash{\overline{J(\overline{s})}}$, the curve $\gamma_{2}:=\smash{\overline{\gamma_{1}}}$ is also mapped bijectively by $J$ to $(a,b)$. 

Thus, $J$ maps bijectively the boundaries of $\mathbb{C}\setminus \overline{D}$ to the boundaries of $\mathbb{C}\setminus[a,b]$. It is also straightforward to see that $J^{\theta}$ maps bijectively $[-1,0)$ to $(-\infty,0]$. The claim that the maps \eqref{two bijection} are bijections can now be proved exactly as in \cite[Section 4.1]{ClaeysRomano}.
\end{proof}

\begin{figure}[h!]
\begin{center}
\begin{tikzpicture}[master]
\node at (0,0) {\includegraphics[scale=0.2]{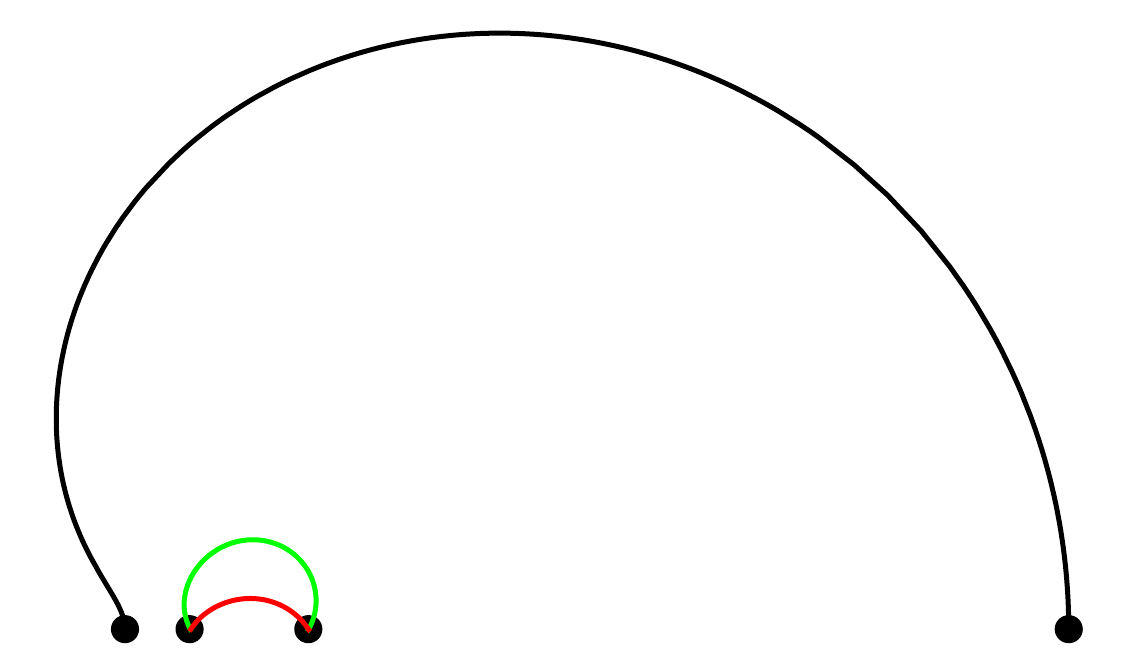}};
\end{tikzpicture}
\begin{tikzpicture}[slave]
\node at (0,0) {\includegraphics[scale=0.2]{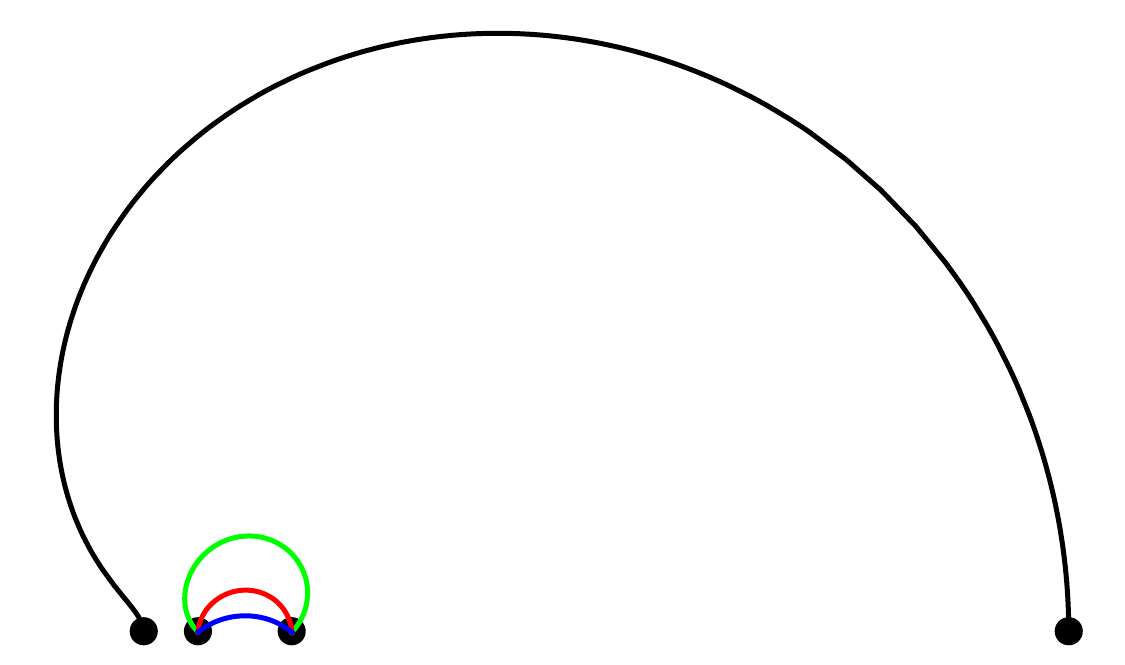}};
\end{tikzpicture}
\hspace{0.1cm}\begin{tikzpicture}[slave]
\draw[dashed] (-2,-1.2)--(-2,1.2);
\draw[black, line width=0.3 mm] (-75:1.5)--(0,0)--(75:1.5);
\draw[black, line width=0.3 mm, dashed] (0:1.5)--(0,0)--(150:1.5);
\draw[black, line width=0.3 mm] (180:1.5)--(0,0);
\draw[black, line width=0.3 mm, dashed] (0,0)--(-150:1.5);
\draw[black,line width=0.3 mm] ([shift=(-75:0.3cm)]0,0) arc (-75:75:0.3cm);
\node at (0.9,0.2) {\color{black}\small $k=0$};
\draw[green ,line width=0.3 mm] ([shift=(-150:0.4cm)]0,0) arc (-150:0:0.4cm);
\node at (0.7,-0.5) {\color{green}\small $k=-1$};
\draw[red ,line width=0.3 mm] ([shift=(75:0.4cm)]0,0) arc (75:180:0.4cm);
\node at (120:0.65) {\color{red}\small $k=2$};
\end{tikzpicture}
\end{center}
\caption{\label{fig: the mapping J several curves}
The two figures on the left correspond to $\theta=0.17$ and $\theta=\frac{1}{7.7}$. The four dots are $s_{a}$, $-1$, $0$ and $s_{b}$. The black, green, red and blue curves correspond to the points $re^{i\phi}$, $\phi \in (0,\pi)$, satisfying \eqref{arg annoying} for $k=0$, $k=-1$, $k=-2$ and $k=-3$, respectively. (These figures have been made with $c_{0}=0.8$ and $c_{1}=0.47$.) The right-most figure shows the projections in the $y$-plane of $\mathcal{H}_{\theta,k}$, $k=-2,\ldots,2$ for $\theta=\frac{5}{12}$.}
\end{figure}
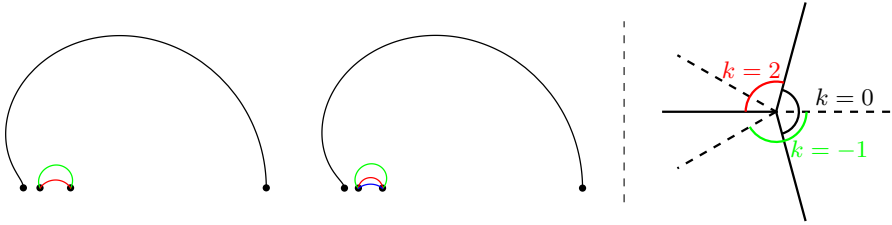
As can be seen from Proposition \ref{prop:ClaeysRomano generalization}, for $\theta < 1$ the mapping $J:D \setminus [-1,0]\to \mathbb{H}_{\theta}\setminus [a,b]$ is not a bijection and therefore one cannot define $I_{2}$ as in \eqref{def of I1 and I2}. In view of \eqref{two bijection generalization}, instead of working with the set $\mathbb{H}_{\theta}$, one is naturally led to consider the following Riemann surface $\mathcal{H}_{\theta}$. 
\begin{definition}
Let $\mathcal{H}_{\theta}$ be the Riemann surface
\begin{align*}
\mathcal{H}_{\theta} = \Big\{(z,y)\in \mathbb{C}^{2} : z = y^{\frac{1}{\theta}}, \, y \in \mathbb{C}\setminus (-\infty,0]\Big\}, \quad  y^{\frac{1}{\theta}} := |y|^{\frac{1}{\theta}}e^{\frac{i}{\theta}\arg y}, \;\arg y \in (-\pi,\pi),
\end{align*}
endowed with the atlas $\{ \varphi_{\theta,k}:\mathcal{H}_{\theta,k}\to \mathbb{C} \}_{k=-\lceil \frac{1}{\theta}-1 \rceil,\ldots,\lceil \frac{1}{\theta}-1 \rceil}$, where
\begin{align*}
& \mathcal{H}_{\theta,k} = \Big\{(z,y)\in \mathbb{C}^{2} : z = y^{\frac{1}{\theta}}, \, \max\{(k-1)\pi\theta,-\pi\} < \arg y < \min\{(k+1)\pi\theta,\pi\} \Big\}, 
\end{align*}
and $\varphi_{\theta,k}(z,w):=z$, see also Figure \ref{fig: the mapping J several curves} (right). 
\end{definition}
\begin{remark}
For $\theta \geq 1$, there is just a single map $\varphi_{\theta,0}$ in the atlas, and it satisfies $\varphi_{\theta,0}(\mathcal{H}_{\theta,0})=\mathbb{H}_{\theta}$, where we recall that $\mathbb{H}_{\theta} = \{z \in \mathbb{C}\setminus\{0\}: -\frac{\pi}{\theta}<\arg z < \frac{\pi}{\theta}\}$.
\end{remark}
\begin{definition}\label{def: from C to H}
A mapping $f:B\subset \mathbb{C} \to \mathcal{H}_{\theta}$ is analytic if for all $k$ with $f(B)\cap \mathcal{H}_{\theta,k}\neq \emptyset$, the function $\varphi_{\theta,k}\circ f:B\cap f^{-1}(\mathcal{H}_{\theta,k}) \to \mathbb{C}$ is analytic.
\end{definition}
\begin{definition}\label{def: from H to C}
A mapping $h:H \subset \mathcal{H}_{\theta} \to \mathbb{C}$ is analytic if for all $k$ with $H\cap \mathcal{H}_{\theta,k}\neq \emptyset$, the function $h \circ \varphi_{\theta,k}^{-1}:\varphi_{\theta,k}(\mathcal{H}_{\theta,k}\cap H) \to \mathbb{C}$ is analytic.
\end{definition}
\begin{definition}
For notational convenience, given $I \subset \mathbb{C}$, we define 
\begin{align*}
\mathcal{H}_{\theta}\setminus I:= \{(z,y)\in \mathbb{C}^{2} : z = y^{\frac{1}{\theta}}, \, y \in \mathbb{C}\setminus (-\infty,0], \; z \notin I \big)\} \subset \mathcal{H}_{\theta}.
\end{align*}
\end{definition}
Proposition \ref{prop:ClaeysRomano generalization} and Definition \ref{def: from C to H} imply that
\begin{align}\label{J Jtheta}
(J,J^{\theta}): D\setminus [-1,0] \to \mathcal{H}_{\theta}\setminus [a,b]
\end{align}
is an analytic bijection. Let $\widetilde{I}_{2} : \mathbb{C}\setminus\big( (-\infty,0]\cup [a^{\theta},b^{\theta}] \big) \to D \setminus [-1,0]$ be the inverse of $J^{\theta}$. The inverse of \eqref{J Jtheta} is then given by
\begin{align*}
\widehat{I}_{2}: \mathcal{H}_{\theta}\setminus[a,b] \to D\setminus [-1,0], \qquad (z,y)\mapsto \widehat{I}_{2}(z,y)= \widetilde{I}_{2}(y).
\end{align*}
\begin{remark}
For $\theta \geq 1$, the map $J:D\setminus [-1,0]\to \mathbb{H}_{\theta}\setminus [a,b]$ is a bijection and there is no need to define $\mathcal{H}_{\theta}$ and $\widehat{I}_{2}$. In fact, for $\theta \geq 1$ and $z \in \mathbb{H}_{\theta}\setminus [a,b]$, $\widehat{I}_{2}(z,y)$ and $I_{2}(z)$ are directly related by $I_{2}(z)=\widehat{I}_{2}(z,y)$, where $y \in \mathbb{C}\setminus\big( (-\infty,0]\cup [a^{\theta},b^{\theta}] \big)$ is the unique solution to
\begin{align*}
z = y^{\frac{1}{\theta}}, \qquad \mbox{and} \quad  y^{\frac{1}{\theta}} = |y|^{\frac{1}{\theta}}e^{\frac{i}{\theta}\arg y}, \;\arg y \in (-\pi,\pi).
\end{align*}
\end{remark}
Define
\begin{align}
& \widehat{g}(z,y) = \int_{a}^{b} \log(y-x^{\theta})d\mu_{\theta}(x), & & (z,y) \in \mathcal{H}_{\theta}\setminus [0,b]. \label{def of g hat}
\end{align}
Now, to prove Propositions \ref{prop:density} and \ref{prop:simplified expression for ell} for general $\theta>0$, it suffices to follow the analysis of Section \ref{section: density proof theta geq 1} and to replace all occurrences of $\widetilde{g}$, $z \in \mathbb{H}_{\theta}$, $z^{\theta}$ and $I_{2}(z)$ as follows
\begin{align}\label{annoying substitutions}
\widetilde{g} \mapsto \widehat{g}, \qquad z \in \mathbb{H}_{\theta} \mapsto (z,y) \in \mathcal{H}_{\theta}, \qquad z^{\theta}\mapsto y, \qquad I_{2}(z)\mapsto \widehat{I}_{2}(z,y).
\end{align}

\section{Asymptotic analysis of $Y$: first steps}\label{section: steepest descent analysis}
	We start by recalling the RH problem for $Y$ from \cite{ClaeysRomano} which uniquely characterizes $\kappa_{n}^{-1}p_{n}$ as well as $\kappa_{n}^{-1}Cp_{n}$ (recall that $p_{n}$ and $Cp_{n}$ are defined in \eqref{p ortho in bioortho} and \eqref{def of Cpj}). For convenience, we say that a function $f$ is defined in $\mathbb{H}_{\theta}^{c}$ if it is defined in $\mathbb{H}_{\theta}$, that the limits $f(e^{\pm \frac{\pi i}{\theta}}x) = \lim_{\smash{z \to e^{\pm \frac{\pi i}{\theta}}x, \; z \in \mathbb{H}_{\theta}}} f(z)$ exist for all $x\geq 0$, and furthermore $f(e^{\frac{\pi i}{\theta}}x) = f(e^{-\frac{\pi i}{\theta}}x)$ for all $x\geq 0$. 
\begin{theorem}(\cite[Theorem 1.3]{ClaeysRomano}).\label{thm: RHP for Y ClaeysRomano}
Define $Y$ by
\begin{align}\label{def of Y}
Y(z) = \bigg( \frac{1}{\kappa_{n}}p_{n}(z), \frac{1}{\kappa_{n}}Cp_{n}(z) \bigg).
\end{align}
If $Y$ exists, then it is the unique function which satisfies the following conditions:
\end{theorem}
\subsubsection*{RH problem for $Y$}
\begin{itemize}
\item[(a)] $Y=(Y_{1},Y_{2})$ is analytic in $(\mathbb{C},\mathbb{H}_{\theta}^{c}\setminus [a,b])$. 
\item[(b)] The jumps are given by
\begin{align*}
& Y_{+}(x) = Y_{-}(x) \begin{pmatrix}
1 & \frac{1}{\theta x^{\theta-1}}w(x) \\
0 & 1
\end{pmatrix}, & & x\in (a,b)\setminus\{t_{1},\ldots,t_{m}\}.
\end{align*}
\item[(c)] $Y_{1}(z) = z^{n} + \bigO(z^{n-1}) \hspace{0.1cm}$ as $z \to \infty$,

$Y_{2}(z) = \bigO(z^{-(n+1)\theta}) \hspace{0.32cm}$ as $z \to \infty$ in $\mathbb{H}_{\theta}$.

\item[(d)] As $z \to t_{j}$, $j=0,1,\ldots,m,m+1$, we have
\begin{align*}
Y_{1}(z) = \bigO(1), \qquad Y_{2}(z) = \begin{cases}
\bigO(1)+\bigO((z-t_{j})^{\alpha_{j}}), & \mbox{if } \alpha_{j} \neq 0, \\
\bigO(\log (z-t_{j})), & \mbox{if } \alpha_{j}=0,
\end{cases}
\end{align*}
where $t_{0}:=a>0$ and $t_{m+1}:=b$.
\end{itemize}
As mentioned in the introduction, if $w$ is positive, then the existence of $Y$ is ensured by \cite[Section 2]{ClaeysRomano}. In our case, $w$ is complex valued and this is no longer guaranteed. Nevertheless, it will follow from our analysis that $Y$ exists for all large enough $n$.

\begin{remark}
In a similar way as in Section \ref{section: proof of measure for all theta}, we mention that to be formal, for $\theta < 1$ one would need to replace all occurrences of $\widetilde{g}$, $\mathbb{H}_{\theta}$, $z^{\theta}$ and $I_{2}(z)$ as in \eqref{annoying substitutions} and to define $Y_{2}$ as
\begin{align}\label{def of widehatY}
Y_{2}(z,y) = \frac{1}{2\pi i \kappa_{n}}\int_{a}^{b} \frac{p_{n}(x)}{x^{\theta}-y}w(x)dx, \qquad  (z,y) \in \mathcal{H}_{\theta}\setminus [a,b].
\end{align}
However, the $y$ coordinate will always be clear from the context, and for convenience we will slightly abuse notation and use $\widetilde{g}$, $\mathbb{H}_{\theta}$, $z^{\theta}$, $I_{2}(z)$ and $Y_{2}(z)$ for all values of $\theta>0$. 
\end{remark}

\medskip In the rest of this section, we will perform the first steps of the asymptotic analysis of $Y$ as $n \to +\infty$, following the method of \cite{ClaeysWang}. 
\subsection{First transformation: $Y \mapsto T$}
Recall that $g$ and $\widetilde{g}$ are defined in \eqref{def of g} and \eqref{def of g tilde}, and that $\ell$ is the Euler-Lagrange constant appearing in \eqref{EL equality} and in condition (b) of RH problem for $(g,\widetilde{g})$. 
The first transformation is defined by
\begin{align}\label{Y to T transformation}
T(z) = e^{\frac{n \ell}{2}}Y(z) \begin{pmatrix}
e^{-ng(z)} & 0 \\
0 & e^{n \widetilde{g}(z)}
\end{pmatrix} e^{-\frac{n \ell}{2}\sigma_{3}}, \qquad \mbox{where } \quad \sigma_{3}=\begin{pmatrix}
1 & 0 \\ 0 & -1
\end{pmatrix}.
\end{align}
Using the RH conditions of $Y$ and $(g,\widetilde{g})$, it can be checked that $T$ satisfies the following RH problem.
\subsubsection*{RH problem for $T$}
\begin{itemize}
\item[(a)] $T=(T_{1},T_{2})$ is analytic in $(\mathbb{C}\setminus [a,b],\mathbb{H}_{\theta}^{c}\setminus [a,b])$. 
\item[(b)] The jumps are given by
\begin{align*}
& T_{+}(x) = T_{-}(x) \begin{pmatrix}
e^{-n(g_{+}(x)-g_{-}(x))} & \frac{\omega(x)e^{W(x)}}{\theta x^{\theta-1}} \\
0 & e^{n(\widetilde{g}_{+}(x)-\widetilde{g}_{-}(x))}
\end{pmatrix}, & & x\in (a,b)\setminus\{t_{1},\ldots,t_{m}\}.
\end{align*}
\item[(c)] $T_{1}(z) = 1 + \bigO(z^{-1}) \hspace{0.1cm}$ as $z \to \infty$,

$T_{2}(z) = \bigO(z^{-\theta}) \hspace{0.7cm}$ as $z \to \infty$ in $\mathbb{H}_{\theta}$.
\item[(d)] As $z \to t_{j}$, $j=0,1,\ldots,m,m+1$, we have
\begin{align*}
T_{1}(z) = \bigO(1), \qquad T_{2}(z) = \begin{cases}
\bigO(1)+\bigO((z-t_{j})^{\alpha_{j}}), & \mbox{if } \alpha_{j} \neq 0, \\
\bigO(\log (z-t_{j})), & \mbox{if } \alpha_{j}=0.
\end{cases}
\end{align*}
\end{itemize}

\subsection{Second transformation: $T \mapsto S$}
Let $\mathcal{U}$ be an open small neighborhood of $[a,b]$ which is contained in both $\mathbb{C}$ and $\mathbb{H}_{\theta}$, and define
\begin{align}\label{def of phi}
\phi(z) = g(z) + \widetilde{g}(z) +\ell, \qquad z \in \mathcal{U}\setminus (0,b).
\end{align}
Using the RH conditions of $(g,\widetilde{g})$, we conclude that $\phi$ satisfies the jumps
\begin{align*}
& \phi_{+}(x) = \phi_{-}(x) + 4\pi i, & & x \in (0,a)\cap \mathcal{U}, \\
& \phi_{+}(x) + \phi_{-}(x) = 0, & & x \in (a,b).
\end{align*}
For $x \in (a,b) \setminus \{t_{1},\ldots,t_{m}\}$, we will use the following factorization of the jump matrix for $T$:
\begin{multline}
\begin{pmatrix}
e^{-n(g_{+}(z)-g_{-}(z))} & \frac{\omega(x)e^{W(x)}}{\theta x^{\theta-1}} \\
0 & e^{n(\widetilde{g}_{+}(x)-\widetilde{g}_{-}(x))}
\end{pmatrix} = \begin{pmatrix}
1 & 0 \\ e^{-n\phi_{-}(z)} \frac{\theta x^{\theta-1}}{\omega(x)e^{W(x)}} & 1
\end{pmatrix}  \\ \times  \begin{pmatrix}
0 & \frac{\omega(x)e^{W(x)}}{\theta x^{\theta-1}} \\
-\frac{\theta x^{\theta-1}}{\omega(x)e^{W(x)}} & 0
\end{pmatrix} \begin{pmatrix}
1 & 0 \\
e^{-n \phi_{+}(x)} \frac{\theta x^{\theta-1}}{\omega(x)e^{W(x)}} & 1
\end{pmatrix}. \label{factorization of the jumps}
\end{multline}
Before opening the lenses, we first note that $\omega_{\alpha_{k}}$ and $\omega_{\beta_{k}}$ can be analytically continued as follows:
\begin{equation}
\omega_{\alpha_{k}}(z) = \left\{ \hspace{-0.1cm} \begin{array}{l l}
(t_{k}-z)^{\alpha_{k}}, & \hspace{-0.1cm}\mbox{if } \re z < t_{k}, \\
(z-t_{k})^{\alpha_{k}}, & \hspace{-0.1cm}\mbox{if } \re z > t_{k},
\end{array} \right. \;\; \omega_{\beta_{k}}(z) = \left\{ \hspace{-0.1cm}\begin{array}{l l}
e^{i\pi\beta_{k}}, & \hspace{-0.1cm}\mbox{if } \re z < t_{k}, \\
e^{-i \pi \beta_{k}}, & \hspace{-0.1cm}\mbox{if } \re z > t_{k}.
\end{array}  \right.
\end{equation}
For each $j\in \{1,\ldots,m+1\}$, let $\sigma_{j,+}, \sigma_{j,-} \subset \mathcal{U}$ be open curves starting at $t_{j-1}$, ending at $t_{j}$, and lying in the upper and lower half plane, respectively (see also Figure \ref{fig:contour for S}). We also let $\mathcal{L}_{j} \subset \mathcal{U}$ denote the open bounded lens-shaped region surrounded by $\sigma_{j,+}\cup \sigma_{j,-}$. In view of \eqref{factorization of the jumps}, we define
\begin{align}\label{T to S transformation}
S(z) = \begin{cases}
T(z) \begin{pmatrix}
1 & 0 \\ -e^{-n\phi(z)} \frac{\theta z^{\theta-1}}{\omega(z)e^{W(z)}} & 1
\end{pmatrix}, & z \in \mathcal{L} \mbox{ and }\im z >0, \\
T(z) \begin{pmatrix}
1 & 0 \\ e^{-n\phi(z)} \frac{\theta z^{\theta-1}}{\omega(z)e^{W(z)}} & 1
\end{pmatrix}, & z \in \mathcal{L} \mbox{ and }\im z <0, \\
T(z), & \mbox{otherwise}.
\end{cases}
\end{align}
where $\mathcal{L}:=\cup_{j=1}^{m+1}\mathcal{L}_{j}$. $S$ satisfies the following RH problem.
\subsubsection*{RH problem for $S$}
\begin{itemize}
\item[(a)] $S=(S_{1},S_{2})$ is analytic in $(\mathbb{C}\setminus ([a,b]\cup \sigma_{+}\cup \sigma_{-}),\mathbb{H}_{\theta}^{c}\setminus ([a,b]\cup \sigma_{+}\cup \sigma_{-}))$, where $\sigma_{\pm} := \cup_{j=1}^{m+1}\sigma_{j,\pm}$.
\item[(b)] The jumps are given by
\begin{align}
& S_{+}(z) = S_{-}(z) \begin{pmatrix}
0 & \frac{\omega(z)e^{W(z)}}{\theta z^{\theta-1}} \\
-\frac{\theta z^{\theta-1}}{\omega(z)e^{W(z)}} & 0
\end{pmatrix}, & & z\in (a,b)\setminus\{t_{1},\ldots,t_{m}\}, \nonumber \\
& S_{+}(z) = S_{-}(z)\begin{pmatrix}
1 & 0 \\ e^{-n\phi(z)} \frac{\theta z^{\theta-1}}{\omega(z)e^{W(z)}} & 1
\end{pmatrix}, & & z \in \sigma_{+}\cup \sigma_{-}. \label{jumps for S on the lenses}
\end{align}
\item[(c)] $S_{1}(z) = 1 + \bigO(z^{-1}) \hspace{0.1cm}$ as $z \to \infty$,

$S_{2}(z) = \bigO(z^{-\theta}) \hspace{0.7cm}$ as $z \to \infty$ in $\mathbb{H}_{\theta}$.
\item[(d)] As $z \to t_{j}$, $z \notin \mathcal{L}$, $j=0,1,\ldots,m,m+1$, we have
\begin{align*}
S_{1}(z) = \bigO(1), \qquad S_{2}(z) = \begin{cases}
\bigO(1)+\bigO((z-t_{j})^{\alpha_{j}}), & \mbox{if } \alpha_{j} \neq 0, \\
\bigO(\log (z-t_{j})), & \mbox{if } \alpha_{j}=0.
\end{cases}
\end{align*}
\end{itemize}

\begin{figure}
\centering
\begin{tikzpicture}
\draw[fill] (0,0) circle (0.05);
\draw (0,0) -- (8,0);

\draw (0,0) .. controls (1.25,1.5) and (2.25,1.5) .. (3.5,0);
\draw (0,0) .. controls (1.25,-1.5) and (2.25,-1.5) .. (3.5,0);
\draw (3.5,0) .. controls (4,1) and (5,1) .. (5.5,0);
\draw (3.5,0) .. controls (4,-1) and (5,-1) .. (5.5,0);
\draw (5.5,0) .. controls (6.25,1.3) and (7.25,1.3) .. (8,0);
\draw (5.5,0) .. controls (6.25,-1.3) and (7.25,-1.3) .. (8,0);

\draw[fill] (3.5,0) circle (0.05);
\draw[fill] (5.5,0) circle (0.05);
\draw[fill] (8,0) circle (0.05);

\node at (0,-0.3) {$a$};
\node at (3.5,-0.3) {$t_{1}$};
\node at (5.55,-0.3) {$t_{2}$};
\node at (8,-0.3) {$b$};

\draw[black,arrows={-Triangle[length=0.18cm,width=0.12cm]}]
(0:1.8) --  ++(0:0.001);
\draw[black,arrows={-Triangle[length=0.18cm,width=0.12cm]}]
(0:4.55) --  ++(0:0.001);
\draw[black,arrows={-Triangle[length=0.18cm,width=0.12cm]}]
(0:6.8) --  ++(0:0.001);

\draw[black,arrows={-Triangle[length=0.18cm,width=0.12cm]}]
(1.8,1.12) --  ++(0:0.001);
\draw[black,arrows={-Triangle[length=0.18cm,width=0.12cm]}]
(1.8,-1.12) --  ++(0:0.001);

\draw[black,arrows={-Triangle[length=0.18cm,width=0.12cm]}]
(4.55,0.76) --  ++(0:0.001);
\draw[black,arrows={-Triangle[length=0.18cm,width=0.12cm]}]
(4.55,-0.76) --  ++(0:0.001);

\draw[black,arrows={-Triangle[length=0.18cm,width=0.12cm]}]
(6.8,0.97) --  ++(0:0.001);
\draw[black,arrows={-Triangle[length=0.18cm,width=0.12cm]}]
(6.8,-0.97) --  ++(0:0.001);

\end{tikzpicture}
\caption{Jump contours for the RH problem for $S$ with $m=2$.}
\label{fig:contour for S}
\end{figure}
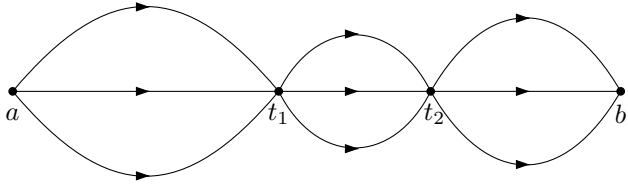

Using \eqref{EL equality}, \eqref{def of g} and \eqref{def of phi}, we see that $\phi$ satisfies 
\begin{align}\label{phi der}
\phi_{\pm}'(x) = g_{\pm}'(x)+\widetilde{g}_{\pm}'(x) = g_{\pm}'(x) - g_{\mp}'(x) = \mp 2\pi i \rho(x), \qquad x \in (a,b).
\end{align}
Since $\rho(x)>0$ for all $x \in (a,b)$, \eqref{phi der} implies by the Cauchy-Riemann equations that there exists a neighborhood of $(a,b)$, denoted $\mathcal{U}'$, such that
\begin{align}\label{real part of phi is bigger than 0}
\re \phi(z) > 0, \qquad \mbox{for all }z \in \mathcal{U}', \; \im z \neq 0.
\end{align}
In the $T \mapsto S$ transformation, we have some freedom in choosing $\sigma_{+},\sigma_{-}$. Now, we use this freedom to require that $\sigma_{+},\sigma_{-} \subset \mathcal{U}'$. By \eqref{jumps for S on the lenses} and \eqref{real part of phi is bigger than 0}, this implies that for any $z \in \sigma_{+}\cup\sigma_{-}$, the jump matrix for $S(z)$ tends to the identity matrix as $n \to +\infty$. This convergence is uniform only for $z\in \sigma_{+}\cup\sigma_{-}$ bounded away from $a,t_{1},\ldots,t_{m},b$. 

\medskip In the next two sections, we construct local and global parametrices for $S$ following the method of \cite{ClaeysWang}. Compared to steepest descent analysis of classical orthogonal polynomials, these steps need to be modified substantially. For example, the construction of the global parametrix relies on the map $J$, and our local parametrices are of a different size than $S$ and therefore are not, strictly speaking, local approximations to $S$ (although they do contain local information about the behavior of $S$). 



\section{Local parametrices and the $S\to P$ transformation}\label{section: local parametrices and S to P transformation}

In this section, we construct local parametrices around $a,t_{1},\ldots,t_{m},b$ and then perform the $S\to P$ transformation, following the method of \cite{ClaeysWang}.

\medskip For each $p \in \{a,t_{1},\ldots,t_{m},b\}$, let $\mathcal{D}_{p}$ be a small open disk centered at $p$. Assume that there exists $\delta \in (0,1)$ independent of $n$ such that
\begin{align}\label{assumption on t1 tm theta}
\min_{1\leq j\neq k \leq m}\{ |t_{j}-t_{k}|,|t_{j}-b|,|t_{j}-a|\} \geq \delta, \qquad \theta \in (\delta,\tfrac{1}{\delta}).
\end{align}
This assumption implies that $\mathcal{U}=\mathcal{U}(\delta)$ can be chosen independently of $\theta$, and that the radii of the disks can be chosen to be $\leq \frac{\delta}{3}$ but independent of $n$ and such that $\mathcal{D}_{p} \subset \mathcal{U}$ for all $p \in \{a,t_{1},\ldots,t_{m},b\}$. 

\subsection{Local parametrix near $t_{k}$, $k=1,\ldots,m$}\label{subsection:local param near tk}
To construct the local parametrix $P^{(t_{k})}$ around $t_{k}$, we use the model RH problem for $\Phi_{\mathrm{HG}}$ from \cite{ItsKrasovsky, DIK, FouMarSou} (the properties of $\Phi_{\mathrm{HG}}$ are also presented in Appendix \ref{subsection: model RH problem for Phi HG}).  Consider the following conformal map
\begin{align*}
f_{t_{k}}(z) = -\begin{cases}
\phi(z)-\phi_{+}(t_{k}), & \im z > 0, \\
-(\phi(z)-\phi_{-}(t_{k})), & \im z < 0,
\end{cases} \qquad z \in \mathcal{D}_{t_{k}}.
\end{align*}
Using \eqref{phi der}, we obtain
\begin{align}\label{expansion of ftk}
f_{t_{k}}(z) = 2\pi i \rho(t_{k}) (z-t_{k})(1+\bigO(z-t_{k})), \qquad \mbox{as } z \to t_{k}.
\end{align}
In a small neighborhood of $t_{k}$, we deform the lenses $\sigma_{+}$ and $\sigma_{-}$ such that 
\begin{align*}
f_{t_{k}}(\sigma_{+}\cap \mathcal{D}_{t_{k}}) \subset \Gamma_{4}\cup \Gamma_{2}, \qquad f_{t_{k}}(\sigma_{-}\cap \mathcal{D}_{t_{k}}) \subset \Gamma_{6}\cup \Gamma_{8},
\end{align*}
where $\Gamma_{4}, \Gamma_{2}, \Gamma_{6}, \Gamma_{8}$ are the contours shown in Figure \ref{Fig:HG}.
The local parametrix is defined by
\begin{align}\label{Ptk explicit}
& P^{(t_{k})}(z) = \Phi_{\mathrm{HG}}(n f_{t_{k}}(z);\alpha_{k},\beta_{k}) \widetilde{W}_{k}(z)^{-\sigma_{3}} \bigg( \frac{\omega_{t_{k}}(z)e^{W(z)}}{\theta z^{\theta-1}} \bigg)^{-\frac{\sigma_{3}}{2}} e^{-\frac{n\phi(z)}{2}\sigma_{3}},
\end{align}
where
\begin{equation}\label{def of Wk}
\omega_{t_{k}}(z) = \frac{\omega(z)}{ \omega_{\alpha_{k}}(z) \omega_{\beta_{k}}(z)}, \qquad  \widetilde{W}_{k}(z) = \left\{ \begin{array}{l l}
(z-t_{k})^{\frac{\alpha_{k}}{2}}e^{-\frac{ i \pi\alpha_{k}}{2}}, & z \in Q_{+,k}^{R}, \\
(z-t_{k})^{\frac{\alpha_{k}}{2}}, & z \in Q_{+,k}^{L}, \\
(z-t_{k})^{\frac{\alpha_{k}}{2}}, & z \in Q_{-,k}^{L}, \\
(z-t_{k})^{\frac{\alpha_{k}}{2}}e^{\frac{ i \pi\alpha_{k}}{2}}, & z \in Q_{-,k}^{R}, \\
\end{array} \right.
\end{equation}
and $Q_{+,k}^{R}$, $Q_{+,k}^{L}$, $Q_{-,k}^{L}$, $Q_{-,k}^{R}$ are the preimages by $f_{t_{k}}$ of the four quadrants:
\begin{align*}
& Q_{\pm,k}^{R} = \{ z \in \mathcal{D}_{t_{k}}: \mp \re f_{t_{k}}(z) > 0 \mbox{, } \im f_{t_{k}}(z) >0 \}, \\
& Q_{\pm,k}^{L} = \{ z \in \mathcal{D}_{t_{k}}: \mp \re f_{t_{k}}(z) > 0 \mbox{, } \im f_{t_{k}}(z) <0 \}.
\end{align*}
Using the jumps \eqref{jumps PHG3} for $\Phi_{\mathrm{HG}}$, it is easy to verify that $P^{(t_{k})}$ and $S$ have the same jumps inside $\mathcal{D}_{t_{k}}$, which implies that $S(P^{(t_{k})})^{-1}$ is analytic in $\mathcal{D}_{t_{k}}\setminus \{t_{k}\}$. Furthermore, the RH condition (d) of the RH problem for $S$ and \eqref{lol 35} imply that the singularity at $t_{k}$ is removable, so that $S(P^{(t_{k})})^{-1}$ is in fact analytic in the whole disk $\mathcal{D}_{t_{k}}$. We end this section with an analysis that will be useful in Section \ref{subsection: P tranformation}. Let us consider
\small
\begin{align}\label{def of Etk}
& \hspace{-0.25cm} E_{t_{k}}(z) = \hspace{-0.08cm} \bigg( \frac{\omega_{t_{k}}(z)e^{W(z)}}{\theta z^{\theta-1}} \bigg)^{\hspace{-0.08cm}\frac{\sigma_{3}}{2}} \hspace{-0.08cm} \widetilde{W}_{k}(z)^{\sigma_{3}}\hspace{-0.08cm} \left\{ \hspace{-0.18cm} \begin{array}{l l}
e^{ \frac{i\pi\alpha_{k}}{4}\sigma_{3}}e^{-i\pi\beta_{k} \sigma_{3}}, \hspace{-0.2cm} & \hspace{-0.25cm} z \in Q_{+,k}^{R} \\
e^{-\frac{i\pi\alpha_{k}}{4}\sigma_{3}}e^{-i\pi\beta_{k}\sigma_{3}}, \hspace{-0.2cm} & \hspace{-0.25cm} z \in Q_{+,k}^{L} \\
e^{\frac{i\pi\alpha_{k}}{4}\sigma_{3}}\begin{pmatrix}
0 & 1 \\ -1 & 0
\end{pmatrix} , \hspace{-0.2cm} & \hspace{-0.25cm} z \in Q_{-,k}^{L} \\
e^{-\frac{i\pi\alpha_{k}}{4}\sigma_{3}}\begin{pmatrix}
0 & 1 \\ -1 & 0
\end{pmatrix} , \hspace{-0.2cm} & \hspace{-0.25cm} z \in Q_{-,k}^{R} \\
\end{array} \hspace{-0.2cm} \right\} \hspace{-0.08cm} e^{\frac{n\phi_{+}(t_{k})}{2}\sigma_{3}} (nf_{t_{k}}(z))^{\beta_{k}\sigma_{3}}\hspace{-0.05cm}.
\end{align}
\normalsize
Note that $E_{t_{k}}$ is analytic in $\mathcal{D}_{t_{k}}\setminus (a,b)$ (see \eqref{jump of Etk} below for its jump relations) and is such that
\begin{align*}
\mathrm{E}_{t_{k}}(z) := E_{t_{k}}(z)_{11}(z-t_{k})^{-(\beta_{k}+ \frac{\alpha_{k}}{2})}, \qquad z \in Q_{+,k}^{R},
\end{align*}
remains bounded as $z \to t_{k}$, $z \in Q_{+,k}^{R}$. Let $J_{P}(z) := E_{t_{k}}(z)P^{(t_{k})}(z)$ for $z \in \partial \mathcal{D}_{t_{k}}$. Using \eqref{Asymptotics HG}, as $n \to +\infty$ we obtain
\begin{align}\label{comp tk 1}
J_{P}(z) = I + \frac{v_{k}}{n f_{t_{k}}(z)} E_{t_{k}}(z) \begin{pmatrix}
-1 & \tau(\alpha_{k},\beta_{k}) \\ - \tau(\alpha_{k},-\beta_{k}) & 1
\end{pmatrix}E_{t_{k}}(z)^{-1} + \bigO (n^{-2+2|\re\beta_{k}|}),
\end{align}
uniformly for $z \in \partial \mathcal{D}_{t_{k}}$, where $v_{k} = \beta_{k}^{2}-\frac{\alpha_{k}^{2}}{4}$ and $\tau(\alpha_{k},\beta_{k})$ is defined in \eqref{def of tau}. For $z \in Q_{+,k}^{R}$, we have $E_{t_{k}}(z)=\mathrm{E}_{t_{k}}(z)^{\sigma_{3}}(z-t_{k})^{(\frac{\alpha_{k}}{2}+\beta_{k})\sigma_{3}}$, and thus \eqref{comp tk 1} implies 
\begin{align}
J_{P}(z) = & \; I + \frac{v_{k}}{n f_{t_{k}}(z)}  \begin{pmatrix}
-1 & \hspace{-0.3cm}\tau(\alpha_{k},\beta_{k})\mathrm{E}_{t_{k}}(z)^{2}(z-t_{k})^{\alpha_{k}+2\beta_{k}} \\  \frac{-\tau(\alpha_{k},-\beta_{k})}{\mathrm{E}_{t_{k}}(z)^{2}(z-t_{k})^{\alpha_{k}+2\beta_{k}}} & \hspace{-0.3cm}1
\end{pmatrix} \nonumber \\
& + \bigO (n^{-2+2|\re\beta_{k}|}), \label{comp tk 2}
\end{align}
as $n \to +\infty$ uniformly for $z \in \partial \mathcal{D}_{t_{k}} \cap Q_{+,k}^{R}$. Note also that $\mathrm{E}(t_{k})^{2}=\mathrm{E}(t_{k};n)^{2}$ is given by
\begin{align}\label{mathrmE at tk}
\mathrm{E}(t_{k})^{2} := \lim_{z \to t_{k}, z \in Q_{+,k}^{R}}\mathrm{E}(z)^{2} = \frac{\omega_{t_{k}}(t_{k})e^{W(t_{k})}}{\theta t_{k}^{\theta-1}} e^{ -\frac{i\pi\alpha_{k}}{2}}e^{-i\pi \beta_{k}} e^{n\phi_{+}(t_{k})} (n2\pi \rho(t_{k}))^{2\beta_{k}}.
\end{align}

\subsection{Local parametrix near $b$}\label{subsection: local param near b}
Inside the disk $\mathcal{D}_{b}$, the local parametrix $P^{(b)}$ is built out of a model RH problem whose solution $\Phi_{\mathrm{Be}}$ is expressed in terms of Bessel functions. This RH problem is well known \cite{KMcLVAV}, and for convenience it is also presented in Appendix \ref{ApB}. Define $\psi$ by 
\begin{align*}
\rho(x) = \frac{\psi(x)}{\sqrt{x-a}\sqrt{b-x}}, \qquad x \in (a,b).
\end{align*}
By \eqref{asymptotics of rho near b}--\eqref{asymptotics of rho near a}, $\psi$ is well-defined at $a$ and $b$. Define
\begin{align*}
f_{b}(z) = \phi(z)^{2}/16.
\end{align*}
Using \eqref{phi der}, we obtain
\begin{align}\label{asymp of conformal map near b}
& f_{b}(z) = f_{b}^{(0)} (z-b) \big( 1+ \bigO(z-b) \big) \quad \mbox{as } z \to b, \quad \mbox{where } f_{b}^{(0)}=\bigg( \frac{\pi \psi(b)}{\sqrt{b-a}} \bigg)^{2}.
\end{align}
In a small neighborhood of $b$, we deform the lenses such that they are mapped through $f_{b}$ on a subset of $\Sigma_{\mathrm{Be}}$ (see Figure \ref{fig:Bessel}). More precisely, we require that 
\begin{align*}
f_{b}(\sigma_{+}\cap \mathcal{D}_{b}) \subset e^{\frac{2\pi i}{3}}(0,+\infty), \qquad f_{b}(\sigma_{-}\cap \mathcal{D}_{b}) \subset e^{-\frac{2\pi i}{3}}(0,+\infty).
\end{align*}
We define the local parametrix by
\begin{align}\label{Pb explicit}
& P^{(b)}(z) = \Phi_{\mathrm{Be}}(n^{2}f_{b}(z);\alpha_{m+1}) \bigg( \frac{\omega_{b}(z)e^{W(z)}}{\theta z^{\theta-1}} \bigg)^{-\frac{\sigma_{3}}{2}} \hspace{-0.1cm} e^{-\frac{n\phi(z)}{2}\sigma_{3}}(z-b)^{-\frac{\alpha_{m+1}}{2}\sigma_{3}},
\end{align}
where $\omega_{b}(z) := \omega(z)/(b-x)^{\alpha_{m+1}}$ and the principal branches for the roots are taken. Using \eqref{Jump for P_Be}, one verifies that $S(P^{(b)})^{-1}$ is analytic in $\mathcal{D}_{b}\setminus \{b\}$. By \eqref{local behaviour near 0 of P_Be}, the singularity of $S(P^{(b)})^{-1}$ at $b$ is removable, which implies that $S(P^{(b)})^{-1}$ is in fact analytic in the whole disk $\mathcal{D}_{b}$. It will also be convenient to consider the following function 
\begin{align}\label{def of Eb}
& E_{b}(z) = \bigg( \frac{\omega_{b}(z)e^{W(z)}}{\theta z^{\theta-1}} \bigg)^{\frac{\sigma_{3}}{2}} (z-b)^{\frac{\alpha_{m+1}}{2}\sigma_{3}} A^{-1}(2\pi n f_{b}(z)^{1/2})^{\frac{\sigma_{3}}{2}}, \quad A := \frac{1}{\sqrt{2}}\begin{pmatrix}
1 & i \\ i & 1
\end{pmatrix}.
\end{align}
It can be verified that $E_{b}$ is analytic in $\mathcal{D}_{b}\setminus [a,b]$ (the jumps of $E_{b}$ are given in \eqref{jump of Etk} below). For $z \in \partial \mathcal{D}_{b}$, let $J_{P}(z) := E_{b}(z)P^{(b)}(z)$. Using \eqref{large z asymptotics Bessel}, we obtain
\begin{align}
J_{P}(z) & = I + \frac{1}{16n f_{b}(z)^{1/2}}  \begin{pmatrix}
-(1+4\alpha_{m+1}^{2}) & \ds \hspace{-0.5cm} -2i  \frac{\omega_{b}(z)e^{W(z)}}{\theta z^{\theta-1}}(z-b)^{\alpha_{m+1}}  \\ \ds -2i \frac{\theta z^{\theta-1}}{\omega_{b}(z)e^{W(z)}} (z-b)^{-\alpha_{m+1}} & \hspace{-0.5cm} 1+4\alpha_{m+1}^{2}
\end{pmatrix} \nonumber \\
& + \bigO(n^{-2}), \label{jumps for P on Db}
\end{align}
as $n \to +\infty$ uniformly for $z \in \partial \mathcal{D}_{b}$.

\subsection{Local parametrix near $a$}\label{subsection: local param near a}
The construction of the local parametrix $P^{(a)}$ inside $\mathcal{D}_{a}$ is similar to that of $P^{(b)}$ and also relies on the model RH problem $\Phi_{\mathrm{Be}}$. Define
\begin{align*}
f_{a}(z) = -(\phi(z)-2\pi i)^{2}/16.
\end{align*}
As $z \to a$, using \eqref{phi der} we get
\begin{align*}
& f_{a}(z) = f_{a}^{(0)} (z-a) \big( 1+ \bigO(z-a) \big), \qquad \mbox{where } f_{a}^{(0)}=\bigg( \frac{\pi \psi(a)}{\sqrt{b-a}} \bigg)^{2}.
\end{align*}
In a small neighborhood of $a$, we choose $\sigma_{+}$ and $\sigma_{-}$ such that
\begin{align*}
-f_{a}(\sigma_{+}\cap \mathcal{D}_{a}) \subset e^{-\frac{2\pi i}{3}}(0,+\infty), \qquad -f_{a}(\sigma_{-}\cap \mathcal{D}_{a}) \subset e^{\frac{2\pi i}{3}}(0,+\infty).
\end{align*}
The local parametrix $P^{(a)}$ is defined by
\begin{align}\label{Pa explicit}
& \hspace{-0.15cm} P^{(a)}(z) = \sigma_{3}\Phi_{\mathrm{Be}}(-n^{2}f_{a}(z);\alpha_{0})\sigma_{3} \bigg( \frac{\omega_{a}(z)e^{W(z)}}{\theta z^{\theta-1}} \bigg)^{-\frac{\sigma_{3}}{2}} \hspace{-0.15cm} e^{-\frac{n\phi(z)}{2}\sigma_{3}}(a-z)^{-\frac{\alpha_{0}}{2}\sigma_{3}},
\end{align}
where $\omega_{a}(z):= \omega(z)/(x-a)^{\alpha_{0}}$ and the principal branches are taken for the roots. Like in Section \ref{subsection: local param near b}, using \eqref{Jump for P_Be} and \ref{local behaviour near 0 of P_Be} one verifies that $S(P^{(a)})^{-1}$ is analytic in the whole disk $\mathcal{D}_{a}$. It is will also be useful to define
\begin{align}\label{def of Ea}
E_{a}(z) = (-1)^{n}\bigg( \frac{\omega_{a}(z)e^{W(z)}}{\theta z^{\theta-1}} \bigg)^{\frac{\sigma_{3}}{2}} (a-z)^{\frac{\alpha_{0}}{2}\sigma_{3}} A(2\pi n (-f_{a}(z))^{1/2})^{\frac{\sigma_{3}}{2}}.
\end{align}
Note that $E_{a}$ is analytic in $\mathcal{D}_{a}\setminus [a,b]$ (the jumps of $E_{a}$ are stated in \eqref{jump of Etk} below). For $z \in \partial \mathcal{D}_{a}$, let $J_{P}(z) := E_{a}(z)P^{(a)}(z)$. Using \eqref{large z asymptotics Bessel}, we get
\begin{align}
J_{P}(z) & = I + \frac{1}{16n (-f_{a}(z))^{1/2}}   \begin{pmatrix}
-(1+4\alpha_{0}^{2}) & \ds \hspace{-0.4cm} 2i \frac{\omega_{a}(z)e^{W(z)}}{\theta z^{\theta-1}} (a-z)^{\alpha_{0}} \\ \ds 2i \frac{\theta z^{\theta-1}}{\omega_{a}(z)e^{W(z)}} (a-z)^{-\alpha_{0}} & \hspace{-0.4cm} 1+4\alpha_{0}^{2}
\end{pmatrix} \nonumber \\
& + \bigO(n^{-2}), \label{jumps for P on Da}
\end{align}
as $n \to \infty$ uniformly for $z \in \partial \mathcal{D}_{a}$.
\subsection{Third transformation $S \mapsto P$}\label{subsection: P tranformation}
Define
\begin{align}\label{S to T transformation}
\hspace{-0.15cm} P(z) = \begin{cases}
S(z), & \hspace{-0.15cm} z \in \mathbb{C}\setminus(\bigcup_{j=0}^{m+1} \mathcal{D}_{t_{j}} \cup [a,b]\cup \sigma_{+}\cup \sigma_{-}), \\
S(z) \Big( E_{t_{k}}(z)P^{(t_{k})}(z) \Big)^{-1}, & \hspace{-0.15cm} z \in \mathcal{D}_{t_{k}}\setminus ([a,b]\cup \sigma_{+}\cup \sigma_{-}), 
\end{cases}
\end{align}
where $k=0,1,\ldots,m,m+1$ and we recall that $t_{0}:=a$ and $t_{m+1}:=b$. It follows from the analysis of Sections \ref{subsection:local param near tk}--\ref{subsection: local param near a} that for each $k \in \{0,1,\ldots,m+1\}$, $S(z)P^{(t_{k})}(z)^{-1}$ is analytic in $\mathcal{D}_{t_{k}}$ and that $E_{t_{k}}$ is analytic in $\mathcal{D}_{t_{k}}\setminus [a,b]$. Hence, $P$ has no jumps on $(\sigma_{+}\cup \sigma_{-})\cap\bigcup_{k=0}^{m+1}\mathcal{D}_{t_{k}}$, and therefore $(P_{1},P_{2})$ is analytic in $(\mathbb{C}\setminus \Sigma_{P},\mathbb{H}_{\theta}\setminus \Sigma_{P})$, where
\begin{align}\label{def of Sigma P}
\Sigma_{P} := \Big( (\sigma_{+}\cup \sigma_{-})\setminus \bigcup_{j=0}^{m+1}\mathcal{D}_{t_{j}} \Big) \cup \bigcup_{j=0}^{m+1}\partial\mathcal{D}_{t_{j}} \cup [a,b].
\end{align} 
Furthermore, for each $j \in \{0,\ldots,m+1\}$, the jumps of $P$ on $[a,b]\cap \mathcal{D}_{t_{j}}$ are identical to those of $E_{t_{j}}$. These jumps can be obtained using \eqref{def of Etk}, \eqref{def of Eb} and \eqref{def of Ea}: for all $j \in \{0,1,\ldots,m,m+1\}$ we find
\begin{align}\label{jump of Etk}
E_{t_{j},+}(z)^{-1} = E_{t_{j},-}(z)^{-1}\begin{pmatrix}
0 & \frac{\omega(z)e^{W(z)}}{\theta z^{\theta-1}} \\
-\frac{\theta z^{\theta-1}}{\omega(z)e^{W(z)}} & 0
\end{pmatrix}, \qquad z \in (a,b)\cap \mathcal{D}_{t_{j}}.
\end{align}
For convenience, for each $j\in\{0,\ldots,m+1\}$ the orientation of $\partial\mathcal{D}_{t_{j}}$ is chosen to be clockwise, as shown in Figure \ref{fig:contour for P for Gaussian-type}. The properties of $P$ are summarized in the following RH problem.

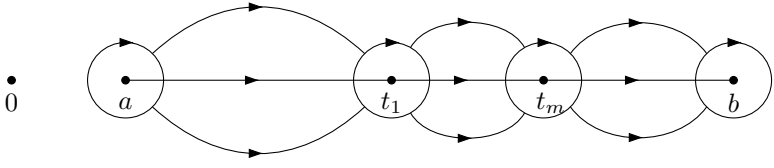
\begin{figure}
\centering
\begin{tikzpicture}
\draw[fill] (0,0) circle (0.05);
\draw (0,0) -- (8,0);
\draw[fill] (-1.5,0) circle (0.05);

\draw (45:0.5) .. controls (1.25,1.18) and (2.25,1.18) .. ($(3.5,0)+(135:0.5)$);
\draw (-45:0.5) .. controls (1.25,-1.18) and (2.25,-1.18) .. ($(3.5,0)+(-135:0.5)$);
\draw ($(3.5,0)+(60:0.5)$) .. controls (4,0.88) and (5,0.88) .. ($(5.5,0)+(120:0.5)$);
\draw ($(3.5,0)+(-60:0.5)$) .. controls (4,-0.88) and (5,-0.88) .. ($(5.5,0)+(-120:0.5)$);
\draw ($(5.5,0)+(45:0.5)$) .. controls (6.25,0.9) and (7.25,0.9) .. ($(8,0)+(135:0.5)$);
\draw ($(5.5,0)+(-45:0.5)$) .. controls (6.25,-0.9) and (7.25,-0.9) .. ($(8,0)+(-135:0.5)$);

\draw[fill] (3.5,0) circle (0.05);
\draw[fill] (5.5,0) circle (0.05);
\draw[fill] (8,0) circle (0.05);

\draw (0,0) circle (0.5);
\draw (3.5,0) circle (0.5);
\draw (5.5,0) circle (0.5);
\draw (8,0) circle (0.5);

\node at (-1.5,-0.3) {$0$};
\node at (0,-0.3) {$a$};
\node at (3.5,-0.3) {$t_{1}$};
\node at (5.6,-0.3) {$t_{m}$};
\node at (8,-0.3) {$b$};

\draw[black,arrows={-Triangle[length=0.18cm,width=0.12cm]}]
(0:1.75) --  ++(0:0.001);
\draw[black,arrows={-Triangle[length=0.18cm,width=0.12cm]}]
(0:4.5) --  ++(0:0.001);
\draw[black,arrows={-Triangle[length=0.18cm,width=0.12cm]}]
(0:6.75) --  ++(0:0.001);

\draw[black,arrows={-Triangle[length=0.18cm,width=0.12cm]}]
(1.8,0.97) --  ++(0:0.001);
\draw[black,arrows={-Triangle[length=0.18cm,width=0.12cm]}]
(1.8,-0.97) --  ++(0:0.001);

\draw[black,arrows={-Triangle[length=0.18cm,width=0.12cm]}]
(4.55,0.77) --  ++(0:0.001);
\draw[black,arrows={-Triangle[length=0.18cm,width=0.12cm]}]
(4.55,-0.77) --  ++(0:0.001);

\draw[black,arrows={-Triangle[length=0.18cm,width=0.12cm]}]
(6.8,0.76) --  ++(0:0.001);
\draw[black,arrows={-Triangle[length=0.18cm,width=0.12cm]}]
(6.8,-0.76) --  ++(0:0.001);

\draw[black,arrows={-Triangle[length=0.18cm,width=0.12cm]}]
(0.1,0.5) --  ++(0:0.001);
\draw[black,arrows={-Triangle[length=0.18cm,width=0.12cm]}]
(3.6,0.5) --  ++(0:0.001);
\draw[black,arrows={-Triangle[length=0.18cm,width=0.12cm]}]
(5.6,0.5) --  ++(0:0.001);
\draw[black,arrows={-Triangle[length=0.18cm,width=0.12cm]}]
(8.1,0.5) --  ++(0:0.001);

\end{tikzpicture}
\caption{Jump contours $\Sigma_{P}$ with $m=2$.}
\label{fig:contour for P for Gaussian-type}
\end{figure}

\subsubsection*{RH problem for $P$}
\begin{itemize}
\item[(a)] $(P_{1},P_{2})$ is analytic in $(\mathbb{C}\setminus \Sigma_{P},\mathbb{H}_{\theta}^{c}\setminus \Sigma_{P})$.
\item[(b)] For $z \in \Sigma_{P}$, we have $P_{+}(z)=P_{-}(z)J_{P}(z)$, where
\begin{align*}
& J_{P}(z) = \begin{pmatrix}
1 & 0 \\ e^{-n\phi(z)} \frac{\theta z^{\theta-1}}{\omega(z)e^{W(z)}} & 1
\end{pmatrix}, & & z \in (\sigma_{+}\cup \sigma_{-})\setminus \bigcup_{j=0}^{m+1}\mathcal{D}_{t_{j}}, \\
& J_{P}(z) = \begin{pmatrix}
0 & \frac{\omega(z)e^{W(z)}}{\theta z^{\theta-1}} \\
-\frac{\theta z^{\theta-1}}{\omega(z)e^{W(z)}} & 0
\end{pmatrix}, & & z \in (a,b)\setminus \{t_{1},\ldots,t_{m}\}, \\
& J_{P}(z) = E_{t_{j}}(z)P^{(t_{j})}(z), & & z \in \partial \mathcal{D}_{t_{j}}, \; j\in \{0,1,\ldots,m,m+1\}. 
\end{align*}
\item[(c)] $P_{1}(z) = 1 + \bigO(z^{-1}) \hspace{0.1cm}$ as $z \to \infty$,

$P_{2}(z) = \bigO(z^{-\theta}) \hspace{0.7cm}$ as $z \to \infty$ in $\mathbb{H}_{\theta}$.
\item[(d)] As $z \to t_{j}$, $z \notin \mathcal{L}$, $\im z >0$, $j=0,m+1$, we have 
\begin{align*}
(P_{1}(z),P_{2}(z))=(\bigO((z-t_{j})^{-\frac{1}{4}}),\bigO((z-t_{j})^{-\frac{1}{4}}))(z-t_{j})^{-\frac{\alpha_{j}}{2}\sigma_{3}}.
\end{align*}
As $z \to t_{j}$, $z \notin \mathcal{L}$, $\im z > 0$, $j=1,\ldots,m$, we have
\begin{align*}
(P_{1}(z),P_{2}(z)) = (\bigO(1),\bigO(1))(z-t_{j})^{-(\frac{\alpha_{j}}{2}+\beta_{j})\sigma_{3}}.
\end{align*}
\end{itemize}
By \eqref{real part of phi is bigger than 0} and the fact that $\sigma_{+},\sigma_{-} \subset \mathcal{U}'$, as $n \to + \infty$ we have
\begin{align}\label{JP estimate on the lenses}
& J_{P}(z) = I+\bigO(e^{-cn}), & & \mbox{uniformly for } z \in (\sigma_{+}\cup \sigma_{-})\setminus \bigcup_{j=0}^{m+1}\mathcal{D}_{t_{j}},
\end{align}
for a certain $c>0$. Also, it follows from \eqref{comp tk 1}, \eqref{jumps for P on Db} and \eqref{jumps for P on Da} that as $n \to + \infty$, 
\begin{align}
& \hspace{-0.15cm} J_{P}(z) = I+J_{P}^{(1)}(z)n^{-1}+\bigO(n^{-2}), & & \hspace{-0.2cm} \mbox{unif. for } z \in \partial \mathcal{D}_{a} \cup \partial \mathcal{D}_{b}, \label{JP estimate on the disks a and b} \\
& \hspace{-0.15cm} J_{P}(z) = I+J_{P}^{(1)}(z)n^{-1}+\bigO(n^{-2+2|\re \beta_{j}|}), & & \hspace{-0.2cm} \mbox{unif. for } z \in \partial \mathcal{D}_{t_{j}}, \; j=1,...,m, \label{JP estimate on the disks tj}
\end{align}
where $J_{P}^{(1)}(z) = \bigO(1)$ for $z \in \partial \mathcal{D}_{a} \cup \partial \mathcal{D}_{b}$ and $J_{P}^{(1)}(z) = \bigO(n^{2|\re \beta_{j}|})$ for $z \in \partial \mathcal{D}_{t_{j}}$, $j=1,\ldots,m$. If the parameters $t_{1},\ldots,t_{m}$ and $\theta$ vary with $n$ in such a way that they satisfy \eqref{assumption on t1 tm theta} for a certain $\delta \in (0,1)$, then, as explained at the beginning of Section \ref{section: local parametrices and S to P transformation}, the radii of the disks can be chosen independently of $n$ and therefore the estimates \eqref{JP estimate on the lenses}--\eqref{JP estimate on the disks tj} hold uniformly in $t_{1},\ldots,t_{m},\theta$. It also follows from the explicit expressions of $E_{t_{j}}$ and $P^{(t_{j})}$, $j=0,1,\ldots,m+1$ given by \eqref{Ptk explicit}, \eqref{def of Etk}, \eqref{Pb explicit}, \eqref{def of Eb}, \eqref{Pa explicit}, \eqref{def of Ea} that the estimates \eqref{JP estimate on the lenses}--\eqref{JP estimate on the disks tj} hold uniformly for $\alpha_{0},\ldots,\alpha_{m+1}$ in compact subsets of $\{z \in \mathbb{C}: \re z >-1\}$, and uniformly for $\beta_{1},\ldots,\beta_{m}$ in compact subsets of $\{z \in \mathbb{C}: \re z \in (-\frac{1}{2},\frac{1}{2})\}$.\footnote{The restriction $\re \beta_{j} \in (\frac{-1}{4},\frac{1}{4})$ appearing in Theorem \ref{thm:rigidity} will be important in Section \ref{section:small norm shifted RHP}.} 

\section{Global parametrix}\label{section: global parametrix}
The following RH problem for $P^{(\infty)}$ is obtained from the RH problem for $P$ by disregarding the jumps of $P$ on the lenses and on the boundaries of the disks. In view of \eqref{JP estimate on the lenses}--\eqref{JP estimate on the disks tj}, one expects that $P^{(\infty)}$ will be a good approximation to $P$ as $n \to + \infty$. 
\subsubsection*{RH problem for $P^{(\infty)}$}
\begin{itemize}
\item[(a)] $P^{(\infty)}=(P_{1}^{(\infty)},P_{2}^{(\infty)})$ is analytic in $(\mathbb{C}\setminus [a,b],\mathbb{H}_{\theta}^{c}\setminus [a,b])$.
\item[(b)] The jumps are given by
\begin{align*}
& P_{+}^{(\infty)}(z) = P_{-}^{(\infty)}(z)\begin{pmatrix}
0 & \frac{\omega(z)e^{W(z)}}{\theta z^{\theta-1}} \\
-\frac{\theta z^{\theta-1}}{\omega(z)e^{W(z)}} & 0
\end{pmatrix}, & & z \in (a,b)\setminus \{t_{1},\ldots,t_{m}\}.
\end{align*}
\item[(c)] $P_{1}^{(\infty)}(z) = 1 + \bigO(z^{-1}) \hspace{0.1cm}$ as $z \to \infty$,

$P_{2}^{(\infty)}(z) = \bigO(z^{-\theta}) \hspace{0.7cm}$ as $z \to \infty$ in $\mathbb{H}_{\theta}$.

\item[(d)] As $z \to t_{j}$, $\im z >0$, $j=0,m+1$, we have 
\begin{align*}
(P_{1}^{(\infty)}(z),P_{2}^{(\infty)}(z))=(\bigO((z-t_{j})^{-\frac{1}{4}}),\bigO((z-t_{j})^{-\frac{1}{4}}))(z-t_{j})^{-\frac{\alpha_{j}}{2}\sigma_{3}}.
\end{align*}
As $z \to t_{j}$, $\im z > 0$, $j=1,\ldots,m$, we have
\begin{align*}
(P_{1}^{(\infty)}(z),P_{2}^{(\infty)}(z)) = (\bigO(1),\bigO(1))(z-t_{j})^{-(\frac{\alpha_{j}}{2}+\beta_{j})\sigma_{3}}.
\end{align*}
\end{itemize}
To construct a solution to this RH problem, we follow the strategy of \cite{ClaeysWang} and use the mapping $J$ to transform $P^{(\infty)}$ into a scalar RH problem. Recall that $J$ is defined in \eqref{def of the map J} with $c_{0}>c_{1}>0$ such that \eqref{system that defines c0 and c1} holds, and that some properties of $J$ are stated in Proposition \ref{prop:ClaeysRomano}. We define a function $F$ on $\mathbb{C}\setminus (\gamma_{1}\cup \gamma_{2} \cup [-1,0])$ by
\begin{align*}
F(s) = \begin{cases}
P_{1}^{(\infty)}(J(s)), & s \in \mathbb{C}\setminus \overline{D}, \\
P_{2}^{(\infty)}(J(s)), & s \in D\setminus [-1,0].
\end{cases}
\end{align*}
Note that $P^{(\infty)}$ can be recovered from $F$ via the formulas
\begin{align}
& P_{1}^{(\infty)}(z) = F(I_{1}(z)), & & z \in \mathbb{C}\setminus [a,b], \label{P1inf in terms of F} \\
& P_{2}^{(\infty)}(z) = F(I_{2}(z)), & & z \in \mathbb{H}_{\theta}\setminus [a,b]. \label{P2inf in terms of F}
\end{align}
We make the following observations:
\begin{align*}
& \mbox{(i) }\hspace{0.2cm}P^{(\infty)}_{2}(e^{\frac{\pi i}{\theta}}x) = P^{(\infty)}_{2}(e^{-\frac{\pi i}{\theta}}x) \mbox{ for } x>0 \mbox{ implies that $F$ is analytic on $(-1,0)$}, \\
& \mbox{(ii)\hspace{0.1cm} $P_{2}^{(\infty)}(z) = \bigO(1)$ as $z\to 0$ implies that $F(s)$ remains bounded at $s=-1$}, \\
& \mbox{(iii) $P_{2}^{(\infty)}(z) = \bigO(z^{-\theta}) $ as $z \to \infty$ implies that $F(s)$ has a simple zero at $s=0$.}
\end{align*}
With $\gamma_{1}$ and $\gamma_{2}$ both oriented from $s_{a}$ to $s_{b}$, we have
\begin{align*}
& F_{+}(s) = P_{1,+}^{(\infty)}(J(s)), & & F_{-}(s) = P_{2,-}^{(\infty)}(J(s)), & & s \in \gamma_{1}, \\
& F_{+}(s) = P_{2,+}^{(\infty)}(J(s)), & & F_{-}(s) = P_{1,-}^{(\infty)}(J(s)), & & s \in \gamma_{2},
\end{align*}
and therefore $F$ satisfies the following RH problem.
\subsubsection*{RH problem for $F$}
\begin{itemize}
\item[(a)] $F$ is analytic in $\mathbb{C}\setminus (\gamma_{1}\cup \gamma_{2})$.
\item[(b)] $F_{+}(s) = -\frac{\theta J(s)^{\theta-1}}{\omega(J(s))e^{W(J(s))}} F_{-}(s) \hspace{0.2cm}$ for $s \in \gamma_{1}$,

$F_{+}(s) = \frac{\omega(J(s))e^{W(J(s))}}{\theta J(s)^{\theta-1}} F_{-}(s) \hspace{0.5cm}$ for $s \in \gamma_{2}$.
\item[(c)] $F(s) = 1+\bigO(s^{-1}) \hspace{2.05cm}$ as $s \to \infty$,

$F(s) = \bigO(s) \hspace{3.03cm}$ as $s \to 0$,

$F(s) = \bigO((s-s_{a})^{-\frac{1}{2}-\alpha_{0}}) \hspace{1.02cm}$ as $s \to s_{a}$, $s \in \mathbb{C} \setminus \overline{D}$,

$F(s) = \bigO((s-s_{b})^{-\frac{1}{2}-\alpha_{m+1}}) \hspace{0.67cm}$ as $s \to s_{b}$, $s \in \mathbb{C} \setminus \overline{D}$,

$F(s) = \bigO((s-I_{1,+}(t_{j}))^{-\frac{\alpha_{j}}{2}-\beta_{j}}) \hspace{0.1cm}$ as $s \to I_{1,+}(t_{j})$, $s \in \mathbb{C} \setminus \overline{D}$, $j=1,\ldots,m$,

$F(s) = \bigO((s-I_{2,+}(t_{j}))^{\frac{\alpha_{j}}{2}+\beta_{j}}) \hspace{0.32cm}$ as $s \to I_{2,+}(t_{j})$, $s \in D$, $j=1,\ldots,m$.
\end{itemize}
The jumps of this RH problem can be simplified via the transformation
\begin{align}\label{F to G transformation}
G(s) = F(s)\sqrt{(s-s_{a})(s-s_{b})},
\end{align}
where the square root is discontinuous along $\gamma_{1}$ and behaves as $s+\bigO(1)$ as $s \to \infty$. Indeed, using \eqref{F to G transformation} and the jumps for $F$, it is easily seen that
\begin{align}\label{boundary value of G}
G_{+}(s) = \frac{\omega(J(s))e^{W(J(s))}}{\theta J(s)^{\theta-1}} G_{-}(s), \qquad s \in \gamma,
\end{align}
where the boundary values of $G$ in \eqref{boundary value of G} are taken with respect to the orientation of $\gamma$, which we recall is oriented in the counterclockwise direction.\footnote{Thus $\gamma \cap \{z:\im z >0\}$ and $\gamma_{1}$ have opposite orientations, while $\gamma \cap \{z:\im z <0\}$ and $\gamma_{2}$ have the same orientation.} Noting that
\begin{align*}
& \frac{1}{\theta J(s)^{\theta-1}} = \frac{1}{\theta (c_{1}s+c_{0})^{\theta-1}}  \bigg(\frac{s}{s+1}\bigg)^{\frac{\theta-1}{\theta}}, & & s \in \gamma, \quad c_{0}>c_{1}>1,
\end{align*}
we define
\begin{align}\label{G to H transformation}
G(s) = H(s)\begin{cases}
\ds s\Big(\frac{s+1}{s}\Big)^{\frac{\theta-1}{\theta}}, & \ds s \in \mathbb{C}\setminus \overline{D}, \\[0.2cm]
\ds \frac{s}{\theta (c_{1}s+c_{0})^{\theta-1}} , & \ds s \in D.
\end{cases}
\end{align}
$H$ satisfies the following RH problem.
\subsubsection*{RH problem for $H$}
\begin{itemize}
\item[(a)] $H$ is analytic in $\mathbb{C}\setminus \gamma$.
\item[(b)] $H_{+}(s) = \omega(J(s))e^{W(J(s))} H_{-}(s)$ for $s \in \gamma$. 
\item[(c)] $H(s) = 1+\bigO(s^{-1}) \hspace{2.05cm}$ as $s \to \infty$,

$H(s) = \bigO((s-s_{a})^{-\alpha_{0}}) \hspace{1.45cm}$ as $s \to s_{a}$, $s \in \mathbb{C} \setminus \overline{D}$,

$H(s) = \bigO((s-s_{b})^{-\alpha_{m+1}}) \hspace{1.1cm}$ as $s \to s_{b}$, $s \in \mathbb{C} \setminus \overline{D}$,

$H(s) = \bigO((s-I_{1,+}(t_{j}))^{-\frac{\alpha_{j}}{2}-\beta_{j}}) \hspace{0.1cm}$ as $s \to I_{1,+}(t_{j})$, $s \in \mathbb{C} \setminus \overline{D}$, $j=1,\ldots,m$,

$H(s) = \bigO((s-I_{2,+}(t_{j}))^{\frac{\alpha_{j}}{2}+\beta_{j}}) \hspace{0.32cm}$ as $s \to I_{2,+}(t_{j})$, $s \in D$, $j=1,\ldots,m$.
\end{itemize}
An explicit solution to this RH problem can be obtained by a direct application of the Sokhotski–Plemelj formula:
\begin{align}
H(s) & = \exp \bigg( \frac{1}{2\pi i}\oint_{\gamma} \frac{W(J(\xi)) + \log \omega(J(\xi))}{\xi-s}d\xi \bigg) \label{def of H} \\
& = \exp \bigg( \frac{-1}{2\pi i}\int_{a}^{b} \Big( W(\zeta) + \log \omega(\zeta) \Big) \bigg( \frac{I_{1,+}'(\zeta)}{I_{1,+}(\zeta)-s}-\frac{I_{2,+}'(\zeta)}{I_{2,+}(\zeta)-s} \bigg)  d\zeta \bigg), \; s \notin \gamma.  \nonumber
\end{align}
Inverting the transformations $F \mapsto G \mapsto H$ with \eqref{F to G transformation} and \eqref{G to H transformation}, we obtain
\begin{align}\label{final expression for F}
F(s) = \frac{H(s)}{\sqrt{(s-s_{a})(s-s_{b})}}\begin{cases}
\ds s\Big(\frac{s+1}{s}\Big)^{\frac{\theta-1}{\theta}}, & \ds s \in \mathbb{C}\setminus \overline{D}, \\[0.2cm]
\ds \frac{s}{\theta (c_{1}s+c_{0})^{\theta-1}} , & \ds s \in D.
\end{cases}
\end{align}
By \eqref{P1inf in terms of F}--\eqref{P2inf in terms of F}, the associated solution to the RH problem for $P^{(\infty)}$ is thus given by
\begin{align}
& P_{1}^{(\infty)}(z) = s\Big(\frac{s+1}{s}\Big)^{\frac{\theta-1}{\theta}} \frac{H(s)}{\sqrt{(s-s_{a})(s-s_{b})}}, & & s = I_{1}(z), & & z \in \mathbb{C}\setminus [a,b], \label{final expression for P1}\\
& P_{2}^{(\infty)}(z) = \frac{s}{\theta (c_{1}s+c_{0})^{\theta-1}} \frac{H(s)}{\sqrt{(s-s_{a})(s-s_{b})}}, & & s = I_{2}(z), & & z \in \mathbb{H}_{\theta} \setminus [a,b]. \label{final expression for P2}
\end{align}
Our next task is to simplify the expression for $H$.
\subsection{Simplification of $H$}
For $j=0,1,\ldots,m,m+1$, define
\begin{align}\label{def of Halphaj}
& \hspace{-0.15cm} H_{\alpha_{j}}(s) = \exp \hspace{-0.05cm} \bigg( \frac{1}{2\pi i}\oint_{\gamma} \frac{\log \omega_{\alpha_{j}}(J(\xi))}{\xi-s}d\xi \hspace{-0.05cm} \bigg) = \exp \bigg( \frac{\alpha_{j}}{2\pi i}\oint_{\gamma} \frac{\log | J(\xi)-t_{j}|}{\xi-s}d\xi \bigg).
\end{align}
\begin{proposition}\label{prop: simplified form of Ha}
$H_{\alpha_{j}}$ is analytic in $\mathbb{C}\setminus \gamma$ and admits the following expression
\begin{align*}
H_{\alpha_{j}}(s) = \begin{cases}
\ds \frac{c_{1}^{\alpha_{j}}(s-I_{1,+}(t_{j}))^{\frac{\alpha_{j}}{2}}(s-I_{2,+}(t_{j}))^{\frac{\alpha_{j}}{2}}}{(J(s)-t_{j})^{\alpha_{j}}} , & s \in \mathbb{C}\setminus \overline{D}, \\[0.3cm]
\ds c_{1}^{\alpha_{j}}(s-I_{1,+}(t_{j}))^{\frac{\alpha_{j}}{2}}(s-I_{2,+}(t_{j}))^{\frac{\alpha_{j}}{2}}, & s \in D,
\end{cases},
\end{align*}
where
\begin{align*}
& (s-I_{1,+}(t_{j}))^{\frac{\alpha_{j}}{2}} \mbox{ is analytic in } \mathbb{C}\setminus\Big((-\infty,s_{a}]\cup \gamma_{1,t_{j}}\Big), \\
& (s-I_{2,+}(t_{j}))^{\frac{\alpha_{j}}{2}} \mbox{ is analytic in } \mathbb{C}\setminus\Big((-\infty,s_{a}]\cup \gamma_{2,t_{j}}\Big),  \\
& (J(s)-t_{j})^{\alpha_{j}} \mbox{ is analytic in } \mathbb{C}\setminus\Big((-\infty,s_{a}]\cup \overline{D}\Big),
\end{align*}
where $\gamma_{k,t_{j}}$ is the part of $\gamma_{k}$ that joins $s_{a}$ with $I_{k,+}(t_{j})$ ($k=1,2$), $\arg (s-I_{k,+}(t_{j}))=0 \mbox{ if } s-I_{k,+}(t_{j})>0$ ($k=1,2$), and $\arg (J(s)-t_{j})=0 \mbox{ if } J(s)-t_{j}>0$.
\end{proposition}
\begin{proof}
The strategy of the proof is similar to that of \cite[eqs (50)--(51)]{Krasovsky}. For $\eta \in [0,1]$, define
\begin{align*}
f_{\alpha_{j}}(s;\eta) := \frac{1}{2\pi i}\oint_{\gamma} \frac{\log \omega_{\alpha_{j}}(\eta J(\xi))}{\xi-s}d\xi = \frac{\alpha_{j}}{2\pi i}\oint_{\gamma} \frac{\log |\eta J(\xi)-t_{j}|}{\xi-s}d\xi.
\end{align*}
Since $f_{\alpha_{j}}(s;1) = \log H_{\alpha_{j}}(s)$, we have
\begin{align*}
\log H_{\alpha_{j}}(s) = f_{\alpha_{j}}(s;0) + \int_{0}^{1}\partial_{\eta}f_{\alpha_{j}}(s;\eta)d\eta, 
\end{align*}
where
\begin{align}\label{der of faj for all values of eta}
\partial_{\eta} f_{\alpha_{j}}(s;\eta) = \frac{\alpha_{j}}{2\pi i}\dashint_{\gamma} \frac{J(\xi)}{(\xi-s)(\eta J(\xi)-t_{j})}d\xi, \qquad \eta \in (0,1), \; s \in \mathbb{C}\setminus \gamma.
\end{align}
The notation $\dashint$ stands for the Cauchy principal value and is relevant only for $\eta \in (\frac{t_{j}}{b},1)$, see below. The explicit value of $f_{\alpha_{j}}(s;0)$ is easy to obtain,
\begin{align}\label{value of faj at eta=0}
f_{\alpha_{j}}(s;0) = \frac{\alpha_{j} \log t_{j}}{2\pi i} \oint_{\gamma} \frac{ds}{\xi-s} = \begin{cases}
0, & \mbox{if } s \in \mathbb{C}\setminus \overline{D}, \\
\alpha_{j} \log t_{j}, & \mbox{if } s \in D.
\end{cases}
\end{align}
The rest of the proof consists of finding an explicit expression for $\int_{0}^{1}\partial_{\eta}f_{\alpha_{j}}(s;\eta)d\eta$. This is achieved in two steps: we first evaluate $\int_{0}^{\substack{ \vspace{-0.1cm}\frac{t_{j}}{b}}} \partial_{\eta}f_{\alpha_{j}}(s;\eta)d\eta$ and then \\ $\int_{\frac{t_{j}}{b}}^{1}\partial_{\eta}f_{\alpha_{j}}(s;\eta)d\eta$. For $\eta \in (0,\frac{t_{j}}{b})$, we have $\frac{t_{j}}{\eta}\in (b,+\infty)$, and thus $\eta J(\xi)-t_{j}=0$ if and only if $\xi = I_{1}(\frac{t_{j}}{\eta}) \in (s_{b},+\infty)$ or $\xi = I_{2}(\frac{t_{j}}{\eta}) \in (0,s_{b})$. Using the residue theorem, we then obtain
\begin{align*}
& \partial_{\eta} f_{\alpha_{j}}(s;\eta)  =   - \mbox{Res}\bigg( \frac{\alpha_{j}J(\xi)}{(\xi-s)(\eta J(\xi)-t_{j})}, \xi = I_{1}\Big(\frac{t_{j}}{\eta}\Big) \bigg) \\
&  +\frac{\alpha_{j}}{2\pi i}\oint_{\gamma_{\mathrm{out}}} \frac{J(\xi)}{(\xi-s)(\eta J(\xi)-t_{j})}d\xi - \begin{cases}
\mbox{Res}\big( \frac{\alpha_{j}J(\xi)}{(\xi-s)(\eta J(\xi)-t_{j})}, \xi = s \big), & \mbox{if } s \in \mathbb{C}\setminus \overline{D}, \\
0, & \mbox{if } s \in D,
\end{cases}
\end{align*}
where $\gamma_{\mathrm{out}} \subset \mathbb{C}\setminus \overline{D}$ is a closed curve oriented in the counterclockwise direction and surrounding $s$. Each of these three terms can be evaluated explicitly by a elementary computation, and we obtain
\begin{align*}
\partial_{\eta} f_{\alpha_{j}}(s;\eta) = \frac{\alpha_{j}}{\eta} - \frac{\alpha_{j}\frac{t_{j}}{\eta}}{(I_{1}(\frac{t_{j}}{\eta})-s)\eta J'(I_{1}(\frac{t_{j}}{\eta}))} + \begin{cases}
-\frac{\alpha_{j}J(s)}{\eta J(s)-t_{j}}, & \mbox{if } s \in \mathbb{C}\setminus \overline{D}, \\
0, & \mbox{if } s \in D,
\end{cases} \; \eta \in (0,\tfrac{t_{j}}{b}).
\end{align*}
Using the change of variables
\begin{align*}
\widetilde{\eta} = I_{1}\Big( \frac{t_{j}}{\eta} \Big), \qquad \frac{d\eta}{\eta} = - \frac{J'(\widetilde{\eta})}{J(\widetilde{\eta})}d\widetilde{\eta},
\end{align*}
we note that
\begin{align*}
& \int_{0}^{\frac{t_{j}}{b}} \bigg( 1 - \frac{\frac{t_{j}}{\eta}}{(I_{1}(\frac{t_{j}}{\eta})-s) J'(I_{1}(\frac{t_{j}}{\eta}))} \bigg) \frac{d\eta}{\eta} = \int_{s_{b}}^{+\infty} \bigg( \frac{J'(\widetilde{\eta})}{J(\widetilde{\eta})} - \frac{1}{\widetilde{\eta}-s} \bigg)d\widetilde{\eta} \\
& = \lim_{R\to \infty} \Big( \log J(R) - \log b - \log (R-s)+\log (s_{b}-s) \Big) = \log \Big( \frac{c_{1}(s_{b}-s)}{b} \Big),
\end{align*}
where $s \mapsto \log(s_{b}-s)$ is analytic in $\mathbb{C}\setminus [s_{b},+\infty)$ and $\arg (s_{b}-s) \in (-\pi,\pi)$. On the other hand, for $s \in \mathbb{C}\setminus \overline{D}$ we have
\begin{align*}
\int_{0}^{\frac{t_{j}}{b}}-\frac{\alpha_{j}J(s)}{\eta J(s)-t_{j}}d\eta = -\alpha_{j} \log \frac{\frac{t_{j}}{b}J(s)-t_{j}}{-t_{j}} = - \alpha_{j} \Big( \log (b-J(s)) - \log b \Big),
\end{align*}
where $s \mapsto \log(b-J(s))$ is analytic in $\mathbb{C}\setminus (\overline{D}\cup[s_{b},+\infty))$ and $\arg (b-J(s)) \in (-\pi,\pi)$. Hence, we have shown that
\begin{align*}
\int_{0}^{\frac{t_{j}}{b}}\partial_{\eta} f_{\alpha_{j}}(s;\eta)d\eta & = \alpha_{j} \log \bigg( \frac{c_{1}(s_{b}-s)}{b} \bigg) - \begin{cases}
\alpha_{j} \big( \log (b-J(s)) - \log b \big), & s \in \mathbb{C}\setminus \overline{D}, \\
0, & s \in D,
\end{cases}
\end{align*}
which, by \eqref{value of faj at eta=0}, implies
\begin{align}\label{expression for faj at tj/b}
f_{\alpha_{j}}(s;\tfrac{t_{j}}{b}) = \begin{cases}
\ds \alpha_{j} \log \bigg( \frac{c_{1}(s-s_{b})}{J(s)-b} \bigg), & s \in \mathbb{C}\setminus \overline{D}, \\[0.3cm]
\ds \alpha_{j} \log \bigg( \frac{c_{1}t_{j}(s_{b}-s)}{b} \bigg), & s \in D,
\end{cases}
\end{align}
where in \eqref{expression for faj at tj/b} the principal branches for the logarithms are taken. We now turn to the explicit evaluation of $\int_{\frac{t_{j}}{b}}^{1}\partial_{\eta}f_{\alpha_{j}}(s;\eta)d\eta$. For $\eta \in (\frac{t_{j}}{b},1)$, we have $\frac{t_{j}}{\eta} \in (t_{j},b)$, and therefore $\eta J(\xi)-t_{j}=0$ if and only if $\xi = I_{1,+}(\frac{t_{j}}{\eta}) \in \gamma_{1}$ or $\xi = I_{2,+}(\frac{t_{j}}{\eta}) \in \gamma_{2}$. Hence, using \eqref{der of faj for all values of eta}, we obtain
\small
\begin{align*}
& \partial_{\eta} f_{\alpha_{j}}(s;\eta) = \frac{\alpha_{j}}{2\pi i}\oint_{\gamma_{\mathrm{out}}} \frac{J(\xi)}{(\xi-s)(\eta J(\xi)-t_{j})}d\xi - \begin{cases}
\mbox{Res}\big( \frac{\alpha_{j}J(\xi)}{(\xi-s)(\eta J(\xi)-t_{j})}, \xi = s \big) & \mbox{if } s \in \mathbb{C}\setminus \overline{D} \\
0 & \mbox{if } s \in D
\end{cases} \\
&  - \frac{1}{2} \mbox{Res}\bigg( \frac{\alpha_{j}J(\xi)}{(\xi-s)(\eta J(\xi)-t_{j})}, \xi = I_{1,+}\Big(\frac{t_{j}}{\eta}\Big) \bigg)- \frac{1}{2} \mbox{Res}\bigg( \frac{\alpha_{j}J(\xi)}{(\xi-s)(\eta J(\xi)-t_{j})}, \xi = I_{2,+}\Big(\frac{t_{j}}{\eta}\Big) \bigg),
\end{align*}
\normalsize
where again $\gamma_{\mathrm{out}} \subset \mathbb{C}\setminus \overline{D}$ is a closed curve oriented in the counterclockwise direction and surrounding $s$. After an explicit evaluation of these residues, it becomes
\begin{align}
 & \partial_{\eta} f_{\alpha_{j}}(s;\eta) = \frac{\alpha_{j}}{\eta}  + \begin{cases}
-\frac{\alpha_{j}J(s)}{\eta J(s)-t_{j}} & \mbox{if } s \in \mathbb{C}\setminus \overline{D} \\
0 & \mbox{if } s \in D
\end{cases} \label{first int in the proof} \\
& - \frac{1}{2} \frac{\alpha_{j}\frac{t_{j}}{\eta}}{(I_{1,+}(\frac{t_{j}}{\eta})-s)\eta J'(I_{1,+}(\frac{t_{j}}{\eta}))} - \frac{1}{2} \frac{\alpha_{j}\frac{t_{j}}{\eta}}{(I_{2,+}(\frac{t_{j}}{\eta})-s)\eta J'(I_{2,+}(\frac{t_{j}}{\eta}))}, \qquad \eta \in (\tfrac{t_{j}}{b},1). \nonumber
\end{align}
Using the change of variables $\widetilde{\eta} = I_{k,+}\big( \frac{t_{j}}{\eta} \big)$, $\frac{d\eta}{\eta} = - \frac{J'(\widetilde{\eta})}{J(\widetilde{\eta})}d\widetilde{\eta}$, $k=1,2$, we get
\begin{align*}
\int_{t_{j}/b}^{1} \bigg( 1 - \frac{\frac{t_{j}}{\eta}}{(I_{k,+}(\frac{t_{j}}{\eta})-s) J'(I_{k,+}(\frac{t_{j}}{\eta}))} \bigg) \frac{d\eta}{\eta} & = \int_{I_{k,+}(t_{j})}^{s_{b}} \bigg( \frac{J'(\widetilde{\eta})}{J(\widetilde{\eta})} - \frac{1}{\widetilde{\eta}-s} \bigg)d\widetilde{\eta} \\
& = \log \bigg(\frac{b(s-I_{k,+}(t_{j}))}{t_{j}(s-s_{b})}\bigg),
\end{align*}
where the path of integration in $\widetilde{\eta}$ goes from $I_{k,+}(t_{j})$ to $s_{b}$ following $\gamma_{k}$, and the branch of the logarithm is taken accordingly.
We also note that
\begin{align*}
\int_{\frac{t_{j}}{b}}^{1}-\frac{\alpha_{j}J(s)}{\eta J(s)-t_{j}}d\eta = -\alpha_{j} \log \frac{b(J(s)-t_{j})}{t_{j}(J(s)-b)}, \qquad s \in \mathbb{C}\setminus \overline{D},
\end{align*}
where the principal branch for the logarithm is taken. Hence, we obtain 
\begin{align}
\int_{\frac{t_{j}}{b}}^{1}\partial_{\eta} f_{\alpha_{j}}(s;\eta)d\eta = & \, \frac{\alpha_{j}}{2} \log \bigg( \frac{b(s-I_{1,+}(t_{j}))}{t_{j}(s-s_{b})} \bigg)+\frac{\alpha_{j}}{2} \log \bigg( \frac{b(s-I_{2,+}(t_{j}))}{t_{j}(s-s_{b})} \bigg) \nonumber \\
& - \begin{cases}
\alpha_{j} \log \frac{b(J(s)-t_{j})}{t_{j}(J(s)-b)}, & s \in \mathbb{C}\setminus \overline{D}, \\
0, & s \in D.
\end{cases} \label{second int in the proof}
\end{align}
We obtain the claim after combining \eqref{first int in the proof} with \eqref{second int in the proof}.
\end{proof}
For $j=1,\ldots,m$, define
\small
\begin{align}\label{def of Hbetaj}
& H_{\beta_{j}}(s) = \exp \bigg( \frac{1}{2\pi i}\oint_{\gamma} \frac{\log \omega_{\beta_{j}}(J(\xi))}{\xi-s}d\xi \bigg) = \exp \bigg( \frac{1}{2\pi i}\oint_{\gamma_{a,t_{j}}} \hspace{-0.15cm} \frac{i \pi \beta_{j}}{\xi-s}d\xi + \frac{1}{2\pi i}\oint_{\gamma_{b,t_{j}}} \hspace{-0.2cm} \frac{-i \pi \beta_{j}}{\xi-s}d\xi \bigg),
\end{align}
\normalsize
where $\gamma_{a,t_{j}}$ is the part of $\gamma$ that starts at $I_{1,+}(t_{j})$, passes through $s_{a}$, and ends at $I_{2,+}(t_{j})$, while $\gamma_{b,t_{j}}$ is the part of $\gamma$ that starts at $I_{2,+}(t_{j})$, passes through $s_{b}$, and ends at $I_{1,+}(t_{j})$. After a straightforward evaluation of these integrals, we obtain
\begin{proposition}\label{prop: simplified form of Hb}
$H_{\beta_{j}}$ is analytic in $\mathbb{C}\setminus  \gamma$ and admits the following expression
\begin{align*}
H_{\beta_{j}}(s) = \bigg( \frac{s-I_{1,+}(t_{j})}{s-I_{2,+}(t_{j})} \bigg)_{a}^{-\frac{\beta_{j}}{2}}\bigg( \frac{s-I_{1,+}(t_{j})}{s-I_{2,+}(t_{j})} \bigg)_{b}^{-\frac{\beta_{j}}{2}},
\end{align*}
where the $a$ and $b$ subscripts denote the following branches:
\begin{align}
& \bigg( \frac{s-I_{1,+}(t_{j})}{s-I_{2,+}(t_{j})} \bigg)_{a}^{-\frac{\beta_{j}}{2}} \mbox{ is analytic in } \mathbb{C}\setminus \gamma_{a,t_{j}} \mbox{ and tends to } 1 \mbox{ as } s \to \infty, \label{branch a} \\
& \bigg( \frac{s-I_{1,+}(t_{j})}{s-I_{2,+}(t_{j})} \bigg)_{b}^{-\frac{\beta_{j}}{2}} \mbox{ is analytic in } \mathbb{C}\setminus \gamma_{b,t_{j}} \mbox{ and tends to } 1 \mbox{ as } s \to \infty. \label{branch b}
\end{align}
\end{proposition}
\begin{remark}
The two functions \eqref{branch a} and \eqref{branch b} coincide on $\mathbb{C}\setminus \overline{D}$, and on $D$ we have
\begin{align*}
\bigg( \frac{s-I_{1,+}(t_{j})}{s-I_{2,+}(t_{j})} \bigg)_{a}^{-\frac{\beta_{j}}{2}}e^{-\pi i \beta_{j}} = \bigg( \frac{s-I_{1,+}(t_{j})}{s-I_{2,+}(t_{j})} \bigg)_{b}^{-\frac{\beta_{j}}{2}}, \qquad s \in D.
\end{align*}
\end{remark}

\subsection{Asymptotics of $P^{(\infty)}$ as $z \to t_{k}$, $\im z > 0$}
For convenience, define $\beta_{0}=0$, $\beta_{m+1}=0$, $H_{\beta_{0}}(s) \equiv 0$, $H_{\beta_{m+1}}(s) \equiv 0$, and
\begin{align}
H_{W}(s) & = \exp \bigg( \frac{1}{2\pi i}\oint_{\gamma} \frac{W(J(\xi))}{\xi-s}d\xi \bigg) \nonumber \\
&  = \exp \bigg( \frac{-1}{2\pi i}\int_{a}^{b} W(\zeta) \bigg( \frac{I_{1,+}'(\zeta)}{I_{1,+}(\zeta)-s}-\frac{I_{2,+}'(\zeta)}{I_{2,+}(\zeta)-s} \bigg)  d\zeta \bigg), \quad s \notin \gamma. \label{def of HW}
\end{align}
The $\pm$ boundary values of $H$, $H_{W}$, $H_{\alpha_{k}}$, $H_{\beta_{k}}$ and $\sqrt{(s-s_{a})(s-s_{b})}$ will be taken with respect to the orientation of $\gamma_{1}$ and $\gamma_{2}$ (recall that the orientation of $\gamma_{1}$ is different from that of $\gamma$). In particular, 
\begin{align*}
\lim_{\epsilon \to 0_{+}} H(I_{1}(t_{k}\pm i\epsilon)) =: H_{\pm}(I_{1,\pm}(t_{k})), \qquad \lim_{\epsilon \to 0_{+}} H(I_{2}(t_{k}\pm i\epsilon)) =: H_{\pm}(I_{2,\pm}(t_{k})).
\end{align*}
\begin{lemma}\label{lemma: asymp of Ha and Hb near tk}
As $z \to t_{k}$, $\im z > 0$, we have
\begin{align*}
H_{\alpha_{k}}(I_{1}(z)) & = c_{1}^{\alpha_{k}}|I_{1,+}(t_{k})-I_{2,+}(t_{k})|^{\frac{\alpha_{k}}{2}}e^{\frac{i\pi \alpha_{k}}{4}} I_{1,+}'(t_{k})^{\frac{\alpha_{k}}{2}} (z-t_{k})^{-\frac{\alpha_{k}}{2}}(1+\bigO(z-t_{k}))\\
H_{\alpha_{k}}(I_{2}(z)) & = c_{1}^{\alpha_{k}}|I_{1,+}(t_{k})-I_{2,+}(t_{k})|^{\frac{\alpha_{k}}{2}}e^{-\frac{i\pi \alpha_{k}}{4}} I_{2,+}'(t_{k})^{\frac{\alpha_{k}}{2}}(z-t_{k})^{\frac{\alpha_{k}}{2}}(1+\bigO(z-t_{k})), \\
H_{\beta_{k}}(I_{1}(z)) & = |I_{1,+}(t_{k})-I_{2,+}(t_{k})|^{\beta_{k}}e^{\frac{i \pi \beta_{k}}{2}} I_{1,+}'(t_{k})^{-\beta_{k}} (z-t_{k})^{-\beta_{k}}(1+\bigO(z-t_{k})), \\
H_{\beta_{k}}(I_{2}(z)) & = |I_{1,+}(t_{k})-I_{2,+}(t_{k})|^{-\beta_{k}}e^{-\frac{i \pi \beta_{k}}{2}} I_{2,+}'(t_{k})^{\beta_{k}} (z-t_{k})^{\beta_{k}}(1+\bigO(z-t_{k})),
\end{align*}
where the principal branches are taken for each root.
\end{lemma}
\begin{proof}
These expansions follow from Propositions \ref{prop: simplified form of Ha} and \ref{prop: simplified form of Hb}.
\end{proof}

By combining Lemma \ref{lemma: asymp of Ha and Hb near tk} with \eqref{final expression for P1} and \eqref{final expression for P2}, we obtain

\begin{proposition}\label{prop: asymp of Pinf near tk}
As $z \to t_{k}$, $\im z > 0$, we have
\small
\begin{align*}
P_{1}^{(\infty)}(z) & \hspace{-0.05cm} = \hspace{-0.05cm} H_{W,+}(I_{1,+}(t_{k}))c_{1}^{\alpha_{k}}I_{1,+}(t_{k})\bigg(\frac{I_{1,+}(t_{k})+1}{I_{1,+}(t_{k})}\bigg)^{\hspace{-0.15cm}\frac{\theta-1}{\theta}} \frac{\prod_{\substack{j=0 \\j \neq k}}^{m+1}\hspace{-0.1cm}H_{\alpha_{j},+}(I_{1,+}(t_{k}))H_{\beta_{j},+}(I_{1,+}(t_{k}))}{\sqrt{(I_{1,+}(t_{k})-s_{a})(I_{1,+}(t_{k})-s_{b})}_{+}} \\
& \hspace{-1cm} \times |I_{1,+}(t_{k})-I_{2,+}(t_{k})|^{\frac{\alpha_{k}}{2}+\beta_{k}}e^{i\pi(\frac{ \alpha_{k}}{4}+\frac{\beta_{k}}{2})} I_{1,+}'(t_{k})^{\frac{\alpha_{k}}{2}-\beta_{k}} (z-t_{k})^{-\frac{\alpha_{k}}{2}-\beta_{k}}(1+\bigO(z-t_{k})) \\
P_{2}^{(\infty)}(z) & = H_{W,+}(I_{2,+}(t_{k}))\frac{c_{1}^{\alpha_{k}}I_{2,+}(t_{k})}{\theta (c_{1}I_{2,+}(t_{k})+c_{0})^{\theta-1}} \frac{\prod_{\substack{j=0 \\j \neq k}}^{m+1}H_{\alpha_{j},+}(I_{2,+}(t_{k}))H_{\beta_{j},+}(I_{2,+}(t_{k}))}{\sqrt{(I_{2,+}(t_{k})-s_{a})(I_{2,+}(t_{k})-s_{b})}} \\
& \hspace{-1cm} \times |I_{1,+}(t_{k})-I_{2,+}(t_{k})|^{\frac{\alpha_{k}}{2}-\beta_{k}}e^{-i\pi(\frac{ \alpha_{k}}{4}+\frac{\beta_{k}}{2})} I_{2,+}'(t_{k})^{\frac{\alpha_{k}}{2}+\beta_{k}}(z-t_{k})^{\frac{\alpha_{k}}{2}+\beta_{k}}(1+\bigO(z-t_{k})).
\end{align*}
\normalsize
In particular, as $z \to t_{k}$, $\im z > 0$, we have
\begin{align*}
& \frac{P_{2}^{(\infty)}(z)}{P_{1}^{(\infty)}(z)} = C^{(\infty)}_{21,k} (z-t_{k})^{\alpha_{k}+2\beta_{k}}(1+\bigO(z-t_{k})), \\
& C^{(\infty)}_{21,k} = \frac{H_{W,+}(I_{2,+}(t_{k}))}{H_{W,+}(I_{1,+}(t_{k}))} \frac{I_{2,+}(t_{k})}{I_{1,+}(t_{k})} \frac{I_{1,+}(t_{k})^{\frac{\theta-1}{\theta}}}{\theta (c_{1}I_{2,+}(t_{k})+c_{0})^{\theta-1}(I_{1,+}(t_{k})+1)^{\frac{\theta-1}{\theta}}}  \\
& \times e^{-i\pi(\frac{\alpha_{k}}{2}+\beta_{k})} \frac{I_{2,+}'(t_{k})^{\frac{\alpha_{k}}{2}+\beta_{k}}}{I_{1,+}'(t_{k})^{\frac{\alpha_{k}}{2}-\beta_{k}}} \frac{\sqrt{(I_{1,+}(t_{k})-s_{a})(I_{1,+}(t_{k})-s_{b})}_{+}}{\sqrt{(I_{2,+}(t_{k})-s_{a})(I_{2,+}(t_{k})-s_{b})}} \\
& \times |I_{1,+}(t_{k})-I_{2,+}(t_{k})|^{-2\beta_{k}}  \prod_{\substack{j=0 \\j \neq k}}^{m+1} \frac{H_{\alpha_{j},+}(I_{2,+}(t_{k}))H_{\beta_{j},+}(I_{2,+}(t_{k}))}{H_{\alpha_{j},+}(I_{1,+}(t_{k}))H_{\beta_{j},+}(I_{1,+}(t_{k}))}.
\end{align*}
\end{proposition}

\subsection{Asymptotics of $P^{(\infty)}$ as $z \to b$}\label{asymp of Pinf near b}

\begin{lemma}\label{lemma: asymptotics of I1 and I2 near b}
As $z \to b$, we have
\begin{align}
& I_{1}(z) = s_{b} + \frac{\sqrt{z-b}}{\sqrt{J''(s_{b})/2}} + \bigO(z-b), \quad I_{2}(z) = s_{b} - \frac{\sqrt{z-b}}{\sqrt{J''(s_{b})/2}} + \bigO(z-b), \label{asymptotics of I1 near b}
\end{align}
where $J''(s_{b})>0$ and the principal branches are taken for the roots.
\end{lemma}
\begin{proof}
This follows from \eqref{expansion of J near the endpoints b} and the identities $J(I_{1}(z))=z$ and $J(I_{2}(z))=z$.
\end{proof}
Define
\begin{align*}
& \widetilde{H}_{b}(s) = H_{W}(s)\prod_{j=0}^{m}H_{\alpha_{j}}(s)H_{\beta_{j}}(s) = \frac{H(s)}{H_{\alpha_{m+1}}(s)}, \\
& f_{1,b}(s)=s\bigg(\frac{s+1}{s}\bigg)^{\frac{\theta-1}{\theta}} \frac{\widetilde{H}_{b}(s)}{\sqrt{s-s_{a}}}, \qquad f_{2,b}(s)= \frac{s}{\theta(c_{1}s+c_{0})^{\theta-1}}\frac{\widetilde{H}_{b}(s)}{\sqrt{s-s_{a}}},
\end{align*}
where the branch for $\sqrt{s-s_{a}}$ is taken on $(-\infty,s_{a}]$. Using \eqref{final expression for P1}--\eqref{final expression for P2}, Propositions \ref{prop: simplified form of Ha}, \ref{prop: simplified form of Hb} and the expansion \eqref{asymptotics of I1 near b}, we obtain
\begin{proposition}\label{prop:asymp as z to b}
As $z \to b$, we have
\begin{align*}
& P_{1}^{(\infty)}(z) = \frac{c_{1}^{\alpha_{m+1}}f_{1,b,+}(s_{b})\big(\frac{J''(s_{b})}{2}\big)^{\frac{1}{4}-\frac{\alpha_{m+1}}{2}}}{(z-b)^{\frac{1}{4}+\frac{\alpha_{m+1}}{2}}}\Big(1+\bigO( \sqrt{z-b}) \Big), \\
& P_{2}^{(\infty)}(z) = i\frac{c_{1}^{\alpha_{m+1}}f_{2,b,+}(s_{b})\big(\frac{J''(s_{b})}{2}\big)^{\frac{1}{4}-\frac{\alpha_{m+1}}{2}}}{(z-b)^{\frac{1}{4}-\frac{\alpha_{m+1}}{2}}} \Big(1+\bigO( \sqrt{z-b}) \Big),
\end{align*}
where the principal branches are taken for the roots, and
\begin{align*}
f_{1,b,+}(s_{b}) := \lim_{\epsilon \to 0} f_{1,b}(s_{b}+\epsilon), \qquad f_{2,b,+}(s_{b}) := \lim_{\epsilon \to 0} f_{2,b}(s_{b}-\epsilon).
\end{align*}
In particular, as $z \to b$, we have
\begin{align*}
& \frac{P_{2}^{(\infty)}(z)}{P_{1}^{(\infty)}(z)} = \frac{i \omega_{b}(b)e^{W(b)}}{\theta b^{\theta-1}} (z-b)^{\alpha_{m+1}} \Big(1+\bigO( \sqrt{z-b}) \Big).
\end{align*}
\end{proposition}
\subsection{Asymptotics of $P^{(\infty)}$ as $z \to a$}
The analysis done in this section is similar to the one of Section \ref{asymp of Pinf near b}.
\begin{lemma}
As $z \to a$, $\pm \im z  >0$, we have
\begin{align}
& I_{1}(z) = s_{a} \pm \frac{i\sqrt{z-a}}{\sqrt{|J''(s_{a})|/2}} + \bigO(z-a),  \quad  I_{2}(z) = s_{a} \pm \frac{-i\sqrt{z-a}}{\sqrt{|J''(s_{a})|/2}} + \bigO(z-a), \label{asymptotics of I1 near a}
\end{align}
where $J''(s_{a})<0$ and the principal branches are taken for the roots.
\end{lemma}
\begin{proof}
It suffices to combine \eqref{expansion of J near the endpoints a} with the identities $J(I_{1}(z))=z$ and $J(I_{2}(z))=z$.
\end{proof}
Define
\begin{align*}
& \widetilde{H}_{a}(s) = H_{W}(s)\prod_{j=1}^{m+1}H_{\alpha_{j}}(s)H_{\beta_{j}}(s) = \frac{H(s)}{H_{\alpha_{0}}(s)}, \\
& f_{1,a}(s)=s\bigg(\frac{s+1}{s}\bigg)^{\frac{\theta-1}{\theta}} \frac{\widetilde{H}_{a}(s)}{\sqrt{s-s_{b}}}, \qquad f_{2,a}(s)= \frac{s}{\theta(c_{1}s+c_{0})^{\theta-1}}\frac{\widetilde{H}_{a}(s)}{\sqrt{s-s_{b}}},
\end{align*}
where $\sqrt{s-s_{b}}$ is analytic in  $\mathbb{C}\setminus \big((-\infty,s_{a}]\cup \gamma_{1}\big)$ and such that $\sqrt{s-s_{b}}>0$ if $s>s_{b}$. The following proposition follows from \eqref{final expression for P1}, \eqref{final expression for P2} and Propositions \ref{prop: simplified form of Ha} and \ref{prop: simplified form of Hb}.
\begin{proposition}
As $z \to a$, $\im z >0$, we have
\begin{align*}
& P_{1}^{(\infty)}(z) = \frac{c_{1}^{\alpha_{0}}f_{1,a,+}(s_{a})\big( \frac{|J''(s_{a})|}{2} \big)^{\frac{1}{4}-\frac{\alpha_{0}}{2}}e^{\frac{\pi i}{2}(\alpha_{0}-\frac{1}{2})}}{(z-a)^{\frac{1}{4}+\frac{\alpha_{0}}{2}}} \Big(1+ \bigO( \sqrt{z-a}) \Big), \\
& P_{2}^{(\infty)}(z) = \frac{c_{1}^{\alpha_{0}}f_{2,a,+}(s_{a})\big( \frac{|J''(s_{a})|}{2} \big)^{\frac{1}{4}-\frac{\alpha_{0}}{2}}e^{-\frac{\pi i}{2}(\alpha_{0}-\frac{1}{2})}}{(z-a)^{\frac{1}{4}-\frac{\alpha_{0}}{2}}} \Big(1+ \bigO( \sqrt{z-a}) \Big),
\end{align*}
where the principal branches are taken for the roots, and
\begin{align*}
f_{1,a,+}(s_{a}) := \lim_{\epsilon \to 0} f_{1,a}(s_{a}+ e^{\frac{3\pi i}{4}}\epsilon), \qquad f_{2,a,+}(s_{a}) := \lim_{\epsilon \to 0} f_{2,a}(s_{a}+e^{-\frac{\pi i}{4}}\epsilon).
\end{align*}
In particular, as $z \to a$, $\im z>0$, we have
\begin{align*}
\frac{P_{2}^{(\infty)}(z)}{P_{1}^{(\infty)}(z)} =  \frac{i\omega_{a}(a)e^{W(a)}}{\theta a^{\theta-1}} e^{-\pi i \alpha_{0}}(z-a)^{\alpha_{0}} \Big( 1+\bigO(\sqrt{z-a}) \Big).
\end{align*}
\end{proposition}
\subsection{Asymptotics as $z \to \infty$}

Recall that $\beta_{0}=\beta_{m+1}=0$, $t_{0}=a$ and $t_{m+1}=b$.
\begin{lemma}
As $s \to 0$, we have $H(s) = H(0) ( 1 + \bigO(s) )$, where
\begin{align}\label{def of H0}
& H(0) = \exp \bigg( \int_{a}^{b} W(x) \rho(x) dx \bigg) \prod_{j=0}^{m+1} \bigg( e^{\alpha_{j}\int_{a}^{b}\log|t_{j}-x|\rho(x)dx} e^{\pi i \beta_{j} (1-2\int_{t_{j}}^{b}\rho(x)dx)} \bigg).
\end{align}
Furthermore, the identity \eqref{identity for rho and real and im parts} holds.
\end{lemma}
\begin{proof}
Recall that $H(s) = H_{W}(s) \prod_{j=0}^{m+1}H_{\alpha_{j}}(s)H_{\beta_{j}}(s)$. Proposition \ref{prop:density} implies that
\begin{align}\label{rho in the bulk}
\rho(x) = \frac{-1}{2\pi i}\bigg( \frac{I_{1,+}'(x)}{I_{1,+}(x)} - \frac{I_{2,+}'(x)}{I_{2,+}(x)} \bigg), \qquad x \in (a,b).
\end{align}
Hence, using the definition \eqref{def of HW} of $H_{W}$, we get
\begin{align*}
\log H_{W}(0) & = \frac{1}{2\pi i}\oint_{\gamma} \frac{W(J(\xi))}{\xi}d\xi \\
& = \frac{1}{2\pi i}  \int_{a}^{b} W(x)\bigg(\frac{I_{2,+}'(x)}{I_{2,+}(x)} - \frac{I_{1,+}'(x)}{I_{1,+}(x)} \bigg) dx = \int_{a}^{b} W(x) \rho(x) dx.
\end{align*}
Similarly, using \eqref{def of Halphaj} and \eqref{def of Hbetaj}, for $j=1,\ldots,m+1$, we obtain
\begin{align}\label{HaHb at 0 inside proof 2}
\log H_{\alpha_{j}}(0) = \alpha_{j}\int_{a}^{b}\log|t_{j}-x|\rho(x)dx, \qquad \log H_{\beta_{j}}(0) = i \pi \beta_{j}\bigg(  1-2\int_{t_{j}}^{b}\rho(x)dx \bigg),
\end{align}
which already proves \eqref{def of H0}. On the other hand, using Propositions \ref{prop: simplified form of Ha}, \ref{prop: simplified form of Hb} and the fact that $\overline{I_{2,+}(t_{j})} = I_{1,+}(t_{j})$, we obtain
\begin{align}\label{HaHb at 0 inside proof}
H_{\alpha_{j}}(0) = c_{1}^{\alpha_{j}} |I_{1,+}(t_{j})|^{\alpha_{j}}, \qquad H_{\beta_{j}}(0) = e^{\pi i\beta_{j}-2i \beta_{j}\arg I_{1,+}(t_{j})}, \qquad j=0,\ldots,m+1,
\end{align}
where $\arg I_{1,+}(t_{j}) \in [0,\pi]$, $j=0,1,\ldots,m+1$. By comparing \eqref{HaHb at 0 inside proof 2} and \eqref{HaHb at 0 inside proof}, we obtain \eqref{identity for rho and real and im parts}.
\end{proof}

\begin{proposition}\label{prop:asymp of P2inf at inf}
As $z \to \infty$, $z \in \mathbb{H}_{\theta}$, we have
\begin{align*}
& P_{2}^{(\infty)}(z) = \frac{H(0)c_{0}}{-i\sqrt{|s_{a}s_{b}|}\theta} z^{-\theta} \Big( 1 + \bigO( z^{-\theta} )\Big).
\end{align*}
\end{proposition}
\begin{proof}
Since the branch of $\sqrt{(s-s_{a})(s-s_{b})}$ is taken on $\gamma_{1}$, we have
\begin{align*}
\sqrt{(s-s_{a})(s-s_{b})}|_{s=0}=-i\sqrt{|s_{a}s_{b}|}.
\end{align*}
The claim now follows after substituting \eqref{asymp of I1 and I2 at infty} in the expression \eqref{final expression for P2} of $P_{2}^{(\infty)}$.
\end{proof}
\section{The convergence of $P \to P^{(\infty)}$}\label{section:small norm shifted RHP}

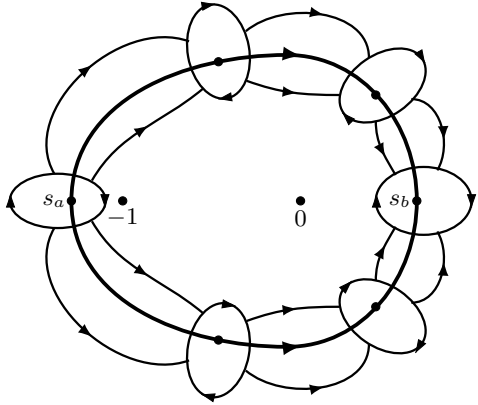
\begin{figure}
\begin{center}
\begin{tikzpicture}[master, scale=0.9]
\draw[black, line width=0.5 mm, ->-=0.6] (-2.55,0) to [out=90, in=180] (0.6,2.15) to [out=0, in=90] (2.5,0);
\draw[black, line width=0.5 mm, ->-=0.6] (-2.55,0) to [out=-90, in=180] (0.6,-2.15) to [out=0, in=-90] (2.5,0);

\node at (-2.8,0) {\small $s_{a}$};
\node at (2.25,0) {\small $s_{b}$};
\node at (-1.8,-0.25) {\small $-1$};
\node at (0.8,-0.25) {\small $0$};

\draw[fill] (-2.55,0) circle (0.06);
\draw[fill] (2.5,0) circle (0.06);
\draw[fill] (-1.8,0) circle (0.06);
\draw[fill] (0.8,0) circle (0.06);


\draw[black, line width=0.3 mm] (2.6,0) ellipse (0.7cm and 0.5cm);
\draw[black, line width=0.3 mm] (-2.75,0) ellipse (0.7cm and 0.4cm);

\draw[fill] (1.9,1.56) circle (0.06);
\draw[black, line width=0.3 mm, rotate around={35:(2,1.7)}] (2,1.7) ellipse (0.7cm and 0.45cm);

\draw[fill] (-0.4,2.05) circle (0.06);
\draw[black, line width=0.3 mm, rotate around={100:(-0.4,2.2)}] (-0.4,2.2) ellipse (0.7cm and 0.45cm);

\draw[fill] (1.9,-1.56) circle (0.06);
\draw[black, line width=0.3 mm, rotate around={-35:(2,-1.7)}] (2,-1.7) ellipse (0.7cm and 0.45cm);

\draw[fill] (-0.4,-2.05) circle (0.06);
\draw[black, line width=0.3 mm, rotate around={-100:(-0.4,-2.2)}] (-0.4,-2.2) ellipse (0.7cm and 0.45cm);

\draw[black, line width=0.3 mm, ->-=0.55] (-2.8,0.4) to [out=120, in=160] ($(-0.4,2.05)+(-0.43,0.3)$);
\draw[black, line width=0.3 mm, ->-=0.55] ($(-0.4,2.05)+(0.4,0.45)$) to [out=40, in=90] ($(1.9,1.56)+(-0.1,0.55)$);
\draw[black, line width=0.3 mm, ->-=0.65] ($(1.9,1.56)+(0.52,-0.06)$) to [out=0, in=60] ($(2.5,0)+(0.3,0.45)$);

\draw[black, line width=0.3 mm, ->-=0.55] (-2.8,-0.4) to [out=-120, in=-160] ($(-0.4,-2.05)+(-0.43,-0.3)$);
\draw[black, line width=0.3 mm, ->-=0.55] ($(-0.4,-2.05)+(0.4,-0.45)$) to [out=-40, in=-90] ($(1.9,-1.56)+(-0.1,-0.55)$);
\draw[black, line width=0.3 mm, ->-=0.65] ($(1.9,-1.56)+(0.52,0.06)$) to [out=0, in=-60] ($(2.5,0)+(0.3,-0.45)$);

\draw[black, line width=0.3 mm, ->-=0.55] (-2.25,0.3) to [out=60, in=220] ($(-0.4,2.05)+(-0.23,-0.4)$);
\draw[black, line width=0.3 mm, ->-=0.55] ($(-0.4,2.05)+(0.41,-0.28)$) to [out=-20, in=180] ($(1.9,1.56)+(-0.53,0)$);
\draw[black, line width=0.3 mm, ->-=0.65] ($(1.9,1.56)+(0,-0.38)$) to [out=-90, in=120] ($(2.5,0)+(-0.31,0.38)$);

\draw[black, line width=0.3 mm, ->-=0.55] (-2.25,-0.3) to [out=-60, in=-220] ($(-0.4,-2.05)+(-0.23,0.4)$);
\draw[black, line width=0.3 mm, ->-=0.55] ($(-0.4,-2.05)+(0.41,0.28)$) to [out=20, in=-180] ($(1.9,-1.56)+(-0.53,0)$);
\draw[black, line width=0.3 mm, ->-=0.65] ($(1.9,-1.56)+(0,0.38)$) to [out=90, in=-120] ($(2.5,0)+(-0.31,-0.38)$);

\draw[black,arrows={-Triangle[length=0.18cm,width=0.12cm]}]
(3.3,-0.1) --  ++(-90:0.001);
\draw[black,arrows={-Triangle[length=0.18cm,width=0.12cm]}]
(1.91,0.1) --  ++(90:0.001);
\draw[black,arrows={-Triangle[length=0.18cm,width=0.12cm]}]
(-2.06,-0.1) --  ++(-90:0.001);
\draw[black,arrows={-Triangle[length=0.18cm,width=0.12cm]}]
(-3.44,0.1) --  ++(90:0.001);

\draw[black,arrows={-Triangle[length=0.18cm,width=0.12cm]}]
(2.62,2.03) --  ++(-60:0.001);
\draw[black,arrows={-Triangle[length=0.18cm,width=0.12cm]}]
(1.41,1.3) --  ++(135:0.001);
\draw[black,arrows={-Triangle[length=0.18cm,width=0.12cm]}]
(-0.45,2.9) --  ++(10:0.001);
\draw[black,arrows={-Triangle[length=0.18cm,width=0.12cm]}]
(-0.4,1.5) --  ++(-170:0.001);

\draw[black,arrows={-Triangle[length=0.18cm,width=0.12cm]}]
(2.48,-2.24) --  ++(-120:0.001);
\draw[black,arrows={-Triangle[length=0.18cm,width=0.12cm]}]
(1.55,-1.18) --  ++(45:0.001);
\draw[black,arrows={-Triangle[length=0.18cm,width=0.12cm]}]
(-0.65,-2.85) --  ++(165:0.001);
\draw[black,arrows={-Triangle[length=0.18cm,width=0.12cm]}]
(-0.15,-1.55) --  ++(-10:0.001);
\end{tikzpicture}
\end{center}
\caption{\label{fig: contour for mathcal F} The contour $J^{-1}(\Sigma_{P})$ with $m=2$. The thick curves are $\gamma_{1}$ and $\gamma_{2}$. The two dots lying in the upper half plane are $I_{1,+}(t_{j})$, $j=1,2$, and the two dots lying in the lower half plane are $I_{2,+}(t_{j})$, $j=1,2$.}
\end{figure}

In this section we follow the method of \cite[Section 4.7]{ClaeysWang}. As in the construction of $P^{(\infty)}$, the mapping $J$ is used to transform the $1\times 2$ vector valued function $P$ to a scalar valued function $\mathcal{F}$ as follows:
\begin{align}\label{def of mathcal F}
\mathcal{F}(s) := \begin{cases}
P_{1}(J(s)), & s \in \mathbb{C}\setminus \overline{D} \mbox{ and } J(s) \notin \Sigma_{P}, \\
P_{2}(J(s)), & s \in D\setminus [-1,0] \mbox{ and } J(s) \notin \Sigma_{P},
\end{cases}
\end{align}
where $\Sigma_{P}$ was defined in \eqref{def of Sigma P}. We can retrieve $P_{1}$ and $P_{2}$ from $\mathcal{F}$ by
\begin{align}
& P_{1}(z) = \mathcal{F}(I_{1}(z)), & & z \in \mathbb{C}\setminus \Sigma_{P}, \nonumber \\
& P_{2}(z) = \mathcal{F}(I_{2}(z)), & & z \in \mathbb{H}_{\theta}\setminus \Sigma_{P}. \label{P2 in terms of mathcal F}
\end{align}
It will be convenient to write $J^{-1}(\Sigma_{P})$ as the union of three contours as follows:
\begin{align*}
J^{-1}(\Sigma_{P}) = \Sigma'\cup \Sigma'' \cup (\gamma_{1}\cup \gamma_{2}), \; \mbox{where} \; \Sigma' = I_{1}(\Sigma_{P}\setminus[a,b]), \; \Sigma'' = I_{2}(\Sigma_{P}\setminus[a,b]).
\end{align*}
We choose the orientation of $J^{-1}(\Sigma_{P})$ that is induced from the orientation of $\Sigma_{P}$ through $I_{1}$ and $I_{2}$, see also Figure \ref{fig: contour for mathcal F}. Since $P_{2}$ satisfies $P_{2}(e^{\frac{\pi i}{\theta}}x) = P_{2}(e^{-\frac{\pi i}{\theta}}x)$ for all $x \geq 0$, $\mathcal{F}$ is analytic on $(-1,0)$. Furthermore, since 
\begin{align*}
P_{2}(z)=\bigO(1), \; \mbox{as } z \to 0, \, z \in \mathbb{H}_{\theta}, \quad \mbox{ and } \quad P_{2}(z)=\bigO(z^{-\theta}), \; \mbox{as } z \to \infty, \, z \in \mathbb{H}_{\theta},
\end{align*}
the singularities of $\mathcal{F}$ at $-1$ and $0$ are removable and $\mathcal{F}$ has at least a simple zero at $0$, and thus $\mathcal{F}$ is analytic in $\mathbb{C}\setminus J^{-1}(\Sigma_{P})$. By \eqref{def of mathcal F}, we have
\begin{align*}
& P_{1,\pm}(J(s)) = \mathcal{F}_{\pm}(s), & & P_{2,\pm}(J(s)) = \mathcal{F}_{\pm}(I_{2}(J(s)), & & s \in \Sigma', \\
& P_{2,\pm}(J(s)) = \mathcal{F}_{\pm}(s), & & P_{1,\pm}(J(s)) = \mathcal{F}_{\pm}(I_{1}(J(s)), & & s \in \Sigma'',
\end{align*}
which implies that the jumps of $\mathcal{F}$ on $\Sigma'\cup \Sigma''$ are nonlocal and given by
\begin{align}
& \mathcal{F}_{+}(s) = \mathcal{F}_{-}(s) J_{P,11}(J(s)) + \mathcal{F}_{-}(I_{2}(J(s)))J_{P,21}(J(s)), & & s \in \Sigma', \label{nonlocal jumps 1 for mathcal F}\\
& \mathcal{F}_{+}(s) = \mathcal{F}_{-}(I_{1}(J(s))) J_{P,12}(J(s)) + \mathcal{F}_{-}(s)J_{P,22}(J(s)), & & s \in \Sigma''. \label{nonlocal jumps 2 for mathcal F}
\end{align}
The jumps for $\mathcal{F}$ on $\gamma_{1}\cup \gamma_{2}$ can be computed similarly and are identical to those of $F$:
\begin{align*}
& \mathcal{F}_{+}(s) = -\frac{\theta J(s)^{\theta-1}}{\omega(J(s))e^{W(J(s))}} \mathcal{F}_{-}(s), & & s \in \gamma_{1}, \\
& \mathcal{F}_{+}(s) = \frac{\omega(J(s))e^{W(J(s))}}{\theta J(s)^{\theta-1}} \mathcal{F}_{-}(s), & & s \in \gamma_{2}.
\end{align*}
Finaly, using the RH conditions (c) and (d) of the RH problem for $P$, we conclude that $\mathcal{F}$ admits the following behaviors near $\infty,0,s_{a},s_{b},I_{1,+}(t_{j}),I_{2,+}(t_{j})$, $j=1,\ldots,m$:
\begin{align}
& \hspace{-0.18cm} \mathcal{F}(s) = 1+\bigO(s^{-1}), & & \mbox{as } s \to \infty, \label{mathcal F asymp 1} \\
& \hspace{-0.18cm} \mathcal{F}(s) = \bigO(s), & & \mbox{as } s \to 0, \label{mathcal F asymp 2} \\
& \hspace{-0.18cm} \mathcal{F}(s) = \bigO((s-s_{a})^{-\frac{1}{2}-\alpha_{0}}), & & \mbox{as } s \to s_{a}, \; s \in \mathbb{C} \setminus \overline{D}, \label{mathcal F asymp 3} \\
& \hspace{-0.18cm} \mathcal{F}(s) = \bigO((s-s_{b})^{-\frac{1}{2}-\alpha_{m+1}}), & & \mbox{as } s \to s_{b}, \; s \in \mathbb{C} \setminus \overline{D}, \label{mathcal F asymp 4} \\
& \hspace{-0.18cm} \mathcal{F}(s) = \bigO((s-I_{1,+}(t_{j}))^{-\frac{\alpha_{j}}{2}-\beta_{j}}), & & \mbox{as } s \to I_{1,+}(t_{j}), \; s \in \mathbb{C} \setminus \overline{D}, \; j=1,...,m, \label{mathcal F asymp 5} \\
& \hspace{-0.18cm} \mathcal{F}(s) = \bigO((s-I_{2,+}(t_{j}))^{\frac{\alpha_{j}}{2}+\beta_{j}}), & & \mbox{as } s \to I_{2,+}(t_{j}), \; s \in D, \; j=1,\ldots,m. \label{mathcal F asymp 6}
\end{align}
Because of the nonlocal jumps \eqref{nonlocal jumps 1 for mathcal F}--\eqref{nonlocal jumps 2 for mathcal F}, $\mathcal{F}$ does not satisfy a RH problem in the usual sense, and following \cite{Gakhov, ClaeysWang} we will say that $\mathcal{F}$ satisfies a ``shifted" RH problem.

\medskip By \eqref{final expression for F}, $F(s) \neq 0$ for all $s \in \mathbb{C}\setminus (\gamma_{1}\cup \gamma_{2}\cup \{0\})$, and therefore
\begin{align}\label{def of R}
R(s) := \frac{\mathcal{F}(s)}{F(s)}, \qquad \mbox{for } s \in \mathbb{C}\setminus (J^{-1}(\Sigma_{P})\cup \{0\})
\end{align}
is analytic. Since $F$ and $\mathcal{F}$ have the same jumps on $\gamma_{1}\cup \gamma_{2}$, $R(s)$ is analytic on 
\begin{align*}
(\gamma_{1}\cup \gamma_{2}) \setminus \{s_{a},s_{b},I_{1,+}(t_{1}),\ldots,I_{1,+}(t_{m}),I_{2,+}(t_{1}),\ldots,I_{2,+}(t_{m})\}.
\end{align*}
Using \eqref{mathcal F asymp 3}--\eqref{mathcal F asymp 6} and the definition \eqref{final expression for F} of $F$, we verify that the singularities of $R$ at $s_{a},s_{b},$ $I_{1,+}(t_{1}),\ldots,I_{1,+}(t_{m}),I_{2,+}(t_{1}),\ldots,I_{2,+}(t_{m})$ are removable, so that $R$ is in fact analytic in a whole neighborhood of $\gamma_{1}\cup \gamma_{2}$. We summarize the properties of $R$.
\subsubsection*{Shifted RH problem for $R$}
\begin{itemize}
\item[(a)] $R: \mathbb{C}\setminus (\Sigma'\cup \Sigma'') \to \mathbb{C}$ is analytic.
\item[(b)] $R$ satisfies the jumps
\begin{align*}
& R_{+}(s) = R_{-}(s) J_{R,11}(s) + R_{-}(I_{2}(J(s)))J_{R,21}(s), & & s \in \Sigma', \\
& R_{+}(s) = R_{-}(I_{1}(J(s))) J_{R,12}(s) + R_{-}(s)J_{R,22}(s), & & s \in \Sigma'',
\end{align*}
where
\begin{align}
& \hspace{-0.18cm} J_{R,11}(s) = J_{P,11}(J(s)), & & \hspace{-0.8cm} J_{R,21}(s) = J_{P,21}(J(s)) \frac{F(I_{2}(J(s)))}{F(s)}, \label{jumps JR11 and JR21} \\
& \hspace{-0.18cm} J_{R,12}(s) = J_{P,12}(J(s)) \frac{F(I_{1}(J(s)))}{F(s)}, & & J_{R,22}(s) = J_{P,22}(J(s)). \label{jumps JR12 and JR22}
\end{align}
\item[(c)] $R$ is bounded, and $R(s) = 1+\bigO(s^{-1})$ as $s \to \infty$.
\end{itemize}
By \eqref{JP estimate on the lenses}--\eqref{JP estimate on the disks tj} and the explicit expression \eqref{final expression for F} for $F(s)$, as $n \to + \infty$ we have
\footnotesize
\begin{align*}
& J_{R,11}(s) = 1+\bigO(e^{-cn}), & & \hspace{-0.3cm} J_{R,21}(s) = \bigO(e^{-cn}), & & \hspace{-0.3cm} \mbox{u.f. } s \in I_{1}(\sigma_{+}\hspace{-0.1cm}\cup \hspace{-0.05cm} \sigma_{-}\hspace{-0.05cm})\hspace{-0.05cm}\cap\hspace{-0.05cm} \Sigma', \nonumber \\
& J_{R,11}(s) = 1+\tfrac{J_{R,11}^{(1)}(s)}{n}+\bigO(\tfrac{n^{2 \beta_{\max}}}{n^{2}}), & & \hspace{-0.3cm} J_{R,21}(s) = \tfrac{J_{R,21}^{(1)}(s)}{n}+\bigO(\tfrac{n^{2\beta_{\max}}}{n^{2}}), & & \hspace{-0.3cm}\mbox{u.f. } s \in \cup_{j=0}^{m+1}I_{1}(\partial \mathcal{D}_{t_{j}}),  \nonumber \\
& J_{R,22}(s) = 1+\bigO(e^{-cn}), & & \hspace{-0.3cm} J_{R,12}(s) = \bigO(e^{-cn}), & & \hspace{-0.3cm}\mbox{u.f. } s \in I_{2}(\sigma_{+}\hspace{-0.1cm}\cup \hspace{-0.05cm}\sigma_{-}\hspace{-0.05cm})\hspace{-0.05cm}\cap\hspace{-0.05cm} \Sigma'',  \nonumber \\
& J_{R,22}(s) = 1+\tfrac{J_{R,22}^{(1)}(s)}{n}+\bigO(\tfrac{n^{2\beta_{\max}}}{n^{2}}), & & \hspace{-0.3cm} J_{R,12}(s) = \tfrac{J_{R,12}^{(1)}(s)}{n}+\bigO(\tfrac{n^{2\beta_{\max}}}{n^{2}}), & & \hspace{-0.3cm}\mbox{u.f. } s \in \cup_{j=0}^{m+1}I_{2}(\partial \mathcal{D}_{t_{j}}), 
\end{align*}
\normalsize
for a certain $c>0$, where ``u.f." means ``uniformly for", and these estimates hold also uniformly for $\alpha_{0},\ldots,\alpha_{m+1}$ in compact subsets of $\{z \in \mathbb{C}: \re z >-1\}$, uniformly for $\beta_{1},\ldots,\beta_{m}$ in compact subsets of $\{z \in \mathbb{C}: \re z \in (-\frac{1}{2},\frac{1}{2})\}$, and uniformly in $t_{1},\ldots,t_{m},\theta$ such that \eqref{assumption on t1 tm theta} holds for a certain $\delta \in (0,1)$.

Define the operator $\Delta_{R}$ acting on functions defined on $\Sigma_{R}:=\Sigma'\cup \Sigma''$ by
\begin{align*}
& \Delta_{R}f(s) = [J_{R,11}(s)-1]f(s)+J_{R,21}(s)f(I_{2}(J(s))), & & s \in \Sigma', \\
& \Delta_{R}f(s) = J_{R,12}(s)f(I_{1}(J(s))) + [J_{R,22}(s)-1]f(s), & & s \in \Sigma''.
\end{align*}
Let $\Omega$ be a fixed (independent of $n$) compact subset of 
\begin{align}\label{Omega compact subset}
\{\re z > -1 \}^{m+2} \hspace{-0.05cm}\times\hspace{-0.05cm} \{\re z \in (-\tfrac{1}{4}, \tfrac{1}{4})\}^{m} \hspace{-0.05cm}\times\hspace{-0.05cm} \{(t_{1},...,t_{m}):a<t_{1}<...<t_{m}<b\} \hspace{-0.05cm}\times\hspace{-0.05cm} (0,\infty),
\end{align}
and for notational convenience we denote $\mathfrak{p}:=(\alpha_{0},\ldots,\alpha_{m+1},\beta_{1},\ldots,\beta_{m},t_{1},\ldots,t_{m},\theta)$.
The same analysis as in \cite[Section 4.7]{ClaeysWang} shows that in our case, there exists $M=M(\Omega)>0$ such that
\begin{align}\label{norm estimate for Delta R}
||\Delta_{R}||_{L^{2}(\Sigma_{R})} \leq \frac{M}{n^{1-2\beta_{\max}}}, \qquad \mbox{for all } \mathfrak{p}\in \Omega,
\end{align}
so that the operator $1-C_{\Delta_{R}}$ can be inverted and written as a Neumann series for all $n \geq n_{0}=n_{0}(\Omega)$ and all $\mathfrak{p}\in \Omega$. Furthermore, like in \cite[eq (4.100)]{ClaeysWang} the following formula holds
\begin{align}\label{identity for R first time}
\hspace{-0.18cm} R(s)\hspace{-0.08cm}=\hspace{-0.08cm}1+\frac{1}{2\pi i}\int_{\Sigma_{R}} \frac{\Delta_{R}(1)(\xi)}{\xi-s}d\xi + \frac{1}{2\pi i}\int_{\Sigma_{R}} \frac{\Delta_{R}(R_{-}-1)(\xi)}{\xi-s}d\xi, \; s \in \mathbb{C}\setminus \Sigma_{R}.
\end{align}
Let $\delta' >0$ be a small but fixed constant, and let $s_{0} \in \mathbb{C}\setminus \Sigma_{R}$. Since $J_{R,11}, J_{R,21}$ are analytic in a neighborhood of $\Sigma'$ and $J_{R,12}, J_{R,22}$ are analytic in a neighborhood of $\Sigma''$, the contour $\Sigma_{R}$ in \eqref{identity for R first time} can always be deformed into another contour $\Sigma_{R}'$ in such as a way that $|\xi-s_{0}|\geq \delta'$ for all $\xi \in \Sigma_{R}'$. 
Therefore, \eqref{norm estimate for Delta R} and \eqref{identity for R first time} imply that
\begin{align}
& R(s)=1 + R^{(1)}(s)n^{-1} + \bigO(n^{-2+4\beta_{\max}}), \quad \mbox{as } n \to +\infty, \label{asymptotics of R as n to infty} \\
& R^{(1)}(s) = \frac{1}{2\pi i}\int_{\bigcup_{j=0}^{m+1}I_{1}(\partial \mathcal{D}_{t_{j}}) \cup \bigcup_{j=0}^{m+1}I_{2}(\partial \mathcal{D}_{t_{j}})} \frac{\Delta_{R}^{(1)}(1)(\xi)}{\xi-s}d\xi, \nonumber
\end{align}
uniformly for $s \in \mathbb{C}\setminus \Sigma_{R}$ and for $\mathfrak{p}\in \Omega$, where 
\begin{align}
& \Delta_{R}^{(1)}f(s) = J_{R,11}^{(1)}(s)f(s)+J_{R,21}^{(1)}(s)f(I_{2}(J(s))), & & s \in \cup_{j=0}^{m+1}I_{1}(\partial \mathcal{D}_{t_{j}}), \label{def of Deltap1p 1} \\
& \Delta_{R}^{(1)}f(s) = J_{R,12}^{(1)}(s)f(I_{1}(J(s))) + J_{R,22}^{(1)}(s)f(s), & & s \in \cup_{j=0}^{m+1}I_{2}(\partial \mathcal{D}_{t_{j}}). \label{def of Deltap1p 2}
\end{align}
From \eqref{comp tk 1}, \eqref{jumps for P on Db}, \eqref{jumps for P on Da} and \eqref{jumps JR11 and JR21}--\eqref{jumps JR12 and JR22}, one sees that $\Delta_{R}^{(1)}(1)(\xi)$ can be analytically continued from $\bigcup_{j=0}^{m+1}I_{1}(\partial \mathcal{D}_{t_{j}}) \cup \bigcup_{j=0}^{m+1}I_{2}(\partial \mathcal{D}_{t_{j}})$ to $\big( \bigcup_{j=0}^{m+1}I_{1}(\overline{\mathcal{D}_{t_{j}}}\setminus \{t_{j}\}) \cup \bigcup_{j=0}^{m+1}I_{2}(\overline{\mathcal{D}_{t_{j}}}\setminus \{t_{j}\}) \big)$, and that $\Delta_{R}^{(1)}(1)(\xi)$ has simple poles at each of the points $s_{a},s_{b},I_{1,+}(t_{1}),\ldots,I_{1,+}(t_{m}), I_{2,+}(t_{1}),\ldots,I_{2,+}(t_{m})$. Therefore, for all $s \in \mathbb{C}\setminus\big( \bigcup_{j=0}^{m+1}$ $I_{1}(\overline{\mathcal{D}_{t_{j}}}) \cup \bigcup_{j=0}^{m+1}I_{2}(\overline{\mathcal{D}_{t_{j}}}) \big)$ we have
\begin{align}
R^{(1)}(s) & = \frac{1}{(s-s_{a})}\mbox{Res}\Big( \Delta_{R}^{(1)}1(\xi),\xi=s_{a} \Big) + \frac{1}{(s-s_{b})}\mbox{Res}\Big( \Delta_{R}^{(1)}1(\xi),\xi=s_{b} \Big) \nonumber \\
& \hspace{-0.5cm}+ \sum_{k=1}^{m} \bigg( \frac{\mbox{Res}\big( \Delta_{R}^{(1)}1(\xi),\xi=I_{1,+}(t_{j}) \big)}{s-I_{1,+}(t_{k})} + \frac{\mbox{Res}\big( \Delta_{R}^{(1)}1(\xi),\xi=I_{2,+}(t_{k}) \big)}{s-I_{2,+}(t_{k})} \bigg). \label{Rp1p of s}
\end{align}
These residues can be computed explicitly as follows. Define
\begin{align*}
J_{\mathrm{R}}(z) = \begin{pmatrix}
J_{P,11}(z) & J_{P,12}(z) \frac{P_{1}^{(\infty)}(z)}{P_{2}^{(\infty)}(z)} \\
J_{P,21}(z) \frac{P_{2}^{(\infty)}(z)}{P_{1}^{(\infty)}(z)} & J_{P,22}(z)
\end{pmatrix}, \qquad z \in \Sigma_{P}\setminus [a,b].
\end{align*}
In view of \eqref{P1inf in terms of F}--\eqref{P2inf in terms of F} and \eqref{jumps JR11 and JR21}--\eqref{jumps JR12 and JR22}, $J_{\mathrm{R}}$ and $J_{R}$ are related by
\begin{align}
& J_{\mathrm{R},j1}(J(s))=J_{R,j1}(s) \qquad  \mbox{for }s \in \Sigma', \; j=1,2, \label{JR JR 1} \\
& J_{\mathrm{R},j2}(J(s))=J_{R,j2}(s) \qquad \mbox{for } s \in \Sigma'', \; j=1,2. \label{JR JR 2}
\end{align}
From \eqref{expansion of ftk}, \eqref{comp tk 2} and Proposition \ref{prop: asymp of Pinf near tk}, we obtain
\begin{align*}
\mathrm{Res}\Big(J_{\mathrm{R}}^{(1)}(z),z=t_{k}\Big) & = \lim_{z \to t_{k}, z \in Q_{+,k}^{R}} (z-t_{k})J_{\mathrm{R}}^{(1)}(z) \\
& = \frac{v_{k}}{2\pi i \rho(t_{k})}  \begin{pmatrix}
-1 &  \frac{\tau(\alpha_{k},\beta_{k})\mathrm{E}_{t_{k}}(t_{k})^{2}}{C_{21,k}^{(\infty)}} \\  \frac{-\tau(\alpha_{k},-\beta_{k})C_{21,k}^{(\infty)}}{\mathrm{E}_{t_{k}}(t_{k})^{2}} & 1
\end{pmatrix}.
\end{align*}
By \eqref{def of Deltap1p 1}--\eqref{def of Deltap1p 2} and \eqref{JR JR 1}--\eqref{JR JR 2}, we thus find
\begin{align}
& \hspace{-0.18cm} \mbox{Res}\Big( \Delta_{R}^{(1)}1(\xi),\xi=I_{1,+}(t_{k}) \Big) \hspace{-0.1cm} = \hspace{-0.1cm} \frac{1}{J'(I_{1,+}(t_{k}))} \frac{-v_{k}}{2\pi i \rho(t_{k})} \bigg( \hspace{-0.1cm} 1 + \frac{\tau(\alpha_{k},-\beta_{k})C_{21,k}^{(\infty)}}{\mathrm{E}_{t_{k}}(t_{k})^{2}} \bigg), \label{res tk 1} \\
& \hspace{-0.18cm} \mbox{Res}\Big( \Delta_{R}^{(1)}1(\xi),\xi=I_{2,+}(t_{k}) \Big) \hspace{-0.1cm} = \hspace{-0.1cm} \frac{1}{J'(I_{2,+}(t_{k}))} \frac{v_{k}}{2\pi i \rho(t_{k})} \bigg( \hspace{-0.1cm} 1 \hspace{-0.05cm} + \hspace{-0.05cm} \frac{\tau(\alpha_{k},\beta_{k})\mathrm{E}_{t_{k}}(t_{k})^{2}}{C_{21,k}^{(\infty)}} \hspace{-0.05cm} \bigg). \label{res tk 2}
\end{align}
For the residue of $\Delta_{R}^{(1)}1(\xi)$ at $\xi = s_{b}$, we first use \eqref{asymp of conformal map near b}, \eqref{jumps for P on Db} and Proposition \ref{prop:asymp as z to b} to get
\begin{align*}
J_{\mathrm{R}}(z) = \frac{1}{16 (f_{b}^{(0)})^{1/2}\sqrt{z-b}}  \begin{pmatrix}
-1-4\alpha_{m+1}^{2} & \ds  -2 \\ \ds 2  & 1+4\alpha_{m+1}^{2}
\end{pmatrix}+\bigO(1), \qquad \mbox{as }z \to b.
\end{align*}
Hence, using \eqref{def of Deltap1p 1}, \eqref{JR JR 1} and \eqref{expansion of J near the endpoints b} (or alternatively \eqref{def of Deltap1p 2}, \eqref{JR JR 2} and \eqref{expansion of J near the endpoints b}), we obtain
\begin{align}\label{res sb unsimplified}
& \mbox{Res}\Big( \Delta_{R}^{(1)}1(\xi),\xi=s_{b} \Big) = \frac{1-4\alpha_{m+1}^{2}}{16(f_{b}^{(0)})^{1/2}\sqrt{J''(s_{b})/2}}.
\end{align}
The computation for the residue of $\Delta_{R}^{(1)}1(\xi)$ at $\xi = s_{a}$ is similar, and we find
\begin{align}\label{res sa unsimplified}
& \mbox{Res}\Big( \Delta_{R}^{(1)}1(\xi),\xi=s_{a} \Big) = \frac{4\alpha_{0}^{2}-1}{16(f_{a}^{(0)})^{1/2}\sqrt{|J''(s_{a})|/2}}.
\end{align}
The residues \eqref{res sb unsimplified} and \eqref{res sa unsimplified} can be simplified using the expansions of $\rho$ near $b$ and $a$ given by \eqref{asymptotics of rho near b} and \eqref{asymptotics of rho near a}. From \eqref{asymptotics of rho near b} and \eqref{asymp of conformal map near b}, we get
\begin{align*}
(f_{b}^{(0)})^{1/2} = \frac{\pi \psi(b)}{\sqrt{b-a}} = \pi \lim_{x \nearrow b} \rho(x)\sqrt{(b-x)} = \frac{1}{\sqrt{2}s_{b}\sqrt{J''(s_{b})}}
\end{align*}
which gives 
\begin{align}\label{res sb}
& \mbox{Res}\Big( \Delta_{R}^{(1)}1(\xi),\xi=s_{b} \Big) = s_{b}\frac{1-4\alpha_{m+1}^{2}}{8}.
\end{align}
Similarly, using \eqref{asymptotics of rho near a} in \eqref{res sa unsimplified}, we obtain
\begin{align}\label{res sa}
& \mbox{Res}\Big( \Delta_{R}^{(1)}1(\xi),\xi=s_{a} \Big) = s_{a}\frac{1-4\alpha_{0}^{2}}{8}.
\end{align}

\section{Proof of Theorem \ref{theorem J}}\label{section: final computation}
By \eqref{p ortho in bioortho}, \eqref{def of Cpj} and \eqref{def of Y}, as $z \to \infty$, $z \in \mathbb{H}_{\theta}$, we have
\begin{align}\label{Y2 asymp 1}
Y_{2}(z) = \frac{1}{\kappa_{n}}Cp_{n}(z) = -\frac{\kappa_{n}^{-2}}{2\pi i}z^{-(n+1)\theta} + \bigO(z^{-(n+2)\theta}).
\end{align}
On the other hand, using \eqref{Y to T transformation}, \eqref{T to S transformation}, \eqref{S to T transformation}, \eqref{P2 in terms of mathcal F} and \eqref{def of R} to invert the transformations $Y \mapsto T \mapsto S \mapsto P \mapsto \mathcal{F} \mapsto R$ for $z \in \mathbb{H}_{\theta}$, $z \notin \mathcal{L} \cup \bigcup_{j=0}^{m+1}\mathcal{D}_{t_{j}}$, we have
\small
\begin{align*}
T_{2}(z) \hspace{-0.025cm} = \hspace{-0.025cm} e^{n\ell}Y_{2}(z)e^{n\widetilde{g}(z)} \hspace{-0.025cm} = \hspace{-0.025cm} S_{2}(z) \hspace{-0.025cm} = \hspace{-0.025cm} P_{2}(z) \hspace{-0.025cm} = \hspace{-0.025cm} \mathcal{F}(I_{2}(z)) \hspace{-0.025cm} = \hspace{-0.025cm} R(I_{2}(z))F(I_{2}(z)) \hspace{-0.025cm} = \hspace{-0.025cm} R(I_{2}(z))P_{2}^{(\infty)}(z),
\end{align*}
\normalsize
where for the last equality we have used \eqref{P2inf in terms of F}. Let $\Omega$ be a fixed compact subset of \eqref{Omega compact subset}, and let us denote $\mathfrak{p}:=(\alpha_{0},\ldots,\alpha_{m+1},\beta_{1},\ldots,\beta_{m},t_{1},\ldots,t_{m},\theta)$. It follows from the analysis of Section \ref{section:small norm shifted RHP} that there exists $n_{0}=n_{0}(\Omega)$ such that $Y$ exist for all $n \geq n_{0}$ and all $\mathfrak{p} \in \Omega$. For clarity, we will write $R(z)=R(z;n)$ to make explicit the dependence of $R$ in $n$. Using Lemma \ref{lemma:asymp of I1 and I2 at infty}, Proposition \ref{prop:asymp of P2inf at inf}, and the fact that $\widetilde{g}(z) = \theta \log z + \bigO(z^{-\theta})$ as $z \to \infty$ in $\mathbb{H}_{\theta}$, we find
\begin{align}
Y_{2}(z) & = e^{-n\ell} e^{-n\widetilde{g}(z)} R(I_{2}(z);n) P_{2}^{(\infty)}(z) \label{asymp as z to inf} \\
& = e^{-n\ell} z^{-n\theta}(1+\bigO(z^{-\theta}))(R(0;n)+\bigO(z^{-\theta}))\frac{H(0)c_{0}}{-i\sqrt{|s_{a}s_{b}|}\theta} z^{-\theta} \big( 1 + \bigO(z^{-\theta}) \big), \nonumber
\end{align}
where the above expression is valid as $z \to \infty, \; z \in \mathbb{H}_{\theta}$, for all $n \geq n_{0}$ and all $\mathfrak{p} \in \Omega$. Comparing \eqref{asymp as z to inf} with \eqref{Y2 asymp 1}, we find
\begin{align*}
& \kappa_{n}^{-2} = 2\pi \, e^{-n\ell}R(0;n)\frac{H(0)c_{0}}{\sqrt{|s_{a}s_{b}|}\theta}, \qquad \mbox{for all } n \geq n_{0}, \; \mathfrak{p} \in \Omega.
\end{align*}
Hence, by \eqref{Dn in terms of the product}, we have
\begin{align}\label{product diff identity}
D_{N}(w) = D_{n_{0}}(w) \prod_{n=n_{0}}^{N-1} \kappa_{n}^{-2}, \qquad \mbox{for all } N \geq n_{0}, \; \mathfrak{p} \in \Omega.
\end{align}
Furthermore, since $\kappa_{n_{0}}^{-2}$ exists and is non-zero, this implies by \eqref{def of kappa k ^2} that  $D_{n_{0}}(w) \neq 0$. Note that $H(0)$ is independent of $n$ (see \eqref{def of H0}). Also, by \eqref{asymptotics of R as n to infty}, as $n \to + \infty$
\begin{align*}
& R(0;n) = 1 + \frac{R^{(1)}(0;n)}{n} + \bigO(n^{-2+4\beta_{\max}}), \qquad R^{(1)}(0;n)=\bigO(n^{2\beta_{\max}}).
\end{align*}
Hence, formula \eqref{product diff identity} can be rewritten as
\begin{align}
D_{N}(w) = & \; \exp \bigg( -\frac{\ell}{2}N^{2} + \bigg[ \frac{\ell}{2} + \log \frac{2\pi H(0)c_{0}}{\sqrt{|s_{a}s_{b}|}\theta} \bigg] N + C_{4}' \bigg) \nonumber \\
& \times \prod_{n=n_{0}}^{N}\bigg(1 + \frac{R^{(1)}(0;n)}{n} + \bigO(n^{-2+4\beta_{\max}})\bigg), \label{product diff identity 2}
\end{align}
for a certain constant $C_{4}'$, where the error term is uniform for all $n \geq n_{0}$ and all $\mathfrak{p} \in \Omega$. Using Proposition \ref{prop:simplified expression for ell}, the identity $s_{a}s_{b} = -\frac{c_{0}}{c_{1}\theta}$, and the expression \eqref{def of H0} for $H(0)$, we verify that
\begin{align*}
-\frac{\ell}{2} = C_{1}, \qquad \frac{\ell}{2} + \log \frac{2\pi H(0)c_{0}}{\sqrt{|s_{a}s_{b}|}\theta} = C_{2},
\end{align*}
where $C_{1}$ and $C_{2}$ are given by \eqref{C1 thm} and \eqref{C2 thm}, respectively. Our next task is to obtain an asymptotic formula for the product in \eqref{product diff identity 2} as $N \to +\infty$. By \eqref{Rp1p of s},
\begin{align*}
& R^{(1)}(0;n) = \frac{-1}{s_{a}}\mbox{Res}\Big( \Delta_{R}1(\xi),\xi=s_{a} \Big) + \frac{-1}{s_{b}}\mbox{Res}\Big( \Delta_{R}1(\xi),\xi=s_{b} \Big) \\
& + \sum_{k=1}^{m} \bigg( \frac{-1}{I_{1,+}(t_{k})}\mbox{Res}\Big( \Delta_{R}1(\xi),\xi=I_{1,+}(t_{j}) \Big) + \frac{-1}{I_{2,+}(t_{k})}\mbox{Res}\Big( \Delta_{R}1(\xi),\xi=I_{2,+}(t_{k}) \Big) \bigg),
\end{align*}
and using \eqref{res tk 1}, \eqref{res tk 2}, \eqref{res sb} and \eqref{res sa}, we get
\begin{align*}
& R^{(1)}(0;n) = \frac{4\alpha_{0}-1}{8} + \frac{4\alpha_{m+1}-1}{8} \\
& + \sum_{k=1}^{m} \frac{\beta_{k}^{2}-\frac{\alpha_{k}^{2}}{4}}{2\pi i \rho(t_{k})} \bigg( \frac{1}{I_{1,+}(t_{k})J'(I_{1,+}(t_{k}))}-\frac{1}{I_{2,+}(t_{k})J'(I_{2,+}(t_{k}))} \bigg) \\
& + \sum_{k=1}^{m} \frac{\beta_{k}^{2}-\frac{\alpha_{k}^{2}}{4}}{2\pi i \rho(t_{k})} \bigg( \frac{\tau(\alpha_{k},-\beta_{k})C_{21,k}^{(\infty)}}{I_{1,+}(t_{k})J'(I_{1,+}(t_{k}))\mathrm{E}_{t_{k}}(t_{k};n)^{2}} - \frac{\tau(\alpha_{k},\beta_{k})\mathrm{E}_{t_{k}}(t_{k};n)^{2}}{I_{2,+}(t_{k})J'(I_{2,+}(t_{k}))C_{21,k}^{(\infty)}} \bigg),
\end{align*}
where we have explicitly written the dependence of $\mathrm{E}_{t_{k}}(t_{k})$ in $n$. Using $J'(I_{j,+}(t_{k})) = I_{j,+}'(t_{k})^{-1}$ for $k=1,\ldots,m$, $j=1,2$ and \eqref{rho in the bulk}, we obtain 
\begin{align*}
& \frac{1}{2\pi i \rho(t_{k})} \bigg( \frac{1}{I_{1,+}(t_{k})J'(I_{1,+}(t_{k}))}-\frac{1}{I_{2,+}(t_{k})J'(I_{2,+}(t_{k}))} \bigg) \\
& = \frac{1}{2\pi i \rho(t_{k})}\bigg( \frac{I_{1,+}'(t_{k})}{I_{1,+}(t_{k})} - \frac{I_{2,+}'(t_{k})}{I_{2,+}(t_{k})} \bigg) = -1,
\end{align*}
and therefore $R^{(1)}(0;n)$ can be rewritten as
\begin{align*}
& R^{(1)}(0;n) = C_{3} \\
& + \sum_{k=1}^{m} \frac{\beta_{k}^{2}-\frac{\alpha_{k}^{2}}{4}}{2\pi i \rho(t_{k})} \bigg( \frac{\tau(\alpha_{k},-\beta_{k})C_{21,k}^{(\infty)}}{I_{1,+}(t_{k})J'(I_{1,+}(t_{k}))\mathrm{E}_{t_{k}}(t_{k};n)^{2}} - \frac{\tau(\alpha_{k},\beta_{k})\mathrm{E}_{t_{k}}(t_{k};n)^{2}}{I_{2,+}(t_{k})J'(I_{2,+}(t_{k}))C_{21,k}^{(\infty)}} \bigg),
\end{align*}
where $C_{3}$ is given by \eqref{C3 thm}. From \eqref{mathrmE at tk}, we see that $\mathrm{E}_{t_{k}}(t_{k};n)^{2} = \bigO(n^{2\beta_{k}})$ as $n \to \infty$. However, $\phi(b)=0$ and \eqref{phi der} imply that $-i\phi_{+}(t_{k}) \in (0,2\pi)$ for all $k=1,\ldots,m$, which in turn implies that $\mathrm{E}_{t_{k}}(t_{k};n)^{2}$ oscillates quickly as $n \to +\infty$, and more precisely that
\begin{align*}
\prod_{n=n_{0}}^{N} \big( 1 + \mathrm{E}_{t_{k}}(t_{k};n)^{\pm 2}n^{-1} \big) = C_{4,\pm} + \bigO(N^{-1\pm 2 \beta_{k}}), \qquad \mbox{as } N \to + \infty,
\end{align*}
where $C_{4,\pm}$ are some constants. Hence, as $N \to + \infty$,
\begin{align}\label{reason why we have need beta<1/4}
\prod_{n=n_{0}}^{N}\bigg(1 + \frac{R^{(1)}(0;n)}{n} + \bigO(n^{-2+4\beta_{\max}})\bigg) = C_{3} \log N + C_{4}'' + \bigO(N^{-1+4\beta_{\max}}), 
\end{align}
for a certain constant $C_{4}''$, which finishes the proof of \eqref{asymp thm Jn}.

\section{Proof of Theorem \ref{thm:rigidity}}\label{Section: rigidity}
Let $x_{1},\ldots,x_{n}$ be distributed according to the Muttalib-Borodin ensemble \eqref{MB density}, and recall that the counting function is denoted by $N_{n}(t) = \#\{x_{j}: x_{j}\leq t\}$, $t \geq 0$, and that the ordered points are denoted by $a \leq \xi_{1} \leq \xi_{2} \leq \ldots \leq \xi_{n} \leq b$. 

\medskip  Parts (a) and (b) of Theorem \ref{thm:rigidity} can be proved in a similar way as in \cite[Corollaries 1.2 and 1.3]{ChBessel}. For part (a), we first set $m=1$ in \eqref{moment generating function in introduction} (and rename $t_{1}\to t$, $\alpha_{1}\to \alpha$, $2\pi i \beta_{1}\to \gamma$):
\begin{align}\label{moment generating function in last section}
\mathbb{E} \bigg( |p_{n}(t)|^{\alpha} e^{\gamma N_{n}(t)} \bigg) = \frac{D_{n}(w)|_{m=1}}{D_{n}(\mathsf{w})}e^{\frac{\gamma}{2} n}.
\end{align}
Let $h(\alpha,\beta)=h(\alpha,\beta;n)$ denote the right-hand side of \eqref{moment generating function in last section}. Theorem \ref{theorem J} gives the formula
\begin{align}\label{asymp for h}
h(\alpha,\beta) \hspace{-0.05cm} = \hspace{-0.05cm} \exp \bigg( \hspace{-0.1cm} \alpha \hspace{-0.1cm} \int_{a}^{b} \hspace{-0.1cm} \log|t-x|\rho(x)dx \, n + \gamma \hspace{-0.1cm} \int_{a}^{t_{j}}\hspace{-0.2cm} \rho(x)dx \, n + \bigg( \frac{\alpha^{2}}{4}+\frac{\gamma^{2}}{4\pi^{2}} \bigg) \log n + \bigO(1) \hspace{-0.1cm} \bigg),
\end{align}
as $n \to + \infty$, and these asymptotics are uniform for $\alpha$ and $\gamma$ in complex neighborhood of $0$. Since $h(\alpha,\beta)$ is analytic in $\alpha$ and $\beta$, this implies, by Cauchy's formula, that the asymptotics \eqref{asymp for h} can be differentiated any number of times without worsening the error term. Hence, differentiating \eqref{moment generating function in last section} and \eqref{asymp for h} once with respect to $\alpha$ and then evaluating at $\alpha=0$, as $n \to + \infty$ we obtain
\begin{align*}
\partial_{\alpha}\mathbb{E} \bigg( |p_{n}(t)|^{\alpha} e^{\gamma N_{n}(t)} \bigg)\bigg|_{\alpha=0} = \mathbb{E}(\log |p_{n}(t)|) = \int_{a}^{b} \log|t-x|\rho(x)dx \, n + \bigO(1),
\end{align*}
which is \eqref{asymp expectation ln |pn|}. Formula \eqref{asymp expectation} is obtained similarly by differentiating \eqref{moment generating function in last section} and \eqref{asymp for h} once with respect to $\gamma$, and the asymptotics \eqref{asymp variance} are obtained by taking the second derivatives with respect to $\alpha$ and $\gamma$. 

\medskip Now, we prove part (b) of Theorem \ref{thm:rigidity}. Since $D_{n}(w)$ is analytic in $\alpha_{1},\ldots,\alpha_{m}$, $\beta_{1},\ldots,\beta_{m}$, Theorem \ref{theorem J} implies that
\begin{align}\label{lol3}
\frac{D_{n}(w)}{D_{n}(\mathsf{w})}\prod_{k=1}^{m}e^{i\pi n \beta_{k}}  = \prod_{k=1}^{m}e^{\alpha_{k} n \int_{a}^{b}\log|t_{k}-x|\rho(x)dx}e^{2\pi i \beta_{k} n \int_{a}^{t_{k}}\rho(x)dx}n^{\frac{\alpha_{k}^{2}}{4}-\beta_{k}^{2}}\mathcal{H}_{n},
\end{align}
where $\mathcal{H}_{n}$ is analytic in $\alpha_{1},\ldots,\alpha_{m},\beta_{1},\ldots,\beta_{m}$, satisfies $\mathcal{H}_{n}|_{\alpha_{1}=\ldots=\alpha_{m}=\beta_{1}=\ldots=\beta_{m}=0}$ $=1$, and is bounded as $n \to +\infty$ uniformly for $\alpha_{1},\ldots,\alpha_{m},\beta_{1},\ldots,\beta_{m}$ in small neighborhoods of $0$. This implies, again by Cauchy's formula, that all the derivatives of $\mathcal{H}_{n}$ with respect to $\alpha_{j},\beta_{j}$ are also bounded as $n \to + \infty$ uniformly for $\alpha_{1},\ldots,\alpha_{m},\beta_{1},\ldots,\beta_{m}$ in small neighborhoods of $0$. Let $a_{1},\ldots,a_{m},b_{1},\ldots,b_{m} \in \mathbb{R}$ be arbitrary but fixed. Hence, using \eqref{lol3} with
\begin{align*}
\alpha_{k}= \sqrt{2} \frac{a_{k}}{\sqrt{\log n}}, \qquad 2\pi i\beta_{k} = \sqrt{2}\pi \frac{b_{k}}{\sqrt{\log n}}, \qquad k=1,\ldots,m,
\end{align*}
and using also \eqref{moment generating function in introduction} and \eqref{def of Mn and Nn 1}--\eqref{def of Mn and Nn 2}, as $n \to + \infty$ we obtain
\begin{align}\label{lol4}
\mathbb{E}\bigg[ \prod_{j=1}^{m}e^{a_{j}\mathcal{M}_{n}(t_{j}) + b_{j} \mathcal{N}_{n}(t_{j})} \bigg] = \exp \bigg( \sum_{j=1}^{m} \bigg( \frac{a_{j}^{2}}{2}+\frac{b_{j}^{2}}{2} \bigg) + \bigO\bigg( \frac{1}{\sqrt{\log n}} \bigg) \bigg).
\end{align}
Since $a_{1},\ldots,a_{m},b_{1},\ldots,b_{m} \in \mathbb{R}$ were arbitrary, this implies the convergence in distribution \eqref{convergence in distribution 1}.

\medskip We now turn to the proof of part (c) of Theorem \ref{thm:rigidity}. Our proof is inspired by Gustavsson \cite[Theorem 1.2]{Gustavsson}. Let $k_{j}=[n \int_{a}^{t_{j}}\rho(x)dx]$, $j=1,\ldots,m$, and consider the random variables $Y_{n}(t_{j})$ defined by
\begin{align}\label{Yn tj def}
Y_{n}(t_{j}) = \sqrt{2}\pi \frac{n\int_{a}^{\xi_{k_{j}}}\rho(x)dx - k_{j}}{\sqrt{\log n}} = \frac{\mu_{n}(\xi_{k_{j}})-k_{j}}{\sigma_{n}}, \qquad j=1,\ldots,m,
\end{align} 
where $\mu_{n}(t) := n\int_{a}^{t}\rho(x)dx$, $\sigma_{n} := \frac{1}{\sqrt{2}\pi}\sqrt{\log n}$. 
Given $y_{1},\ldots,y_{m} \in \mathbb{R}$, we have
\begin{align}
& \mathbb{P}\big[ Y_{n}(t_{j}) \leq y_{j} \mbox{ for all } j=1, \ldots,m \big] = \mathbb{P}\Big[\xi_{k_{j}} \leq \mu_{n}^{-1}\big(k_{j} + y_{j} \sigma_{n}\big) \mbox{ for all } j=1, \ldots,m \Big], \nonumber \\
& = \mathbb{P}\Big[N_{n}\Big(\mu_{n}^{-1}\big(k_{j} + y_{j} \sigma_{n} \big)\Big) \geq k_{j} \mbox{ for all } j=1, \ldots,m \Big]. \label{prob1}
\end{align}
For $j=1,\ldots,m$, let $\tilde{t}_{j} := \mu_{n}^{-1}\big(k_{j} + y_{j} \sigma_{n} \big)$. As $n \to +\infty$, we have 
\begin{align}\label{tj tilde remain bounded away from each other}
k_{j} = [\mu_{n}(t_{j})] = \bigO(n), \qquad \tilde{t}_{j}=t_{j}\Big(1+\bigO\Big(\tfrac{\sqrt{\log n}}{n}\Big)\Big). 
\end{align} 
Since Theorem \ref{theorem J} holds also in the case where $t_{1},\ldots,t_{m}$ depend on $n$ but remain bounded away from each other (see \eqref{assumption on tj delta}), note that the same is true for \eqref{lol4}, and therefore also for the convergence in distribution \eqref{convergence in distribution 1}. Now, we rewrite \eqref{prob1} as
\begin{align*}
\mathbb{P}\big[ Y_{n}(t_{j}) \leq y_{j} \mbox{ for all } j=1, ...,m \big] & = \mathbb{P}\bigg[ \frac{N_{n}(\tilde{t}_{j})-\mu_{n}(\tilde{t}_{j})}{\sqrt{\sigma_{n}^{2}}} \geq \frac{k_{j}-\mu_{n}(\tilde{t}_{j})}{\sqrt{\sigma_{n}^{2} }}, \; j=1, ...,m \bigg] \\
& = \mathbb{P}\bigg[ \frac{\mu_{n}(\tilde{t}_{j})-N_{n}(\tilde{t}_{j})}{\sqrt{\sigma_{n}^{2}}} \leq y_{j} \mbox{ for all } j=1, \ldots,m \bigg].
\end{align*}
By \eqref{tj tilde remain bounded away from each other}, the parameters $\tilde{t}_{1},\ldots,\tilde{t}_{m}$ remain bounded away from each other, and therefore Theorem \ref{thm:rigidity} (b) implies that $\big( Y_{n}(t_{1}),Y_{n}(t_{2}),\ldots,Y_{n}(t_{m})\big)  \smash{\overset{d}{\longrightarrow}} \mathsf{N}(\vec{0},I_{m})$. Now, using the definitions \eqref{def of Zn} and \eqref{Yn tj def} of $Z_{n}(t_{j})$ and $Y_{n}(t_{j})$, we obtain \vspace{-0.1cm}
\small
\begin{align*}
& \mathbb{P}\big[ Z_{n}(t_{j}) \leq y_{j}, \; j=1, ...,m \big] = \mathbb{P}\bigg[ Y_{n}(t_{j}) \leq  \frac{\mu_{n}(\kappa_{k_{j}}+y_{j} \frac{\sigma_{n}}{n\rho(\kappa_{k_{j}})})-\mu_{n}(\kappa_{k_{j}})}{\sigma_{n}}, \; j=1, ...,m \bigg] \\
&   = \mathbb{P}\big[ Y_{n}(t_{j}) \leq y_{j}+o(1) \mbox{ for all } j=1, \ldots,m \big]
\end{align*}
\normalsize
as $n\to + \infty$, which implies the convergence in distribution \eqref{convergence in distribution 2}.

\medskip The rest of this section is devoted to the proof of Theorem \ref{thm:rigidity} (d), and is inspired from \cite{ChCl4}. We first prove \eqref{probabilistic upper bound 1} in Lemma \ref{lemma: A r eps} below. The proof of \eqref{probabilistic upper bound 2} is given at the end of this section.

\medskip Combining \eqref{moment generating function in introduction} and Theorem \ref{theorem J} with $m=1$, $\alpha_{1}=0$ and $\beta_{1} \in i \mathbb{R}$, and setting $\gamma:=2\pi i \beta_{1}$ and $t:=t_{1}$, we infer that for any $\delta \in (0,\frac{b-a}{2})$ and $M>0$, there exists $n_{0}'=n_{0}'(\delta,M)\in \mathbb{N}$ and $\mathrm{C}=\mathrm{C}(\delta,M)>0$ such that
\begin{align}\label{expmomentbound}
\mathbb{E} \big( e^{\gamma N_{n}(t)} \big) \leq  \mathrm{C} \exp \bigg( \gamma \mu_{n}(t) +\frac{\gamma^{2}}{2}\sigma_{n}^{2} \bigg), \; \mu_{n}(t) = n\int_{a}^{t}\rho(x)dx, \; \sigma_{n} = \frac{1}{\sqrt{2}\pi}\sqrt{\log n},
\end{align}
for all $n\geq n_{0}'$, $t \in (a+\delta,b-\delta)$ and $\gamma \in [-M,M]$. 

\begin{lemma}\label{lemma: A r eps}
For any $\delta \in (0,\frac{b-a}{2})$, there exist $c>0$ such that for all large enough $n$ and small enough $\epsilon>0$,
\begin{align}\label{prob statement lemma 2.1}
\mathbb P\left(\sup_{a+\delta \leq x \leq b-\delta}\bigg|\frac{N_{n}(x)-\mu_{n}(x)}{\sigma^2_{n}}\bigg|\leq 2\pi\sqrt{1+\epsilon} \right) \geq 1-cn^{-\epsilon}.
\end{align}
\end{lemma}
\begin{proof}
Recall that $\kappa_{k}= \mu_{n}^{-1}(k)$ is the classical location of the $k$-th smallest point $\xi_{k}$ and is defined in \eqref{def of kappa k}. Since $\mu_{n}$ and $N_{n}$ are increasing function, for $x\in[\kappa_{k-1},\kappa_k]$ with $k \in \{1,\ldots,n\}$, we have
\begin{equation}\label{lol1}
N_{n}(x)-\mu_{n}(x)\leq N_{n}(\kappa_k)-\mu_{n}(\kappa_{k-1})
=N_{n}(\kappa_k)-\mu_{n}(\kappa_{k})+1,
\end{equation}
which implies
\begin{align*}
\sup_{a+\delta \leq x \leq b-\delta}\frac{N_{n}(x)-\mu_{n}(x)}{\sigma_{n}^2} \leq \sup_{k\in \mathcal{K}_{n}}
\frac{N_{n}(\kappa_k)-\mu_{n}(\kappa_{k})+1}{\sigma_{n}^2},
\end{align*}
where $\mathcal{K}_{n} = \{k: \kappa_{k}>a+\delta \mbox{ and } \kappa_{k-1}<b-\delta\}$. Using a union bound, for any $\gamma > 0$ we find
\begin{multline}\label{lol2}
\mathbb P\left(\sup_{a+\delta \leq x \leq b-\delta}\frac{N_{n}(x)-\mu_{n}(x)}{\sigma^2_{n}}>\gamma\right)\leq \sum_{k\in \mathcal{K}_{n}}\mathbb P\left(\frac{N_{n}(\kappa_k)-\mu_{n}(\kappa_{k})+1}{\sigma_{n}^2}>\gamma\right) \\ = \sum_{k\in \mathcal{K}_{n}} \mathbb{P}\left( e^{\gamma N_{n}(\kappa_{k})} > e^{\gamma \mu_{n}(\kappa_{k})-\gamma + \gamma^{2} \sigma_{n}^{2}} \right) \leq \sum_{k\in \mathcal{K}_{n}}\mathbb E\left(e^{\gamma N_{n}(\kappa_k)}\right)e^{-\gamma\mu_{n}(\kappa_k)+\gamma-\gamma^2 \sigma^2_{n}},
\end{multline}
where for the last step we have used Markov's inequality. Using \eqref{expmomentbound}, \eqref{lol2} and the fact that $\#\mathcal{K}_{n}$ is proportional to $n$ as $n \to +\infty$, for any fixed $M>0$ we obtain
\begin{align}\label{bound2 sup}
\hspace{-0.15cm}\mathbb P\left(\sup_{a+\delta \leq x \leq b-\delta}\hspace{-0.15cm}\frac{N_{n}(x)-\mu_{n}(x)}{\sigma^2_{n}}>\gamma\right)\leq 
\mathrm{C}(\delta,M) \, 
e^\gamma e^{-\frac{\gamma^2}{2} \sigma_{n}^2} \sum_{k\in \mathcal{K}_{n}} 1 \leq c_{1} n^{1 - \frac{\gamma^{2}}{4\pi^{2}}} 
\end{align}
for all large enough $n$ and $\gamma \in (0,M]$, where $c_{1}=c_{1}(\delta,M)>0$ is independent of $n$. We show similarly that, for any $M>0$,
\begin{align}\label{bound2 sup lol}
\mathbb P\left(\sup_{a+\delta \leq x \leq b-\delta}\frac{\mu_{n}(x)-N_{n}(x)}{\sigma_{n}^2}>\gamma\right)\leq c_{2} n^{1 - \frac{\gamma^{2}}{4\pi^{2}}},
\end{align}
for all large enough $n$ and $\gamma \in (0,M]$, and where $c_{2}=c_{2}(\delta,M)>0$ is independent of $n$. Taking together \eqref{bound2 sup} and \eqref{bound2 sup lol} with $M=4\pi$ (in fact any other choice of $M>2\pi$ would be sufficient for us), we get
\begin{align*}
\mathbb P\left(\sup_{a+\delta \leq x \leq b-\delta}\bigg|\frac{N_{n}(x)-\mu_{n}(x)}{\sigma^2_{n}}\bigg|>\gamma\right) \leq \max\{c_{1}(\delta,4\pi),c_{2}(\delta,4\pi)\} \; n^{1 - \frac{\gamma^{2}}{4\pi^{2}}},
\end{align*}
for all sufficiently large $n$ and for any $\gamma \in (0,4\pi]$. Clearly, the right-hand side converges to $0$ as $n\to +\infty$ for any $\gamma>2\pi$. We obtain the claim after taking $\gamma = 2\pi\sqrt{1+\epsilon}$ and setting $c=\max\{c_{1}(\delta,4\pi),c_{2}(\delta,4\pi)\}$.
\end{proof}
\begin{lemma}\label{lemma: xk not far away from kappa k}
Let $\delta \in (0,\frac{b-a}{4})$ and $\epsilon > 0$. For all sufficiently large $n$, if the event
\begin{align}\label{event holds true}
\sup_{a+\delta \leq x \leq b-\delta}\left|\frac{N_{n}(x)-\mu_{n}(x)}{\sigma^2_{n}}\right| \leq 2\pi\sqrt{1+\epsilon}
\end{align}
holds true, then we have
\begin{align}\label{upper and lower bound for mu xk}
\sup_{k \in (\mu_{n}(a+2\delta),\mu_{n}(b-2\delta))} \bigg|\frac{\mu_{n}(\xi_k) - k}{\sigma^2_{n}}\bigg| \leq 2\pi \sqrt{1+\epsilon} + \frac{1}{\sigma_{n}^{2}},
\end{align}

\end{lemma}

\begin{proof}
We first show that
\begin{align}\label{intermediate proof}
\xi_{k} \in (a+\delta,b-\delta), \qquad \mbox{for all } k \in (\mu_{n}(a+2\delta),\mu_{n}(b-2\delta))
\end{align}
and for all large enough $n$. Assume that $\xi_{k} \leq a+\delta<a+2\delta \leq \kappa_{k}$. Since $\mu_{n}$ and $N_{n}$ are increasing,
\begin{align*}
\mu_{n}(a+2\delta) \leq  \mu_{n}(\kappa_k)=k=N_{n}(\xi_k) \leq N_{n}(a+\delta),
\end{align*}
and therefore
\begin{align*}
& \frac{N_{n}(a+\delta)-\mu_{n}(a+\delta)}{\sigma^2_{n}}\geq \frac{\mu_{n}(a+2\delta)-\mu_{n}(a+\delta)}{\sigma_{n}^2}\geq \frac{\delta\inf_{a+\delta\leq \xi\leq a+2\delta}\mu_{n}'(\xi)}{\sigma^2_{n}}.
\end{align*}
Since $\mu_{n}'=n \rho$, the right-hand side tends to $+\infty$ as $n \to +\infty$, which contradicts \eqref{event holds true} for large enough $n$. Similarly, if $\xi_{k} \geq b-\delta > b-2\delta \geq \kappa_{k}$, then 
\begin{align*}
\mu_{n}(b-2\delta) \geq  \mu_{n}(\kappa_k)=k=N_{n}(\xi_k) \geq N_{n}(b-\delta),
\end{align*}
and we find
\begin{align*}
& \frac{\mu_{n}(b-\delta)-N_{n}(b-\delta)}{\sigma^2_{n}}\geq \frac{\mu_{n}(b-\delta)-\mu_{n}(b-2\delta)}{\sigma_{n}^2}\geq \frac{\delta\inf_{b-2\delta\leq \xi\leq b-\delta}\mu_{n}'(\xi)}{\sigma^2_{n}},
\end{align*}
which again contradicts \eqref{event holds true} for sufficiently large $n$. We conclude that \eqref{intermediate proof} holds for all large enough $n$. 

\medskip Now, we prove \eqref{upper and lower bound for mu xk} in two steps. First, we show that
\begin{align}\label{upper bound in proof}
\mu_{n}(\xi_k)\leq k+1+ 2\pi\sqrt{1+\epsilon} \; \sigma^2_{n}, \qquad \mbox{for all } k \in (\mu_{n}(a+2\delta),\mu_{n}(b-2\delta)),
\end{align}
and for all large enough $n$. For this, let $m = m(k) \in \mathbb{Z}$ be such that $\kappa_{k+m}<\xi_k\leq \kappa_{k+m+1}$. The inequality \eqref{upper bound in proof} is automatically verified for $m < 0$. Now, we consider the case $m \geq 0$. Since $k \in (\mu_{n}(a+2\delta),\mu_{n}(b-2\delta))$, we know from \eqref{intermediate proof} that $\xi_{k} \in (a+\delta,b-\delta)$ for all sufficiently large $n$, so we can use \eqref{event holds true} to obtain
\begin{align*}
2\pi\sqrt{1+\epsilon}\geq\frac{\mu_{n}(\xi_k)-N_{n}(\xi_k)}{\sigma^2_{n}}\geq \frac{m}{\sigma^2_{n}}, \qquad \mbox{i.e.} \qquad m\leq 2\pi \sqrt{1+\epsilon} \, \sigma^2_{n},
\end{align*}
where the above inequality is valid for all sufficiently large $n$. Hence,
\begin{align*}
\mu_{n}(\xi_k)\leq \mu_{n}(\kappa_{k+m+1}) = k+m+1\leq k +1 +2\pi\sqrt{1+\epsilon} \, \sigma^2_{n},
\end{align*}
which proves \eqref{upper bound in proof}. Our next goal is to prove the following complementary lower bound for $\mu(\xi_{k})$:
\begin{align}\label{lower bound in proof}
k- 2\pi \sqrt{1+\epsilon} \, \sigma^2_{n} \leq \mu_{n}(\xi_k), \qquad \mbox{for all } k \in (\mu_{n}(a+2\delta),\mu_{n}(b-2\delta))
\end{align}
for all large enough $n$. Let us assume $\mu_{n}(\xi_k)<k-m$ with $m>0$. Using \eqref{intermediate proof} with \eqref{event holds true}, for all large enough $n$ we obtain
\begin{align*}
2\pi\sqrt{1+\epsilon}\geq \frac{N_{n}(\xi_k)-\mu_{n}(\xi_k)}{\sigma^2_{n}}>\frac{m}{\sigma^2_{n}}, \qquad \mbox{for all } k \in (\mu_{n}(a+2\delta),\mu_{n}(b-2\delta)).
\end{align*}
In particular, we get $m< 2\pi\sqrt{1+\epsilon} \, \sigma^2_{n}$, which yields \eqref{lower bound in proof} and finishes the proof.
\end{proof}

We can now prove \eqref{probabilistic upper bound 2} by combining Lemmas \ref{lemma: A r eps} and \ref{lemma: xk not far away from kappa k}.
\begin{proof}[Proof of \eqref{probabilistic upper bound 2}]
By Lemma \ref{lemma: A r eps}, for any $\delta' \in (0,\frac{b-a}{4})$, there exists $c>0$ such that for all small enough $\epsilon > 0$ and for all large enough $n$, we have
\begin{align}\label{bayes1}
\mathbb{P}\left( \sup_{a+\delta' \leq x \leq b-\delta'} \bigg|\frac{N_{n}(x)-\mu_{n}(x)}{\sigma^{2}_{n}}\bigg| \leq 2\pi\sqrt{1+\epsilon} \right) \geq 1 - cn^{-\epsilon}.
\end{align}
On the other hand, by Lemma \ref{lemma: xk not far away from kappa k} we have
\begin{align}\label{bayes2}
\mathbb{P}\left( A \; \bigg| \; \sup_{a+\delta' \leq x \leq b-\delta'} \frac{|N_{n}(x)-\mu_{n}(x)|}{\sigma^{2}_{n}} \leq 2\pi\sqrt{1+\epsilon} \right) = 1,
\end{align}
for all sufficiently large $n$, where $A$ is the event that
\begin{align*}
\sup_{k \in (\mu_{n}(a+2\delta'),\mu_{n}(b-2\delta'))} \frac{|\mu_{n}(\xi_k) - k|}{\sigma^2_{n}} \leq 2\pi \sqrt{1+\epsilon} + \frac{1}{\sigma_{n}^{2}}.
\end{align*}
Let $\delta>0$ be arbitrarily small but fixed. By applying Bayes' formula on \eqref{bayes1} and \eqref{bayes2} (with $\delta'$ chosen such that $\mu_{n}(a+2\delta') \leq \delta n$ and $(1-\delta)n\leq \mu_{n}(b-2\delta')$), we conclude that there exists $c>0$ such that
\begin{align}\label{lol5}
\mathbb{P}\bigg( \max_{\delta n \leq k \leq (1-\delta)n} \bigg| \int_{a}^{\xi_{k}}\rho(x)dx - \frac{k}{n}\bigg| \leq \frac{\sqrt{1+\epsilon}}{\pi} \frac{\log n}{n} +\frac{1}{n} \bigg) \geq 1-cn^{-\epsilon},
\end{align}
for all sufficiently large $n$. Note that the $\frac{1}{n}$ in the above upper bound is unimportant; it can be removed at the cost of multiplying $c$ by a factor larger than $e^{2\pi \sqrt{1+\epsilon}}$. More precisely, \eqref{lol5} implies
\begin{align*}
\mathbb{P}\bigg( \max_{\delta n \leq k \leq (1-\delta)n} \bigg| \int_{a}^{\xi_{k}}\rho(x)dx - \frac{k}{n}\bigg| \leq \frac{\sqrt{1+\epsilon}}{\pi} \frac{\log n}{n} \bigg) \geq 1-c'n^{-\epsilon},
\end{align*}
for all sufficiently large $n$, where $c'=2 e^{2\pi \sqrt{1+\epsilon}}c$. Hence, for any small enough $\delta>0$ and $\epsilon > 0$, there exists $c>0$ such that
\begin{align*}
& \mathbb{P}\bigg( \max_{\delta n \leq k \leq (1-\delta)n}  \rho(\kappa_{k})|\xi_{k}-\kappa_{k}| \leq \frac{\sqrt{1+\epsilon}}{\pi} \frac{\log n}{n} \bigg) = \mathbb{P}\bigg( \frac{\smash{\mu_{n}(\kappa_{k}-\frac{\sqrt{1+\epsilon}}{\pi}\frac{\log n}{n \rho(\kappa_{k})})-k}}{n}  \\
& \qquad \leq \frac{\mu_{n}(\xi_{k})-k}{n} \leq \frac{\mu_{n}(\kappa_{k}+\frac{\sqrt{1+\epsilon}}{\pi}\frac{\log n}{n \rho(\kappa_{k})})-k}{n}, \; \mbox{for all } k \in (\delta n,(1-\delta)n) \bigg) \\
& \geq \mathbb{P}\bigg( \max_{\delta n \leq k \leq (1-\delta)n} \bigg| \int_{a}^{\xi_{k}}\rho(x)dx - \frac{k}{n}\bigg| \leq \frac{\sqrt{1+\epsilon}}{\pi} \frac{\log n}{n} - \frac{1}{n} \bigg) \geq 1-cn^{-\epsilon},
\end{align*}
for all sufficiently large $n$, which completes the proof of \eqref{probabilistic upper bound 2}.
\end{proof}

\appendix 

\section{Model RH problems}
In this section, $\alpha$ and $\beta$ are such that $\re \alpha >-1$ and $\re \beta \in (-\frac{1}{2},\frac{1}{2})$. 
\subsection{Bessel model RH problem for $\Phi_{\mathrm{Be}}(\cdot) = \Phi_{\mathrm{Be}}(\cdot;\alpha)$}\label{ApB}
\begin{itemize}
\item[(a)] $\Phi_{\mathrm{Be}} : \mathbb{C} \setminus \Sigma_{\mathrm{Be}} \to \mathbb{C}^{2\times 2}$ is analytic, where
$\Sigma_{\mathrm{Be}} = (-\infty,0]\cup e^{\frac{2\pi i}{3}}(0,+\infty) \cup e^{-\frac{2\pi i}{3}}(0,+\infty)$ and is oriented as shown in Figure \ref{fig:Bessel}.
\item[(b)] $\Phi_{\mathrm{Be}}$ satisfies the jump relations
\begin{equation}\label{Jump for P_Be}
\begin{array}{l l} 
\Phi_{\mathrm{Be},+}(z) = \Phi_{\mathrm{Be},-}(z) \begin{pmatrix}
0 & 1 \\ -1 & 0
\end{pmatrix}, & z \in (-\infty,0), \\

\Phi_{\mathrm{Be},+}(z) = \Phi_{\mathrm{Be},-}(z) \begin{pmatrix}
1 & 0 \\ e^{\pi i \alpha} & 1
\end{pmatrix}, & z \in e^{ \frac{2\pi i}{3} }  (0,+\infty), \\

\Phi_{\mathrm{Be},+}(z) = \Phi_{\mathrm{Be},-}(z) \begin{pmatrix}
1 & 0 \\ e^{-\pi i \alpha} & 1
\end{pmatrix}, & z \in e^{ -\frac{2\pi i}{3} }  (0,+\infty). \\
\end{array}
\end{equation}
\item[(c)] As $z \to \infty$, $z \notin \Sigma_{\mathrm{Be}}$, 
\begin{equation}\label{large z asymptotics Bessel}
\Phi_{\mathrm{Be}}(z) = ( 2\pi z^{\frac{1}{2}} )^{-\frac{\sigma_{3}}{2}}A
\left(I+\sum_{k=1}^{\infty} \Phi_{\mathrm{Be},k} z^{-k/2}\right) e^{2z^{\frac{1}{2}}\sigma_{3}}, \qquad A = \frac{1}{\sqrt{2}}\begin{pmatrix}
1 & i \\ i & 1
\end{pmatrix},
\end{equation}
where the matrices $\Phi_{\mathrm{Be},k}$ are independent of $z$, and
\begin{align}\label{Phi BE 1}
\Phi_{\mathrm{Be},1} = \frac{1}{16}\begin{pmatrix}
-(1+4\alpha^{2}) & -2i \\ -2i & 1+4\alpha^{2}
\end{pmatrix}.
\end{align}
\item[(d)] As $z \to 0$, 
\begin{equation}\label{local behaviour near 0 of P_Be}
\begin{array}{l l}
\displaystyle \Phi_{\mathrm{Be}}(z) = \left\{ \begin{array}{l l}
\begin{pmatrix}
\bigO(1) & \bigO(\log z) \\
\bigO(1) & \bigO(\log z) 
\end{pmatrix}, & |\arg z| < \frac{2\pi}{3}, \\
\begin{pmatrix}
\bigO(\log z) & \bigO(\log z) \\
\bigO(\log z) & \bigO(\log z) 
\end{pmatrix}, & \frac{2\pi}{3}< |\arg z| < \pi,
\end{array}  \right., & \displaystyle \mbox{ if }  \re \alpha = 0, \\[0.8cm]
\displaystyle \Phi_{\mathrm{Be}}(z) = \left\{ \begin{array}{l l}
\begin{pmatrix}
\bigO(1) & \bigO(1) \\
\bigO(1) & \bigO(1) 
\end{pmatrix}z^{\frac{\alpha}{2}\sigma_{3}}, & |\arg z | < \frac{2\pi}{3}, \\
\begin{pmatrix}
\bigO(z^{-\frac{\alpha}{2}}) & \bigO(z^{-\frac{\alpha}{2}}) \\
\bigO(z^{-\frac{\alpha}{2}}) & \bigO(z^{-\frac{\alpha}{2}}) 
\end{pmatrix}, & \frac{2\pi}{3}<|\arg z | < \pi,
\end{array} \right. , & \displaystyle \mbox{ if } \re \alpha > 0, \\[0.8cm]
\displaystyle \Phi_{\mathrm{Be}}(z) = \begin{pmatrix}
\bigO(z^{\frac{\alpha}{2}}) & \bigO(z^{\frac{\alpha}{2}}) \\
\bigO(z^{\frac{\alpha}{2}}) & \bigO(z^{\frac{\alpha}{2}}) 
\end{pmatrix}, & \displaystyle \mbox{ if } \re \alpha < 0.
\end{array}
\end{equation}
\end{itemize}
\begin{figure}
\centering
\begin{tikzpicture}
\draw (-3,0)--(0,0);
\draw (0,0)--(120:2);
\draw (0,0)--(-120:2);
\node at (0.1,-0.2) {$0$};
\draw[fill] (0,0) circle (0.05);
\draw[black,arrows={-Triangle[length=0.18cm,width=0.12cm]}]
(180:1.4) --  ++(0:0.001);
\draw[black,arrows={-Triangle[length=0.18cm,width=0.12cm]}]
(120:0.85) --  ++(-60:0.001);
\draw[black,arrows={-Triangle[length=0.18cm,width=0.12cm]}]
(-120:0.85) --  ++(60:0.001);
\end{tikzpicture}
\caption{\label{fig:Bessel}The jump contour $\Sigma_{\mathrm{Be}}$ for $\Phi_{\mathrm{Be}}$.}
\end{figure}
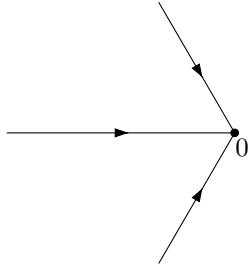
The unique solution to this RH problem is expressed in terms of Bessel functions. Since this explicit expression is unimportant for us, we will not write it down. The interested reader can find more information and background on this RH problem in e.g. \cite[Section 6]{KMcLVAV}.

\subsection{Confluent hypergeometric model RH problem}\label{subsection: model RH problem for Phi HG}
\begin{itemize}
\item[(a)] $\Phi_{\mathrm{HG}} : \mathbb{C} \setminus \Sigma_{\mathrm{HG}} \rightarrow \mathbb{C}^{2 \times 2}$ is analytic, with $\Sigma_{\mathrm{HG}} = \cup_{j=1}^{8}\Gamma_{j}$, and $\Gamma_{1},\ldots,\Gamma_{8}$ are shown in Figure \ref{Fig:HG}.
\item[(b)] $\Phi_{\mathrm{HG}}$ satisfies the jumps
\begin{equation}\label{jumps PHG3}
\Phi_{\mathrm{HG},+}(z) = \Phi_{\mathrm{HG},-}(z)J_{k}, \qquad z \in \Gamma_{k}, \; k = 1,...,8,
\end{equation}
where $J_{8} = \begin{pmatrix}
1 & 0 \\ e^{i\pi\alpha}e^{i\pi\beta} & 1
\end{pmatrix}$ and
\begin{align*}
& \hspace{-0.3cm} J_{1} = \begin{pmatrix}
0 & e^{-i\pi \beta} \\ -e^{i\pi\beta} & 0
\end{pmatrix}, \; J_{5} = \begin{pmatrix}
0 & e^{i\pi\beta} \\ -e^{-i\pi\beta} & 0
\end{pmatrix},\; J_{3} = J_{7} = \begin{pmatrix}
e^{\frac{i\pi\alpha}{2}} & 0 \\ 0 & e^{-\frac{i\pi\alpha}{2}}
\end{pmatrix}, \\
& \hspace{-0.3cm} J_{2} = \begin{pmatrix}
1 & 0 \\ e^{-i\pi\alpha}e^{i\pi\beta} & 1
\end{pmatrix}, \; J_{4} = \begin{pmatrix}
1 & 0 \\ e^{i\pi\alpha}e^{-i\pi\beta} & 1
\end{pmatrix}, \; J_{6} = \begin{pmatrix}
1 & 0 \\ e^{-i\pi\alpha}e^{-i\pi\beta} & 1
\end{pmatrix}.
\end{align*}
\item[(c)] As $z \to \infty$, $z \notin \Sigma_{\mathrm{HG}}$, we have
\begin{equation}\label{Asymptotics HG}
\Phi_{\mathrm{HG}}(z) = \left( I + \sum_{k=1}^{\infty} \frac{\Phi_{\mathrm{HG},k}}{z^{k}} \right) z^{-\beta\sigma_{3}}e^{-\frac{z}{2}\sigma_{3}}M^{-1}(z),
\end{equation}
where 
\begin{equation}\label{def of tau}
\Phi_{\mathrm{HG},1} = \Big(\beta^{2}-\frac{\alpha^{2}}{4}\Big) \begin{pmatrix}
-1 & \tau(\alpha,\beta) \\ - \tau(\alpha,-\beta) & 1
\end{pmatrix}, \qquad \tau(\alpha,\beta) = \frac{- \Gamma\left( \frac{\alpha}{2}-\beta \right)}{\Gamma\left( \frac{\alpha}{2}+\beta + 1 \right)},
\end{equation}
and
\begin{equation}
M(z) = \left\{ \begin{array}{l l}
\displaystyle e^{\frac{i\pi\alpha}{4} \sigma_{3}}e^{- i\pi\beta  \sigma_{3}}, & \displaystyle \frac{\pi}{2} < \arg z < \pi, \\
\displaystyle e^{-\frac{i\pi\alpha}{4} \sigma_{3}}e^{-i\pi\beta  \sigma_{3}}, & \displaystyle \pi < \arg z < \frac{3\pi}{2}, \\
e^{\frac{i\pi\alpha}{4}\sigma_{3}} \begin{pmatrix}
0 & 1 \\ -1 & 0
\end{pmatrix}, & \ds -\frac{\pi}{2} < \arg z < 0, \\
e^{-\frac{i\pi\alpha}{4}\sigma_{3}} \begin{pmatrix}
0 & 1 \\ -1 & 0
\end{pmatrix}, & \ds 0 < \arg z < \frac{\pi}{2}.
\end{array} \right.
\end{equation}
In \eqref{Asymptotics HG}, $z^{-\beta}$ has a cut along $i\mathbb{R}^{-}$, such that $z^{-\beta} \in \mathbb{R}$ as $z \in \mathbb{R}^{+}$.

As $z \to 0$,
\begin{align}
& \Phi_{\mathrm{HG}}(z) = \left\{ \begin{array}{l l}
\begin{pmatrix}
\bigO(1) & \bigO(\log z) \\
\bigO(1) & \bigO(\log z)
\end{pmatrix}, & \mbox{if } z \in II \cup III \cup VI \cup VII, \\
\begin{pmatrix}
\bigO(\log z) & \bigO(\log z) \\
\bigO(\log z) & \bigO(\log z)
\end{pmatrix}, & \mbox{if } z \in I\cup IV \cup V \cup VIII,
\end{array} \right. \nonumber \\
& \Phi_{\mathrm{HG}}(z) = \left\{ \begin{array}{l l}
\begin{pmatrix}
\bigO(z^{\frac{\alpha}{2}}) & \bigO(z^{-\frac{\alpha}{2}}) \\
\bigO(z^{\frac{\alpha}{2}}) & \bigO(z^{-\frac{\alpha}{2}})
\end{pmatrix}, & \mbox{if } z \in II \cup III \cup VI \cup VII, \\
\begin{pmatrix}
\bigO(z^{-\frac{\alpha}{2}}) & \bigO(z^{-\frac{\alpha}{2}}) \\
\bigO(z^{-\frac{\alpha}{2}}) & \bigO(z^{-\frac{\alpha}{2}})
\end{pmatrix}, & \mbox{if } z \in I\cup IV \cup V \cup VIII,
\end{array} \right. \label{lol 35}\\
& \Phi_{\mathrm{HG}}(z) = \begin{pmatrix}
\bigO(z^{\frac{\alpha}{2}}) & \bigO(z^{\frac{\alpha}{2}}) \\
\bigO(z^{\frac{\alpha}{2}}) & \bigO(z^{\frac{\alpha}{2}}) 
\end{pmatrix},\nonumber
\end{align}
where the first, second and third lines read for $\re \alpha = 0$, $\re \alpha > 0$ and $\re \alpha < 0$, respectively.
\end{itemize}
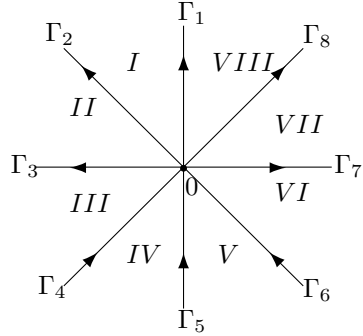
\begin{figure}[h]
    \begin{center}
    \setlength{\unitlength}{0.75truemm}
    \begin{picture}(100,55)(-5,10)
        \put(50,40){\line(-1,0){26}}
        \put(50,40){\line(1,0){26}}        
        \put(50,39.8){\thicklines\circle*{1.2}}
        \put(50,40){\line(-0.5,0.5){21}}
        \put(50,40){\line(-0.5,-0.5){21}}
        \put(50,40){\line(0.5,0.5){21}}
        \put(50,40){\line(0.5,-0.5){21}}
        \put(50,40){\line(0,1){25}}
        \put(50,40){\line(0,-1){25}}
        \put(50.3,35){$0$}
                \put(76.5,39){$\Gamma_7$}        
                \put(71,62){$\Gamma_8$}        
                \put(49,66){$\Gamma_1$}        
                \put(25.8,62.3){$\Gamma_2$}        
                \put(19.7,39){$\Gamma_3$}        
                \put(24.5,17.5){$\Gamma_4$}
                \put(49,11.5){$\Gamma_5$}
                \put(71,17){$\Gamma_6$}        
        \put(30,39.9){\thicklines\vector(-1,0){.0001}}
        \put(68,39.9){\thicklines\vector(1,0){.0001}}
        \put(32,58){\thicklines\vector(-0.5,0.5){.0001}}
        \put(35,25){\thicklines\vector(0.5,0.5){.0001}}
        \put(68,58){\thicklines\vector(0.5,0.5){.0001}}
        \put(65,25){\thicklines\vector(-0.5,0.5){.0001}}
        \put(50,60){\thicklines\vector(0,1){.0001}}
        \put(50,25){\thicklines\vector(0,1){.0001}}
        \put(40,57){$I$}
        \put(30,49){$II$}
        \put(30,32){$III$}
        \put(40,23){$IV$}
        \put(56,23){$V$}
        \put(66,34){$VI$}
        \put(66,46){$VII$}
        \put(55,57){$VIII$}
    \end{picture}
    \caption{\label{Fig:HG}The jump contour $\Sigma_{\mathrm{HG}}$ for $\Phi_{\mathrm{HG}}$. Each of the rays $\Gamma_{1},\ldots,\Gamma_{8}$ forms an angle with $(0,+\infty)$ which is a multiple of $\frac{\pi}{4}$.}
\end{center}
\end{figure}
The unique solution to this RH problem is expressed in terms of hypergeometric functions. Since we will not use the explicit expression of the solution, we will not write it down here. In the case where $\alpha=0$, this RH problem was first solved in \cite{ItsKrasovsky}. We refer the interested reader to \cite[Section 4.2]{DIK} and \cite[Section 2.6]{FouMarSou} for more details and background on this RH problem for general values of $\alpha$ and $\beta$.

\paragraph{Acknowledgment.} The author is grateful to Tom Claeys for useful remarks. This work is supported by the European Research Council, Grant Agreement No. 682537.

\end{document}